\documentclass[11pt, envcountsame]{llncs}

\usepackage{paralist}
\usepackage[utf8]{inputenc}
\usepackage[T1]{fontenc}
\usepackage{amsfonts}
\usepackage{latexsym, amsmath}
\usepackage{stmaryrd}
\usepackage{wasysym}
\usepackage{rotating}
\usepackage{pifont}
\usepackage{booktabs}
\usepackage{array}
\usepackage{multirow}
\usepackage{tikz}
\usepackage{color}
\usepackage{xcolor}
\usepackage{enumitem}
\usepackage{xspace}
\usepackage{thm-restate}
\usepackage{bm}

\usepackage{misc}

\newcommand{\Dln}{{\ensuremath{\textsf{DL}_\textsf{nr}}}\xspace}

\newcommand{\qedex}{\hfill$\dashv$}

\UseRawInputEncoding

\begin{document}

\title{Living Without Beth and Craig: \\ Definitions and Interpolants in Description and Modal Logics \\ with Nominals and Role Inclusions}


\author{Alessandro Artale$^1$ \and
  Jean Christoph Jung$^2$ \and Andrea
  Mazzullo$^1$ \and Ana
Ozaki$^3$ \and Frank Wolter$^4$}

\institute{$^1$ Free University of Bozen-Bolzano\qquad $^2$ TU Dortmund University \\ $^3$ University of Bergen \qquad $^4$ University of
Liverpool} 


%

\maketitle

\begin{abstract}
The Craig interpolation property (CIP) states that an interpolant for
an implication exists iff it is valid. The projective Beth
definability property (PBDP) states that an explicit definition exists
iff a formula stating implicit definability is valid.  Thus, the  CIP
and PBDP reduce potentially hard existence problems to entailment in
the underlying logic.  
Description (and modal) logics with nominals and/or role inclusions do
not enjoy the CIP nor the PBDP,  
but interpolants and explicit definitions have many
applications, in particular in concept learning, ontology
engineering, and ontology-based data management. 
In this article we show that, even without Beth and Craig, the existence of interpolants
and explicit definitions is decidable in description logics with
nominals and/or role inclusions such as $\mathcal{ALCO}$,
$\mathcal{ALCH}$ and $\mathcal{ALCHOI}$ and corresponding hybrid modal
logics. However, living without Beth and Craig makes these problems harder than entailment: the
existence problems become \TwoExpTime-complete in the presence of an
ontology or the universal
modality, and \coNExpTime-complete otherwise. We also analyze explicit
definition existence if all symbols (except the one that is defined)
are admitted in the definition. In this case the complexity depends on
whether one considers individual or concept names. Finally, we
consider the problem of computing interpolants and explicit
definitions if they exist and turn the complexity upper bound proof
into an algorithm computing them, at least for description logics with role inclusions.
\end{abstract}

{\small{
}}



\section{Introduction}

The \emph{Craig Interpolation Property} (CIP) for a logic $\Lmc$ states that an implication $\varphi \Rightarrow \psi$ is valid in $\Lmc$
iff there exists a formula $\chi$ in $\Lmc$ using only
the common symbols of $\varphi$ and $\psi$ such that
$\varphi\Rightarrow \chi$ and $\chi \Rightarrow \psi$ are both valid
in $\Lmc$. The intermediate formula $\chi$ is then called
an $\Lmc$-interpolant for $\varphi \Rightarrow \psi$~\cite{craig_1957}.
The CIP is generally regarded as one of the most
important and useful properties in formal logic~\cite{van2008many}, with numerous
applications ranging from formal verification~\cite{DBLP:conf/cav/McMillan03} and software specification~\cite{diaconescu1993logical} to theory combinations~\cite{DBLP:conf/cade/CimattiGS09,DBLP:conf/cade/GoelKT09,CalEtAl20,DBLP:journals/jar/CalvaneseGGMR22} and query reformulation and rewriting in databases~\cite{DBLP:series/synthesis/2011Toman,DBLP:series/synthesis/2016Benedikt}.
A particularly important consequence of the CIP is the \emph{projective Beth definability property} (PBDP), which states that 
a relation is implicitly definable
using a signature $\Sigma$ of symbols
iff it is
explicitly definable using $\Sigma$.
If $\Sigma$ is the set of all symbols distinct from that relation, then we speak of the (non-projective) \emph{Beth definability property} (BDP)~\cite{Bet56}.

In this paper, we investigate interpolants and explicit definitions in
description logics (DLs), and we also highlight consequences in modal
logic.  In DLs, one distinguishes essentially two forms
of interpolation, both of which are relevant and have their
applications.  Given an entailment $\Omc\models C\sqsubseteq D$, that
is, $C$ is subsumed by $D$ w.r.t. some background knowledge in the form
of a DL ontology $\Omc$, one might either be interested in an
interpolant between the concepts $C$ and $D$ or in an interpolant
between $\Omc$ and the concept inclusion (CI) $C\sqsubseteq D$. In the
first case, the interpolant is a concept, whereas in the second case,
the interpolant is an ontology.  We refer with \emph{CI-interpolation}
to the latter form and call the interpolant a \emph{CI-interpolant}.
The CIP for CI-interpolation has been shown to be the most important 
logical property that ensures the robust behaviour of ontology
modules and decompositions~\cite{DBLP:series/lncs/KonevLWW09,DBLP:conf/kr/KonevLPW10}.  

In this article, we mostly focus on interpolation (in the former sense
of an interpolating concept) and only derive some corollaries for
CI-interpolation.  Hence, unless stated otherwise, here and in what
follows we speak about interpolating concepts and the corresponding
CIP. For explicit definability, one asks for definitions of concepts,
possibly with respect to an ontology; these explicit definitions are
strongly related to interpolants and as stated above BDP and PBDP
follow from the CIP. In DLs, the BDP and PBDP have been used in ontology engineering to
extract explicit definitions of concepts and obtain equivalent acyclic
terminologies from ontologies~\cite{TenEtAl06,TenEtAl13}, they have
been investigated in ontology-based data management to equivalently
rewrite ontology-mediated
queries~\cite{DBLP:conf/ijcai/SeylanFB09,TomanWed20,FraEtAl13,FraKer19,TomWed21}, and
they have been proposed to support the construction of alignments
between ontologies~\cite{DBLP:conf/kr/Jimenez-RuizPST16}.  Interpolants have been used to study
$\textsc{P}$/$\textsc{NP}$ dichotomies in ontology-based query
answering~\cite{DBLP:journals/lmcs/LutzSW19}.

The CIP, PBDP, and BDP are so powerful   because potentially very hard 
existence questions are reduced to straightforward entailment
questions: an interpolant \emph{exists} iff an implication is valid
and an explicit definition \emph{exists} iff a straightforward formula
stating implicit definability is valid. The existence problems are
thus not harder than validity. Many basic DLs such as
$\mathcal{ALC}$, $\mathcal{ALCI}$, and $\mathcal{ALCIQ}$ enjoy the CIP
and PBDP~\cite{TenEtAl13}, and consequently the existence of an interpolant or an
explicit definition can be decided in \ExpTime simply because
entailment checking in these DLs is in \ExpTime (and without ontology
even in \PSpace). 
%
%
Unfortunately, the CIP and the PBDP fail to hold for some important
DLs. The most basic examples are the extension $\mathcal{ALCO}$ of
$\mathcal{ALC}$ with nominals (concepts of the form $\{a\}$ with $a$
an individual name), the extension $\mathcal{ALCH}$ of $\mathcal{ALC}$
with role inclusions (inclusions $r\sqsubseteq s$ between binary
relations/role names $r$ and $s$), and all standard DLs containing
either $\mathcal{ALCO}$ or
$\mathcal{ALCH}$~\cite{DBLP:series/lncs/KonevLWW09,TenEtAl13}.
It follows that for
these DLs the existence of interpolants and explicit definitions
cannot be reduced (directly) to entailment checking.

The aim of this article is to explore the consequences of the failure
of the CIP and PBDP for interpolant and explicit definition existence.
To this end, we investigate the complexity of deciding the existence
of interpolants and explicit definitions for the set $\Dln$ of DLs
containing $\mathcal{ALCO}$, $\mathcal{ALCH}$, and their extensions by
inverse roles and/or the universal role. 
We discuss next two
more applications of interpolants and explicit definitions for 
$\mathcal{ALCO}$ and its extensions. 

\paragraph{Data Separability and Concept Learning.} We show
that interpolants are essentially the same as concepts separating
positive and negative data examples in DL knowledge bases (KBs).
Recall that a DL KB is a pair $(\mathcal{O},\mathcal{D})$ with $\Omc$
a DL ontology and $\Dmc$ a set of data items of the form $A(a)$ and
$r(a,b)$ with $a,b$ individuals, $A$ a concept name, and $r$ a role
name. Let $\mathcal{O}$ be an ontology and $P$ and $N$ sets of
positively and negatively labelled pairs $(\Dmc,a)$ with $\mathcal{D}$
a set of data items and $a$ an individual in $\Dmc$. Then the aim of
supervised concept learning is to determine a concept $C$ in a
signature $\Sigma$ of relevant symbols such that $C$ separates $P$ and
$N$ in the sense that $(\Omc,\Dmc)\models C(a)$ for all positive
examples $(\Dmc,a)\in P$ and $(\mathcal{O}, \mathcal{D})\models \neg
C(a)$ for all negative examples $(\Dmc,a)\in N$.\footnote{This
  condition is called strong separation
  in~\cite{DBLP:conf/ijcai/FunkJLPW19,kr2021,DBLP:journals/ai/JungLPW22}. 
  A weaker version,
  called weak separation, only demands that $(\mathcal{O},
  \mathcal{D})\not\models C(a)$ for all negative examples $(\Dmc,a)\in
  N$. Concept learning systems have been developed for both the weak
  and the strong notion.} Concept learning has received significant
  interest over the past 15 years, where the focus has been on
  developing and analyzing refinement based algorithms for finding
  separating
  concepts~\cite{DBLP:conf/ilp/LehmannH09,DBLP:journals/ml/LehmannH10,Lisi15,DBLP:conf/aaai/SarkerH19,DBLP:conf/ilp/FanizzidE08,DBLP:conf/dlog/Lisi12,DBLP:journals/fgcs/RizzoFd20}.
  Prominent concept learning systems include the {\sc DL
  Learner}~\cite{DBLP:conf/www/BuhmannLWB18,DBLP:journals/ws/BuhmannLW16},
  {\sc DL-Foil}~\cite{DBLP:conf/ekaw/Fanizzi0dE18} and its extension
  {\sc DL-Focl} \cite{DBLP:conf/ekaw/0001FdE18},
  SPaCEL~\cite{DBLP:journals/jmlr/TranDGM17}, {\sc
  YinYang}~\cite{DBLP:journals/apin/IannonePF07}. The existence
  problem for separating concepts has been investigated
  in~\cite{DBLP:conf/ijcai/FunkJLPW19,kr2020,kr2021,DBLP:journals/ai/JungLPW22}. For DLs
  extending $\mathcal{ALCO}$, we establish a one-to-one correspondence
  between interpolants and separating concepts, modulo a rather
  straightforward polynomial time translation. Hence the existence of
  separating concepts reduces to the existence of interpolants and
  finding small such concepts or concepts of a certain syntactic
  shape, as is often useful in supervised learning, also reduces to
  the same task for interpolants. We emphasize that the presence of
  nominals in the DL is critical as they are required to encode the
  individuals used in $\mathcal{D}$ into concepts. 

\paragraph{Referring Expressions.} The computation of explicit
definitions of concept names has been explored in detail since at
least~\cite{TenEtAl06}, see also~\cite{ArtEtAl21}. Only recently, the
focus on defining concept names has been extended to defining
individual names, also called   referring expression generation in
computational linguistics and data
management~\cite{DBLP:journals/coling/KrahmerD12,DBLP:conf/inlg/ArecesKS08,DBLP:conf/kr/BorgidaTW16}.
In fact, it has been convincingly argued that very
often in applications the individual names used in ontologies or data
sets are insufficient ``to allow humans to figure out what real-world
objects they refer to''~\cite{BorEtAl17}. A natural way to address this problem is to
check for such an individual name $a$ whether there exists a concept
$C$ over a set of relevant symbols $\Sigma$ that provides an explicit
definition of $\{a\}$ and present such a concept $C$ to the human
user. Observe that one has to work with DLs extending $\mathcal{ALCO}$
to formulate this problem as an explicit definition existence problem.

\medskip To conclude, data separation, concept learning, and referring
expresssion generation are challenging research problems which
directly benefit from a better understanding of interpolant and
explicit definition existence in extensions of
$\mathcal{ALCO}$.   
%
%
%
%
We now discuss the main results of this article, formulated in an
informal way. Precise formulations are given later. Recall that $\Dln$
is the set of DLs $\mathcal{ALCO}$, $\mathcal{ALCH}$, and their
extensions with inverse roles and the universal role, and that we
assume the presence of a background DL ontology. Our first main
result is as follows.
\begin{theorem}\label{thm:A} 
  Let $\Lmc\in \Dln$. Then \Lmc-interpolant
  existence and $\Lmc$-definition existence are
  \TwoExpTime-complete.
\end{theorem}
Theorem~\ref{thm:A} confirms the suspicion that interpolant and
definition existence are much harder problems than entailment if one
has to live without Beth and Craig. On the positive side, these
problems are still decidable. Interestingly, for DLs in $\Dln$ with
nominals, the \TwoExpTime lower bound for definition existence already
holds if one asks for an explicit definition of an individual over the
signature containing all symbols distinct from that individual. In
contrast, the same problem for concept names is shown to be
\ExpTime-complete and thus not harder than entailment. Hence, in
contrast to concept name definitions, referring expression existence
does not become less complex in the non-projective case when all
symbols are allowed in definitions. 

We next consider the same problems if the background ontology is
empty, or, in the case of DLs in $\Dln$ without nominals, if the
ontology contains only role inclusions. Observe
that if the DL admits the universal role or both nominals and inverse
roles, then the ontology can be encoded as a concept using spy points~\cite{DBLP:conf/csl/ArecesBM99}, so nothing
changes compared to the case with ontologies covered in
Theorem~\ref{thm:A}. For the remaining cases we show the following.
\begin{theorem}\label{thm:B}
(1) If $\Lmc\in \{\mathcal{ALCO},\mathcal{ALCHO}\}$, then for the empty ontology and ontologies containing role inclusions only, \Lmc-interpolant existence and $\Lmc$-definition existence are
	both co\NExpTime-complete;
	
(2) If $\Lmc\in \{\mathcal{ALCH},\mathcal{ALCHI}\}$, then for ontologies containing role inclusions only \Lmc-interpolant existence and $\Lmc$-definition existence are
	    both co\NExpTime-complete.
\end{theorem}	   
It follows that without ontology and ontologies containing role inclusions only interpolant existence and explicit definition existence are still harder than entailment which is \PSpace-complete.

The proofs of Theorems~\ref{thm:A} and~\ref{thm:B} can be adapted to
also obtain results about CI-interpolation and interpolation in modal
logic.
Regarding the former
we show that for the DLs $\Lmc$ which extend $\mathcal{ALCO}$ with the
universal role or with the universal role and inverse roles the
problem of deciding the existence of a CI-interpolant for $\Omc
\models C \sqsubseteq D$ is \TwoExpTime-complete. It follows that
again failure of the CIP leads to an exponentially harder interpolant
existence problem than entailment. We conjecture that the same can be
proved for all DLs in $\Dln$, but leave a proof for future work. 

In modal logic, the CIP and PBDP have been investigated for many
years. In fact, the CIP and PBDP of DLs such as $\mathcal{ALC}$ and
$\mathcal{ALCI}$ follows rather directly from earlier results on the
CIP and PBDP in modal logic~\cite{GabMaks,DBLP:conf/amast/Marx98,DBLP:journals/jphil/Rautenberg83}.
Also the fact that nominals lead to failure of the CIP and PBDP, and how
this could be repaired by adding logical connectives, was first
analyzed in depth in the literature on hybrid modal logic, in
particular~\cite{DBLP:journals/jsyml/ArecesBM01,DBLP:journals/jsyml/Cate05}.
In our investigation of interpolant existence in modal logic, we first
consider basic modal logic with nominals and show that as a direct
consequence of Theorem~\ref{thm:B} the problem of deciding interpolant
existence is \coNExpTime-complete for the standard local consequence
relation. We also show using Theorem~\ref{thm:A} that if one adds the
universal modality, then interpolant existence becomes
\TwoExpTime-complete. In modal logic, nominals are often considered in
tandem with the $@$-operator, where $@_{a}\varphi$ states that formula
$\varphi$ holds at the world denoted by nominal $a$. The resulting
language is more expressive than modal logic with nominals and less
expressive than modal logic with nominals and the universal modality.
We show that for the modal logic with both nominals and the
$@$-operator interpolant existence is still \coNExpTime-complete. Our
complexity results also hold for the modal language with a single
modal operator (and the universal modality, if present).

While the focus in this article is on the decision problem, we also make
initial observations regarding the problem of actually computing interpolants or explicit
definitions if they exist. More specifically, for DLs in $\Dln$ that
do not admit nominals, we present a modification of the decision
procedure from the proof of Theorem~\ref{thm:A} that returns in double
exponential time the
DAG representation of an interpolant (if it exists). This corresponds
to interpolants of worst case triple exponential size which we
conjecture to be optimal. 

\paragraph{Overview of the Paper.} In the following
Section~\ref{sec:rw}, we discuss further related work. In
Sections~\ref{sec:pre} and~\ref{sec:background}, we introduce the
preliminaries on description logics and Craig interpolation and Beth
definability, respectively. In Section~\ref{sec:notions}, we provide
model-theoretic characterizations of the definition and interpolation
existence problems and formulate our main results in detail. The
subsequent four sections are devoted to the proofs of these main
results. In more detail, Section~\ref{sec:upperbound} provides the
upper bound proof for the case with ontologies and
Section~\ref{sec:lowerbound} provides the matching lower bounds.
Sections~\ref{sec:complexity} and~\ref{sec:lower-without-ontology}
cover the ontology-free case and the case of ontologies containing
only role inclusions.
In Section~\ref{sec:computation}, we investigate 
the problem of actually computing interpolants and explicit definitions in case
they exist, and in Section~\ref{sec:modal}
we draw the connections of our results on DLs to modal logic. Finally,
we conclude and point out directions for future work in
Section~\ref{sec:conclusion}.

  An appendix available as supplementary material provides a few proofs that were left out of the paper. Here we prove, in particular, our main result about the computational complexity of non-projective definition existence of concept names.

%
%
%

\section{Related Work}
\label{sec:rw}

This paper is an extended version
of~\cite{DBLP:conf/dlog/ArtaleJMOW20,DBLP:conf/aaai/ArtaleJMOW21,ourDL22}. We include detailed proofs and additionally discuss the link to concept learning,
interpolants between ontologies and concept inclusions, and
applications to modal logic.

Related work on Craig interpolation and the Beth definability property has been discussed already in the introduction. We therefore focus on work on deciding interpolant and explicit definition existence. These decision problems have only very recently been investigated. A notable exception is linear temporal logic, LTL, 
for which the CIP fails and for which decidability of interpolant existence has been shown both over finite linear orderings~\cite{henkell1,henkell2} and over the natural numbers~\cite{DBLP:journals/corr/PlaceZ14}. Note that these results are formulated as separability results for formal languages of finite and, respectively, infinite words: given two regular languages $R_{1}$ and $R_{2}$, does there exist a first-order definable language $L$ separating $R_{1}$ and $R_{2}$ in the sense that $R_{1} \subseteq L$ and $L \cap R_{2}=\emptyset$.
Neither LTL nor Craig interpolation are mentioned in~\cite{DBLP:journals/corr/PlaceZ14,henkell1,henkell2}. Using the fact that regular languages are projectively LTL definable and that LTL and first-order logic are equivalent over the natural numbers, it is, however, easy to see that interpolant existence is the same problem as separability of regular languages in first-order logic, modulo the representation of the inputs.
We note that this result is just one instance of an ongoing exploration of separation between languages in automata theory. The problem of deciding separation is interesting in this context because obtaining an algorithm for separation yields a far deeper understanding of the class under consideration than just membership~\cite{DBLP:journals/lmcs/Place18,DBLP:journals/tocl/PlaceZ20}. We conjecture that deciding interpolant existence could well play a similar role for understanding fragments of first-order logic. 

Indeed, interpolant existence has recently also been studied for the
guarded fragment (GF), the two-variable fragment (FO$^{2}$) of
FO~\cite{DBLP:conf/lics/JungW21}, for Horn description logics
extending $\mathcal{EL}$~\cite{DBLP:journals/corr/abs-2202-07186}, and for first-order modal logics~\cite{DBLP:journals/corr/abs-2303-04598}.
While GF is a good generalization of modal and description logic in
many respects, it neither enjoys the
CIP~\cite{DBLP:conf/lpar/HooglandMO99} nor the
PBDP~\cite{DBLP:conf/mfcs/BaranyBC13}. Failure of the CIP for FO$^{2}$
was shown using algebraic~\cite{comer1969,Pigozzi71} and
model-theoretic techniques~\cite{DBLP:journals/ndjfl/MarxA98}.
Using techniques that are similar to those introduced in this article
it is shown in \cite{DBLP:conf/lics/JungW21} that, in GF, explicit
definability and interpolant existence are both 3\ExpTime-complete in
general, and \TwoExpTime-complete if the arity of relation symbols is
bounded by a constant $c\geq 3$.  In FO$^{2}$, explicit definability
and interpolant existence are in coN\TwoExpTime and \TwoExpTime-hard~\cite{DBLP:conf/lics/JungW21}. 
Failure of the CIP and PBDP for first-order modal logics with constant domain is shown in~\cite{DBLP:journals/jsyml/Fine79,DBLP:journals/ndjfl/MarxA98}. Both properties also fail for their otherwise well-behaved one-variable and monodic fragments~\cite{gabbay2003many}. In~\cite{DBLP:journals/corr/abs-2303-04598}, the complexity of interpolant existence is investigated for first-order S5 with one variable (and some monodic fragments) and for first-order K with one variable.
For S5, explicit definability
and interpolant existence turn out to be in coN\TwoExpTime and \TwoExpTime-hard while for K only a non-elementary upper bound is shown.
These results confirm that for
many logics not enjoying the CIP and PBDP, interpolant and explicit
definition existence are harder than entailment.

It turns out that this is not always the case. It is shown in~\cite{DBLP:journals/corr/abs-2202-07186} that extensions of the description logic $\mathcal{EL}$ with any combination of the universal role, nominals, or inverse roles do not enjoy the CIP nor PBDP, but that interpolant existence and explicit definition existence still have the same complexity as entailment (in \PTime for those that do not admit inverse roles and \ExpTime-complete for those that admit inverse roles).
The proofs are rather different from those given in this paper, as they make use of the universal/canonical model that only exists for Horn logics.
 
We note that for logics that do not enjoy the CIP nor PBDP it is also of interest to look for ``small'' extensions that enjoy the CIP and PBDP and are decidable. For example, the guarded negation fragment of FO is a decidable extension of GF that enjoys the CIP and the PBDP~\cite{DBLP:journals/jsyml/BaranyBC18,DBLP:journals/tocl/BenediktCB16,DBLP:conf/lics/BenediktCB15,DBLP:journals/lmcs/BenediktBB19}. Also the two-variable fragment of GF is 
a decidable extension of $\mathcal{ALCH}$ enjoying both properties~\cite{DBLP:conf/lpar/HooglandMO99,DBLP:journals/sLogica/HooglandM02}. In both cases the complexity of entailment does not increase for the extension (\TwoExpTime-complete for the guarded negation fragment and \ExpTime-complete for the two-variable fragment of FO). 
On the other hand, under mild conditions there is no decidable extension of $\mathcal{ALCO}$ with the universal role nor of modal logic with nominals and the $@$-operator enjoying the CIP~\cite{DBLP:journals/jsyml/Cate05}.

While the problem of deciding interpolant and explicit definition existence for logics that 
neither 
enjoy the CIP nor the PBDP has only been considered rather recently, the problem of computing and deciding the existence of uniform interpolants for logics that do not enjoy the uniform interpolation property (UIP) has been investigated before.
Recall that uniform interpolants generalize Craig interpolants in the sense that
a uniform interpolant is an interpolant for a fixed $\varphi$ and all $\psi$ which are entailed by $\varphi$ and share with $\varphi$ a fixed set of symbols.  First-order logic enjoys the CIP but not the UIP. Propositional intuitionistic logic, local modal logic, and the modal mu-calculus are examples of expressive logics that enjoy the UIP~\cite{DBLP:journals/jsyml/Pitts92,visser1996uniform,DBLP:conf/aiml/DAgostinoH96}, see \cite{kowalski2019uniform,iemhoff2019uniform} for more recent results. In description logic, uniform interpolants of ontologies (extending what we call CI-interpolants in this article) are of particular importance but do not always exist for any standard description logic, including $\mathcal{ALC}$. The complexity of deciding their existence has been investigated in ~\cite{DBLP:conf/kr/LutzSW12,DBLP:conf/ijcai/LutzW11}, their size has been considered in~\cite{DBLP:journals/ai/NikitinaR14,DBLP:conf/aaai/KoopmannS15}, and various approaches to computing them have been developed and implemented~\cite{DBLP:conf/ijcai/KonevWW09,DBLP:conf/cade/KoopmannS14,DBLP:conf/aaai/KoopmannS15,DBLP:conf/ijcai/ZhaoS16}.

\section{Preliminaries}
\label{sec:pre}

We introduce the syntax and semantics of the relevant description
logics, see also~\cite{DL-Textbook}. 
Let \NC, \NR, and \NI be mutually disjoint and countably infinite sets
of \emph{concept}, \emph{role}, and \emph{individual names}. A \emph{role} is a role name $s$, or an \emph{inverse role}
$s^-$, with $s$ a role name and
$(s^-)^- = s$.
We use $u$ to denote the \emph{universal role}.
A \emph{nominal} takes
the form $\{a\}$, with $a$ an individual name. An
\emph{$\mathcal{ALCOI}^{u}$-concept} is defined according to the
syntax rule
\[
C, D ::= \top \mid A \mid \{a\} \mid \neg C \mid C \sqcap D \mid \exists r.C 
\]
where $A$ ranges over concept names, $a$ ranges over individual names, and
$r$ over roles
and the universal role.
We use $C \sqcup D$ as abbreviation for
$\neg (\neg C \sqcap \neg D)$, $C \rightarrow D$ for $\neg C \sqcup D$, $C \leftrightarrow D$ for $(C \rightarrow D) \sqcap (D \rightarrow C)$, and $\forall r.C$ for
$\neg \exists r.\neg C$.  We use several fragments of
$\mathcal{ALCOI}^{u}$, including $\mathcal{ALCOI}$, obtained by
dropping the universal role, $\mathcal{ALCO}^{u}$, obtained by
dropping inverse roles, $\mathcal{ALCO}$, obtained from $\mathcal{ALCO}^{u}$ by
dropping the universal role, and $\mathcal{ALC}$, obtained from $\mathcal{ALCO}$ by dropping nominals. 
If $\Lmc$ is any of the DLs above, then an \emph{\Lmc-concept inclusion ($\Lmc$-CI)}
takes the form $C \sqsubseteq D$ with $C$ and $D$ \Lmc-concepts. An \emph{$\Lmc$-ontology} is a finite set of \Lmc-CIs.
We also consider DLs with \emph{role inclusions (RIs)}, expressions
of the form $r\sqsubseteq s$, where $r$ and $s$ are roles. As usual,
the addition of RIs is indicated by adding the letter $\Hmc$ to the name of the DL, where inverse roles occur in RIs only if the DL admits inverse roles. Thus, for example, \emph{$\mathcal{ALCH}$-ontologies} are finite sets of $\mathcal{ALC}$-CIs and RIs not using inverse roles and \emph{$\mathcal{ALCHOI}^{u}$-ontologies} are finite sets of $\mathcal{ALCOI}^{u}$-CIs and RIs.
In what follows we use $\Dln$ to denote the set of DLs $\mathcal{ALCO}$, $\mathcal{ALCOI}$, $\mathcal{ALCH}$, $\mathcal{ALCHI}$, $\mathcal{ALCHO}$, $\mathcal{ALCHOI}$, and their extensions with the universal role. 
To simplify notation we do not drop the letter $\mathcal{H}$ when speaking about the concepts and CIs of a DL with RIs. Thus, for example, we sometimes use the expressions $\mathcal{ALCHO}$-concept and $\mathcal{ALCHO}$-CI to denote $\mathcal{ALCO}$-concepts and CIs, respectively. An \emph{RI-ontology} is an ontology containing RIs only.
 
The semantics is defined in terms of \emph{interpretations}
$\Imc=(\Delta^\Imc,\cdot^\Imc)$,
where $\Delta^{\mathcal{I}}$ is a non-empty set, called \emph{domain} of $\mathcal{I}$,
and $\cdot^{\mathcal{I}}$ is a function mapping
every $A \in \NC$ to a subset of $\Delta^{\Imc}$,
every $s\in\NR$ to a subset of $\Delta^{\Imc} \times \Delta^{\Imc}$,
the universal role $u$ to $\Delta^{\Imc} \times \Delta^{\Imc}$,
and every $a \in\NI$
to an element in $\Delta^{\Imc}$.
Given a role name $s \in \NR$, we set
$(s^{-})^{\mathcal{I}} = \{ (d, e) \in \Delta^{\Imc} \times \Delta^{\Imc} \mid (e, d) \in s^{\mathcal{I}} \}$.
Moreover, the \emph{extension $C^{\mathcal{I}}$ of an $\mathcal{L}$-concept $C$ in $\mathcal{I}$} is defined as follows, where $r$ ranges over roles and the universal role:
	\begin{align*}
		\top^{\mathcal{I}} & = \Delta^{\mathcal{I}}, \\
		\{ a \}^{\mathcal{I}} & = \{ a^{\mathcal{I}} \}, \\
		\neg C^{\mathcal{I}} & = \Delta^{\Imc} \setminus C^{\mathcal{I}}, \\ 
		(C \sqcap D)^{\mathcal{I}} & = C^{\mathcal{I}} \cap D^{\mathcal{I}}, \\
		(\exists r.C)^{\mathcal{I}} & = \{d \in \Delta^{\Imc} \mid \text{there exists} \ e \in C^{\mathcal{I}}: (d,e) \in r^{\mathcal{I}}\}.
	\end{align*}
An interpretation \Imc
\emph{satisfies} an $\Lmc$-CI $C \sqsubseteq D$ if $C^\Imc \subseteq D^\Imc$ and an RI $r\sqsubseteq s$ if $r^{\Imc}\subseteq s^{\Imc}$. We say that \Imc is a \emph{model} of
an ontology $\Omc$ if it satisfies all inclusions 
in it. We say that an inclusion $\alpha$ follows from an ontology $\Omc$, in symbols $\Omc\models \alpha$, if every
model of $\Omc$ satisfies $\alpha$. We write $\Omc\models C \equiv D$ if $\Omc \models C \sqsubseteq D$ and $\Omc\models D \sqsubseteq C$. We drop $\Omc$ if it is empty and write $\models C \sqsubseteq D$ for $\emptyset\models C \sqsubseteq D$.
A concept $C$ is \emph{satisfiable} w.r.t.\ an
ontology \Omc if there is a model \Imc of \Omc with
$C^\Imc \neq \emptyset$. We use a few well known complexity bounds for reasoning in DLs from $\Dln$.
The \emph{$\Lmc$-subsumption problem} is the problem to decide for any $\Lmc$-ontology $\Omc$ and $\Lmc$-CI $C \sqsubseteq D$ whether $\Omc\models C \sqsubseteq D$.
The \emph{ontology-free $\Lmc$-subsumption problem} and the
\emph{RI-ontology $\Lmc$-subsumption problem} are the sub-problems of
the $\Lmc$-subsumption problem in which the ontology is empty or an RI-ontology, respectively. For any $\Lmc\in \Dln$, the
$\Lmc$-subsumption problem is \ExpTime-complete~\cite{handbookDL,DBLP:conf/csl/ArecesBM99}. If $\Lmc$ admits the universal role or both inverse roles and nominals, then ontologies can be encoded in concepts and so 
ontology-free $\Lmc$-subsumption and RI-ontology $\Lmc$-subsumption are also \ExpTime-complete.
In the remaining cases, that is for $\mathcal{ALCO}$, $\mathcal{ALCH}$, $\mathcal{ALCHO}$, and $\mathcal{ALCHI}$, $\Lmc$-subsumption becomes \PSpace-complete~\cite{handbookDL,DBLP:conf/csl/ArecesBM99}.


A \emph{signature} $\Sigma$ is a set of concept, role, and individual
names, uniformly referred to as \emph{symbols}.  Following standard
practice we do not regard the universal role as a symbol but as a
logical connective. Thus, the universal role is not contained in any
signature.  We use $\text{sig}(X)$ to denote the set of symbols used
in any syntactic object $X$ such as a concept or an ontology.  An
\emph{$\Lmc(\Sigma)$-concept} is an $\Lmc$-concept $C$ with
$\text{sig}(C)\subseteq \Sigma$. A $\Sigma$-role $r$ is a
role with $\text{sig}(r)\subseteq \Sigma$. The \emph{size} of a
(finite) syntactic object $X$, denoted $||X||$, is the number of
symbols needed to represent it as a word. 

We next recall model-theoretic characterizations when
elements in interpretations are indistinguishable by concepts
formulated in one of the DLs \Lmc introduced above. A \emph{pointed
	interpretation} is a pair $\Imc,d$ with \Imc an interpretation
and $d\in \Delta^{\Imc}$. For pointed interpretations $\Imc,d$
and $\Jmc,e$ and a signature $\Sigma$, we write $\Imc,d \equiv_{\Lmc,\Sigma}\Jmc,e$ and say that
$\Imc,d$ and $\Jmc,e$ are \emph{$\Lmc(\Sigma)$-equivalent} if $d\in C^{\Imc}$
iff $e\in C^{\Jmc}$, for all $\Lmc(\Sigma)$-concepts $C$.
\begin{figure}[tb]
	\begin{center}
		\small
		\leavevmode
		\begin{tabular}{|@{\,}l@{\,}|l|}
			\hline 
			[AtomC] & for all $(d,e)\in S$:
			$d \in A^{\Imc}$ iff $e\in A^{\Jmc}$ \\ \hline
			[AtomI] & for all $(d,e)\in S$:
			$d = a^{\Imc}$ iff $e=a^{\Jmc}$ \\ \hline
			[Forth] & if $(d,e)\in S$ and $(d,d')\in r^{\Imc}$, 
			then \\
			& there is a $e'$ with $(e,e')\in r^{\Jmc}$ and
			$(d',e')\in S$.\\ \hline
			[Back] & if $(d,e)\in S$ and $(e,e')\in r^{\Jmc}$, 
			then \\
			& there is a $d'$ with $(d,d')\in r^{\Imc}$ and
			$(d',e')\in S$. \\ \hline
		\end{tabular} 
		\caption{Conditions on $S \subseteq \Delta^{\Imc}\times \Delta^{\Jmc}$.}
		\label{tab:conditions}
	\end{center}
	
	\vspace*{-5mm}
\end{figure}

As for the model-theoretic characterizations, we start with $\mathcal{ALC}$. Let $\Sigma$ be a signature. A
relation $S \subseteq \Delta^{\Imc}\times \Delta^{\Jmc}$ is an
\emph{$\mathcal{ALC}(\Sigma)$-bisimulation} if 
conditions [AtomC], [Forth], and [Back] from
Figure~\ref{tab:conditions} hold, where $A$ and $r$ range over all concept
and role names in $\Sigma$, respectively. We write $\Imc,d
\sim_{\mathcal{ALC},\Sigma}\Jmc,e$ and call $\Imc,d$ and $\Jmc,e$
\emph{$\mathcal{ALC}(\Sigma)$-bisimilar} if there exists an
$\mathcal{ALC}(\Sigma)$-bisimulation $S$ such that
$(d,e)\in S$. For $\mathcal{ALCO}$, we define $\sim_{\mathcal{ALCO},\Sigma}$ 
analogously, but now demand that, in
Figure~\ref{tab:conditions}, also condition [AtomI] holds for all individual names $a\in \Sigma$. For languages $\Lmc$ with inverse roles, we demand that, in Figure~\ref{tab:conditions}, $r$ additionally ranges over inverse roles. For languages $\Lmc$ with the universal role we extend
the respective conditions by demanding that the domain $\text{dom}(S)$ and range $\text{ran}(S)$ of $S$ contain $\Delta^{\Imc}$ and $\Delta^{\Jmc}$, respectively. If a DL $\Lmc$ has RIs, then we use $\Imc,d \sim_{\mathcal{L},\Sigma}\Jmc,e$ to state that $\Imc,d \sim_{\mathcal{L}',\Sigma}\Jmc,e$ for the fragment $\Lmc'$ of $\Lmc$ without RIs.

The next lemma summarizes the model-theoretic characterizations for
all relevant DLs \cite{TBoxpaper,goranko20075}. For the definition of $\omega$-saturated structures, we refer the reader to~\cite{modeltheory}. 
%
\begin{lemma}\label{lem:equivalence}
	Let $\mathcal{I},d$ and $\mathcal{J},e$ be pointed interpretations
	and $\omega$-saturated. Let $\mathcal{L}\in \Dln$ and $\Sigma$ a signature. 
	Then
	\[
	\Imc,d \equiv_{\mathcal{L},\Sigma} \Jmc,e \quad \text{iff} \quad
		\Imc,d \sim_{\mathcal{L},\Sigma}\Jmc,e.
	      \]
For the ``if''-direction, the $\omega$-saturatednesses condition can be dropped.	
\end{lemma}

\section{Craig Interpolation and Beth Definability}
\label{sec:background}

We introduce interpolants and the Craig interpolation property (CIP)
as well as implicit and explicit definitions and the (projective) Beth
definability property ((P)BDP). Recall from the introduction that
there are two forms of interpolants, one pertaining to concepts and
the other pertaining to concept inclusions. We start the discussion
here with the former one, and discuss CI-interpolants later. 
For concept interpolants, we establish a
close link between interpolants and separators of positive and
negative data examples, show that logics in $\Dln$ do not enjoy the
CIP nor PBDP, and determine which DLs in $\Dln$ enjoy the BDP.
 
Let $\Omc$ be an $\Lmc$-ontology, $C_{1},C_{2}$ be $\Lmc$-concepts, and let $\Sigma$ be a signature.
Then, an $\Lmc$-concept $D$ is an \emph{$\Lmc(\Sigma)$-interpolant for $C_{1}\sqsubseteq C_{2}$ under $\Omc$}, if $\text{sig}(D)\subseteq \Sigma$, $\Omc\models C_{1}\sqsubseteq D$, and $\Omc\models D \sqsubseteq C_{2}$. If $\Omc$ is empty, then we drop it and call an $\Lmc(\Sigma)$-interpolant for $C_{1}\sqsubseteq C_{2}$ under $\Omc$ simply an
\emph{$\Lmc(\Sigma)$-interpolant for $C_{1}\sqsubseteq C_{2}$}. Observe that $\Sigma$ is arbitrary in this definition, so it does not follow from $\Omc\models C_{1} \sqsubseteq C_{2}$ that an $\Lmc(\Sigma)$-interpolant for $C_{1}\sqsubseteq C_{2}$ under $\Omc$ exists. 
If $\Omc$ is empty, then we obtain the standard definition of the Craig interpolation property by demanding that for $\Sigma= \text{sig}(C_{1})\cap \text{sig}(C_{2})$ from $\models C_{1}\sqsubseteq C_{2}$ it follows that there exists an $\Lmc(\Sigma)$-interpolant for $C_{1}\sqsubseteq C_{2}$. The obvious generalization of this definition to non-empty ontologies, however, does not work. Consider, for instance, $\Omc=\{A_{1} \sqsubseteq A_{2}, A_{2} \sqsubseteq A_{3}\}$ and $A_{1}\sqsubseteq A_{3}$. Then for $\Sigma=\text{sig}(A_{1})\cap \text{sig}(A_{3})=\emptyset$, we have $\Omc\models A_{1}\sqsubseteq A_{3}$, but there does not exist any $\Lmc(\Sigma)$-interpolant for $A_{1}\sqsubseteq A_{3}$ under $\Omc$.  
In fact, to generalize the Craig interpolation property to non-empty
ontologies, one has to split the ontology $\Omc$ into two. Hence, we
adopt here the following definition of the Craig interpolation property in DLs from~\cite{TenEtAl13}. We set $\text{sig}(\Omc,C)= \text{sig}(\Omc)\cup \text{sig}(C)$, for any ontology $\Omc$ and concept $C$.

\begin{definition}
A DL $\Lmc$ \emph{enjoys the Craig interpolation property} (\emph{CIP}) if for any
$\Lmc$-ontologies $\Omc_{1},\Omc_{2}$ and $\Lmc$-concepts $C_{1},C_{2}$ such that $\Omc_{1}\cup \Omc_{2}\models C_{1}\sqsubseteq C_{2}$ there exists an $\Lmc(\Sigma)$-interpolant for 
$C_{1}\sqsubseteq C_{2}$ under $\Omc_{1}\cup \Omc_{2}$, where $\Sigma=\text{sig}(\Omc_{1},C_{1})\cap \text{sig}(\Omc_{2},C_{2})$.
If $\Omc_{1},\Omc_{2}$ range over $\Lmc$-ontologies containing RIs only or $\Omc_{1} = \Omc_{2} = \emptyset$, then we say that $\Lmc$ \emph{enjoys the CIP for RI-ontologies} and \emph{the CIP for the empty ontology}, respectively.

\end{definition}
Note that the CIP for the empty ontology coincides with the standard definition of the CIP mentioned before. 
It is shown in~\cite{TenEtAl13} that the DLs $\mathcal{ALC}$ and $\mathcal{ALCI}$ and their extensions with qualified number restrictions and the universal role all enjoy the CIP.
In contrast, no DL in $\Dln$ enjoys the CIP. This is implicitly proved in~\cite{TenEtAl13} and is shown in Theorem~\ref{thm:cipandbdp} below. The following illustrating example is folklore and shows that this holds even for the empty ontology for logics admitting nominals.
\begin{example}\label{ex:1}
	Consider $C_{1}=\{a\} \sqcap \exists r.\{a\}$ and $C_{2}=\{b\}
	\rightarrow \exists r.\{b\}$. Then $\models C_{1} \sqsubseteq
	C_{2}$ but there does not exist any
	$\mathcal{ALCHOI}^{u}(\{r\})$-interpolant for $C_{1} \sqsubseteq
	C_{2}$. Intuitively, no such interpolant $D$ exists as it
	would have to be true in exactly the elements $x$ with $r(x,x)$
	and no such $\mathcal{ALCHOI}^{u}(\{r\})$-concept exists. A formal proof is
	given in Example~\ref{ex:3} below. \qedex
\end{example}
%

As discussed in the introduction, the close link between separation of data examples and interpolants is one of our main motivations for studying interpolants for DLs with nominals. We next formalize this link. A \emph{database} $\Dmc$ is a finite set of assertions of the form $A(a)$ and $r(a,b)$ with $a,b$ individuals, $A$ a concept name, and $r$ a role name. By $\text{ind}(\Dmc)$ we denote the set of individual names in $\Dmc$. A \emph{knowledge base (KB)} is a pair $\Kmc= (\Omc,\Dmc)$ consisting of an ontology $\Omc$ and database $\Dmc$. An interpretation $\Imc$ is a \emph{model of $\Kmc$} if it is a model of $\Omc$, $a^{\Imc}\in A^{\Imc}$ for all $A(a)\in \Dmc$, and $(a^{\Imc},b^{\Imc})\in r^{\Imc}$ for all $r(a,b)\in \Dmc$. An assertion $C(a)$ with $C$ a concept and $a$ an individual \emph{follows from $\Kmc$}, in symbols $\Kmc\models C(a)$, if every model $\Imc$ of
$\Kmc$ satisfies $a^{\Imc}\in C^{\Imc}$. A \emph{labelled data set} consists of two sets, $P$ and $N$, of positive and negative examples each containing pairs $(\Dmc,a)$ with $a\in \text{ind}(\Dmc)$. 
Let $\Omc$ be an ontology. An \emph{$\Lmc(\Sigma)$-separator} for $\Omc,P,N$ is an $\Lmc(\Sigma)$-concept $C$ such that $(\Omc,\Dmc)\models C(a)$ for all $(\Dmc,a)\in P$ and $(\Omc,\Dmc)\models \neg C(a)$ for all $(\Dmc,a)\in N$.  
The following result establishes a one-to-one correspondence between interpolants and separators, modulo straightforward polynomial time reductions. Note that we do not require the frequent assumption that the database is uniform across the examples in the sense that $\Dmc=\Dmc'$ for all $(\Dmc,a),(\Dmc',a')\in P\cup N$~\cite{DBLP:journals/ai/JungLPW22}.
\begin{theorem}\label{thm:intsep} Let $\Lmc\in \Dln$ admit nominals. Then one can construct for any ontology $\Omc$, labelled data sets $P,N$, and signature $\Sigma$ with $\Sigma\cap \{a\mid (\Dmc,a)\in P\cup N\}=\emptyset$ in polynomial time $\Lmc$-ontologies $\Omc_{1}$, $\Omc_{2}$ and $\Lmc$-concepts $C_{1},C_{2}$ such that $\Sigma=\text{sig}(\Omc_{1},C_{1})\cap \text{sig}(\Omc_{2},C_{2})$ and the following conditions are equivalent for all $\Lmc$-concepts $C$:
	\begin{enumerate}
		\item $C$ is an $\Lmc(\Sigma)$-separator for $\Omc,P,N$;
		\item $C$ is an $\Lmc(\Sigma)$-interpolant for $C_{1}\sqsubseteq C_{2}$ under $\Omc_{1}\cup \Omc_{2}$.
	\end{enumerate}
	Conversely, asssume that $\Lmc$-ontologies $\Omc_{1}$, $\Omc_{2}$, and $\Lmc$-concepts  $C_{1}, C_{2}$ are given. Then one can construct in polynomial time an ontology $\Omc$
	and labelled data sets $P,N$ such that Conditions (1) and (2) are equivalent for $\Sigma=\text{sig}(\Omc_{1},C_{1})\cap \text{sig}(\Omc_{2},C_{2})$.
\end{theorem}
\begin{proof}
	Assume $\Omc, P, N$, and $\Sigma$ are given. Let $P=\{(\Dmc_{1},a_{1}),$ $\ldots,$ $(\Dmc_{n},a_{n})\}$ and $N=\{(\Dmc_{n+1},a_{n+1}),$ $\ldots,$ $(\Dmc_{n+m},a_{n+m})\}$. 
	If $\Lmc\in \Dln$ admits nominals and the universal role, then a pair $(\Dmc,a)$ can be represented using the $\Lmc$-concept $C_{\Dmc,a}=\{a\} \sqcap \exists u.C_{\Dmc}$, where $C_{\Dmc}$ is the conjunction of all $\{b\} \sqcap A$ with $A(b)\in \Dmc$ and $\{b\}\sqcap \exists r.\{c\}$ with $r(b,c)\in \Dmc$.
	Pick for any symbol $X$ not in $\Sigma$ a fresh copy $X'$.
	Let $\Omc_{1}=\Omc$ and obtain $\Omc_{2}$ from $\Omc$ by replacing all symbols not in $\Sigma$ by their copies. Let $C_{1}= C_{\Dmc_{1},a_{1}} \sqcup \cdots \sqcup C_{\Dmc_{n},a_{n}}$ and obtain $C_{2}$ from $\neg(C_{\Dmc_{n+1},a_{n+1}} \sqcup \cdots \sqcup C_{\Dmc_{n+m},a_{n+m}})$ by replacing all symbols not in $\Sigma$ by their copies.
	If $\Lmc$ admits the universal role, then $\Omc_{1},\Omc_{2}$ and $C_{1},C_{2}$ are as required. Otherwise replace the universal role in any $C_{\Dmc_{i},a_{i}}$ by fresh role names not in $\Sigma$. The resulting $C_{1},C_{2}$ are still as required.
	
    Conversely, assume that $\Omc_{1}$, $\Omc_{2}$ and $C_{1},C_{2}$ are given. Introduce fresh individual names $a,b$ and fresh concept names $A, B$ and let $P=\{(\{A(a)\},a)\}$, $N=\{(\{B(b)\},b)\}$ and  
	$
	\Omc = \Omc_{1} \cup \Omc_{2} \cup \{A \sqsubseteq C_{1}, B \sqsubseteq \neg C_{2}\}$.
	Then $\Omc,P,N$ is as required. 
\end{proof}
%
%
%

We next introduce the relevant definability notions. Let $\Omc$ be an
ontology and $C,C_{0}$ be concepts. Let $\Sigma$
be a signature. An $\Lmc(\Sigma)$-concept $D$ is an \emph{explicit
$\Lmc(\Sigma)$-definition of $C_{0}$ under $\Omc$ and $C$} if
$\Omc \models C \sqsubseteq (C_{0}\leftrightarrow D)$.
We call $C_{0}$ \emph{explicitly
$\Lmc(\Sigma)$-definable
under $\Omc$ and $C$} if there is an explicit $\Lmc(\Sigma)$-definition of $C_{0}$
under $\Omc$ and $C$. If $C=\top$ or $\Omc$ is empty, then we drop $\Omc$ and $C$, respectively.
For instance, an $\Lmc(\Sigma)$-concept $D$ is an explicit
$\Lmc(\Sigma)$-definition of $C_{0}$ under $\Omc$ if $\Omc \models C_{0} \equiv D$. 
The following example illustrates the link between explicit definitions of nominals and referring expressions discussed in the introduction and also indicates that often one can single out an individual from a set of individuals using an explicit definition without being able to provide an `absolute' explicit definition of that individual.
%

\begin{example}\label{exm:exbb}
Let $\mathsf{L}$ be an abbreviation for the $\mathcal{ALCO}$-concept
\[
\{ \mathcal{ALC} \} \sqcup  \{ \mathcal{ALCO} \} \sqcup \{ \mathcal{ALCO}^{u} \} \sqcup \{   \textnormal{ML} \} \sqcup \{ \textnormal{ML}_{n} \} \sqcup \{\textnormal{ML}_{n}^{u} \}, 
\]
where ML, ML$_{n}$, and ML$_{n}^{u}$ are modal logics introduced below in Section~\ref{sec:modal}. Let
$\mathcal{O}$ be the ontology consisting of the following CIs:
\begin{align*}
%
%
%
%
%
%
%
\{ \mathcal{ALC} \} \sqcup  \{ \mathcal{ALCO} \} \sqcup \{ \mathcal{ALCO}^{u} \} \sqsubseteq & \  \exists {\mathsf{hasOperator}}.{\mathsf{DLOperator}} \ \sqcap \\
& \ \lnot \exists {\mathsf{hasOperator}}.{\mathsf{MLOperator}}, \\
\{   \textnormal{ML} \} \sqcup \{ \textnormal{ML}_{n} \} \sqcup \{\textnormal{ML}_{n}^{u} \}  \sqsubseteq & \ \exists {\mathsf{hasOperator}}.{\mathsf{MLOperator}} \ \sqcap \\
& \ \lnot \exists {\mathsf{hasOperator}}.{\mathsf{DLOperator}}, \\
\{ \mathcal{ALC}\}\sqcup\{\textnormal{ML} \}  \sqsubseteq & \ {\mathsf{EnjoysCIP}}, \\
{\mathsf{EnjoysCIP}} \sqsubseteq & \ {\mathsf{EnjoysPBDP}}, \\
{\mathsf{EnjoysPBDP}} \sqsubseteq & \ {\mathsf{EnjoysBDP}}, \\
\{ \mathcal{ALCO}^{u} \}\sqcup\{ \textnormal{ML}_{n}^{u} \}  \sqsubseteq & \ {\mathsf{EnjoysBDP}} \sqcap \lnot {\mathsf{EnjoysPBDP}}, \\
\{ \mathcal{ALCO}\}\sqcup\{ \textnormal{ML}_{n} \}  \sqsubseteq & \ \lnot {\mathsf{EnjoysBDP}}.
\end{align*}
Then $
	\Omc \models \mathsf{L}  \sqsubseteq (\{ \mathcal{ALC} \}
	\leftrightarrow {\mathsf{EnjoysCIP}} \sqcap \exists
	{\mathsf{hasOperator}}.{\mathsf{DLOperator}})$. Hence, $\{ \mathcal{ALC} \}$ is
	explicitly 
	$\mathcal{ALCO}(\Sigma_{0})$-definable
	under $\Omc$ and $\mathsf{L}$, with $\Sigma_{0} = \{ {\mathsf{EnjoysCIP}}, {\mathsf{DLOperator}}, {\mathsf{hasOperator}} \}$.
However, $\{ \mathcal{ALC} \}$ is not
explicitly
$\mathcal{ALCO}(\Sigma)$-definable
under $\Omc$, for any signature $\Sigma$ with $\mathcal{ALC} \not \in \Sigma$, since
there does not exist an $\mathcal{ALCO}(\Sigma)$-concept $C$ such that $\Omc \models \{ \mathcal{ALC} \} \equiv C$.
\qedex
\end{example}

We next define when a concept is implicitly definable. For a signature
$\Sigma$, the \emph{$\Sigma$-reduct} $\Imc_{|\Sigma}$ of an
interpretation $\Imc$ coincides with $\Imc$ except that no
non-$\Sigma$ symbol is interpreted in $\Imc_{|\Sigma}$. A concept $C_{0}$
is called \emph{implicitly
 $\Sigma$-definable
under $\Omc$ and $C$} if
the $\Sigma$-reduct of any pointed model $\Imc,d$ with $\Imc$ a model of $\Omc$ and $d\in C^{\Imc}$ determines whether $d\in C_{0}^{\Imc}$. More formally, $C_{0}$
is implicitly
  $\Sigma$-definable
under $\Omc$ and $C$ if the following holds for all models $\Imc$ and $\Jmc$ of $\Omc$ and $d\in \Delta^{\Imc}=\Delta^{\Jmc}$: if $\Imc_{|\Sigma} = \Jmc_{|\Sigma}$ and $d\in C^{\Imc}$, then
$d\in C_{0}^{\Imc}$ iff $d\in C_{0}^{\Jmc}$.
If $C=\top$, then we drop $C$ and
	say that $C_{0}$
	is implicitly
	$\Sigma$-definable
	under $\Omc$. To illustrate, observe that in Example~\ref{exm:exbb}, $\{\ALC\}$ is not implicitly 
  $\Sigma$-definable, for any $\Sigma$
	such that $\mathcal{ALC} \not \in \Sigma$ under $\Omc$.
	Implicit definability can
be reformulated as a standard reasoning problem as follows: a concept
$C_{0}$ is
implicitly
 $\Sigma$-definable
under $\Omc$ and $C$ iff 
\begin{eqnarray}\label{eq:impl}
\Omc \cup \Omc_{\Sigma} \models C \sqcap C_{0} \sqsubseteq C_{\Sigma} \rightarrow {C_{0}}_{\Sigma}
\end{eqnarray}
where $\Omc_{\Sigma}$, $C_{\Sigma}$, and ${C_{0}}_{\Sigma}$ are obtained from $\Omc$, $C$ and, respectively, $C_{0}$, by
replacing every non-$\Sigma$ symbol uniformly by a fresh symbol. If a
concept is
explicitly
 $\Lmc(\Sigma)$-definable
under $\Omc$ and $C$, then
it is implicitly
 $\Sigma$-definable
under $\Omc$ and $C$, for any
language $\Lmc$. A logic enjoys the projective Beth definability
property if the converse implication holds as well.
\begin{definition}\label{def:pbdp}
  A DL $\Lmc$ \emph{enjoys the projective Beth definability 
  property} (\emph{PBDP}) if for
  any $\Lmc$-ontology $\Omc$, $\Lmc$-concepts $C$ and $C_{0}$, and signature
  $\Sigma\subseteq \text{sig}(\Omc,C)$ the following holds: if $C_{0}$ is
  implicitly
 $\Sigma$-definable
  under $\Omc$ and $C$, then $C_{0}$ is
  explicitly $\Lmc(\Sigma)$-definable under $\Omc$ and $C$. 
If $\Omc$ ranges over $\Lmc$-ontologies containing RIs only or $\Omc = \emptyset$, then we say that $\Lmc$ \emph{enjoys the PBDP for RI-ontologies} and \emph{the PBDP for the empty ontology}, respectively.
\end{definition}
The DLs $\mathcal{ALC}$, $\mathcal{ALCI}$, and their extensions with qualified number restrictions and the universal role all enjoy the PBDP~\cite{TenEtAl13}. The following example shows that, in contrast, $\mathcal{ALCH}$ does not.
\begin{example}\label{ex:2}
	Consider $\Omc=\{r \sqsubseteq r_{1}, r \sqsubseteq r_{2}\}$ and let 
	\[
	C = \big((\neg \exists r.\top \sqcap \exists r_{1}.A) \rightarrow \forall r_{2}.\neg A\big)
	\sqcap \big((\neg \exists r.\top \sqcap \exists r_{1}.\neg A) \rightarrow \forall r_{2}.A\big).
      \]
	Let $\Sigma=\{r_{1},r_{2}\}$ and $C_{0}=\exists r.\top$. Then the concept $D=\exists r_{1} \cap r_{2}.\top$ is
	an explicit definition of $C_{0}$ under $\Omc$ and $C$ in the extension of $\mathcal{ALCH}$ with role intersection (the semantics of $r_{1}\cap r_{2}$ is defined in the obvious way). Hence $C_{0}$ is implicitly
	 $\Sigma$-definable 
	under $\Omc$ and $C$. There does not exist an explicit $\mathcal{ALCH}(\Sigma)$-definition of $C_{0}$ under $\Omc$ and $C$, however. Intuitively, the reason is that role intersection cannot be expressed in $\mathcal{ALCH}$ (see Example~\ref{ex:4} below for a proof). 
	
	Note that an example without "background concept'' $C$ can be obtained by taking the ontology
	\begin{align*}
	\Omc' = & \{\ r\sqsubseteq r_{1},\ \ \ r\sqsubseteq r_{2}, \\
& 	\ \ \neg\exists r.\top \sqcap \exists r_{1}.A \sqsubseteq \forall r_{2}.\neg A, \ \ \ \neg \exists r.\top \sqcap \exists r_{1}.\neg A\sqsubseteq \forall r_{2}.A\ \}
	\end{align*}
	and asking for an explicit $\mathcal{ALCH}(\{r_{1},r_{2}\})$-definition of $\exists r.\top$ 
	under $\Omc'$. \qedex
\end{example}
It is known that the CIP and PBDP are tightly
linked~\cite{TenEtAl13}. We state the inclusion for logics in $\Dln$ only, but the proof shows that it holds under rather mild conditions.
\begin{lemma}\label{lem:CIPBDP} 
  If $\Lmc\in \Dln$ enjoys the CIP, then \Lmc enjoys the PBDP. 
\end{lemma} 
\begin{proof}
Assume that an $\Lmc$-concept $C_{0}$ is
implicitly $\Sigma$-definable
under an $\Lmc$-ontology $\Omc$ and $\Lmc$-concept $C$, for some signature
$\Sigma$.
Then \eqref{eq:impl} holds. Take an $\Lmc(\Sigma)$-interpolant $D$ for $C \sqcap C_{0} \sqsubseteq C_{\Sigma} \rightarrow {C_{0}}_{\Sigma}$ under $\Omc\cup \Omc_{\Sigma}$. 
Then $D$ is an explicit $\Lmc(\Sigma)$-definition of $C_{0}$ under $\Omc$ and $C$.
\end{proof}
 An important special case of explicit definability is the explicit definability of a concept name $A$ from $\text{sig}(\Omc,C)\setminus
\{A\}$ under an ontology $\Omc$ and concept $C$. For this case, we also consider the
following \emph{non-projective version} of the Beth definability
property.

\begin{definition}
A DL $\Lmc$ \emph{enjoys the Beth definability property} (\emph{BDP})
if for any $\Lmc$-ontology $\Omc$, concept $C$, and any concept name $A$ the
following holds: if $A$ is implicitly 
$(\text{sig}(\Omc,C)\setminus \{A\})$-definable

under $\Omc$ and $C$, then $A$ is
explicitly $\Lmc(\text{sig}(\Omc,C)\setminus \{A\})$-definable under
$\Omc$ and $C$.
If $\Omc$ ranges over $\Lmc$-ontologies containing RIs only or $\Omc = \emptyset$, then we say that $\Lmc$ \emph{enjoys the BDP for RI-ontologies} and \emph{the BDP for the empty ontology}, respectively.
\end{definition}

Clearly the PBDP entails the BDP, but we will see below that the converse direction does not always hold. In fact, the following theorem states that no DL in $\Dln$ enjoys the CIP or PBDP, but that quite a few DLs in $\Dln$ enjoy the BDP.
Moreover, all DLs in $\Dln$ enjoy the BDP for RI-ontologies and for the empty ontology.

	As mentioned before,
the theorem is mostly folklore and therefore proved in the appendix. 

\begin{restatable}{theorem}{thmcipandbdp}\label{thm:cipandbdp}
The following statements hold.
  \begin{enumerate}[label={(\arabic*)},leftmargin=6mm]
    \item No $\Lmc\in \Dln$ enjoys the CIP nor the PBDP.
    The CIP and PBDP also do not hold for RI-ontologies and, if $\Lmc$ admits nominals, the empty ontology.
%
    \item All $\Lmc\in \Dln\setminus
      \{\mathcal{ALCO},\mathcal{ALCHO}\}$ enjoy the BDP.
      $\mathcal{ALCO}$ and $\mathcal{ALCHO}$ do not enjoy the BDP.
     \item
     All $\Lmc\in \Dln$ enjoy the BDP for RI-ontologies and the BDP for the empty ontology.
  \end{enumerate}

\end{restatable}
We have seen that all $\Lmc\in \Dln\setminus
\{\mathcal{ALCO},\mathcal{ALCHO}\}$ enjoy the BDP. One might be tempted to conjecture that this holds as well if concept names are replaced by nominals; that is to say, a nominal $\{a\}$ that is implicitly definable using symbols distinct from $a$ is explicitly definable using symbols distinct from $a$. Rather surprisingly, the following example shows that this is not the case for any DL in $\Dln$ with nominals (for DLs without nominals this notion is clearly meaningless).
\begin{example}\label{exm:bdlnominals}
	Let $\Lmc\in \Dln$ admit nominals and assume that 
\begin{align*}
	\Omc = \{
	& \{a\} \sqsubseteq \exists r.\{a\}, \\
	& A \sqcap \neg \{a\} \sqsubseteq \forall r.(\neg \{a\} \rightarrow \neg A), \\
	& \neg A \sqcap \neg \{a\} \sqsubseteq \forall r.(\neg \{a\} \rightarrow A) \}.  
\end{align*}
Thus, $\mathcal{O}$ implies that $a$ is $r$-reflexive and 
that no element distinct from $a$ is $r$-reflexive. Let $\Sigma=\{r,A\}$.
Then $\{a\}$ is implicitly
$\Sigma$-definable
under $\Omc$ since 
we have the following explicit definition in first-order logic:
\[
\Omc \models \forall x ((x=a) \leftrightarrow r(x,x)), 
\]
but one can show that $\{a\}$ is not explicitly $\Lmc(\Sigma)$-definable under $\Omc$
for any $\Lmc\in \Dln$ with nominals.
%
%
		Indeed,
		the interpretation
		$\mathcal{I}$
		in Figure~\ref{fig:nopbdp}
		(where $a^{\Imc} = a$)
		is a model
		of $\Omc$ and the relation
		$S=\Delta^{\Imc}\times \Delta^{\Imc}$ is an
		$\mathcal{L}(\Sigma)$-bisimulation on $\Imc$. Thus,
		Lemma~\ref{lem:equivalence} implies
		%
		$
		\Imc, a^{\Imc} \equiv_{\mathcal{L},\Sigma} \Imc, d
		$
		%
		and there is no explicit
		$\mathcal{L}(\Sigma)$-definition
		for $\{a\}$ under $\mathcal{O}$, as any such definition would apply
		to $d$ as well.
\qedex

		
		
		\begin{figure}[th]
			\centering
			\begin{tikzpicture}
				\tikzset{
					dot/.style = {draw, fill=black, circle, inner sep=0pt, outer sep=0pt, minimum size=2pt}
				}
				
				\draw (-0.75,0.5) node[label=west:$\mathcal{I}$] {};
				\node (a) at (1.5,0) {};
				\draw (a) node[dot, label=south:$a$, label=north:$A$] {};
				\node (b) at (-0.5,0) {};
				\draw (b) node[dot, label=north:$A$,  label=south:$d$] {};
				
				\draw[->, >=stealth]  (a) edge [loop right] node[] {$r$} ();
				\path[->, >=stealth, bend right] (a) edge (b);
				\path[->, >=stealth, bend right] (b) edge (a);
				\node (x) at (0.5,0.5) {$r$};
				\node (y) at (0.5,-0.5) {$r$};
				
				%
				%

				
			\end{tikzpicture}
			\caption{
				Failure of the PBDP for $\Lmc\in \Dln$ with nominals.
			}
			\label{fig:nopbdp}
		\end{figure}

\end{example}

The CIP defined above is concerned with interpolating concepts. In the context of modular ontologies and forgetting there is also an interest in interpolating concept inclusions~\cite{DBLP:series/lncs/KonevLWW09}.
For simplicity, we only consider DLs without RIs.
Let $\Lmc$ not admit RIs and let $\Omc$ and $\Omc'$ be $\Lmc$-ontologies. We write $\Omc\models \Omc'$ if $\Omc\models \alpha$ for all $\alpha \in \Omc'$. Then an $\Lmc$-ontology $\Omc''$ is called an \emph{$\Lmc$-CI interpolant for $\Omc$ and $\Omc'$} if $\text{sig}(\Omc'') \subseteq \text{sig}(\Omc)\cap \text{sig}(\Omc')$, $\Omc\models \Omc''$,
		and $\Omc'' \models \Omc'$. If the particular language $\Lmc$ is clear from the context or not important we drop it and call $\Lmc$-CI interpolants simply CI-interpolants.
		
	\begin{definition}
		Let $\Lmc$ be a DL that does not admit RIs. Then $\Lmc$ has the
		\emph{CI-interpolation property}
		if for all $\Lmc$-ontologies $\Omc$ and $\Omc'$ such that $\Omc\models \Omc'$ there exists an $\Lmc$-CI-interpolant for $\Omc$ and $\Omc'$.
	\end{definition}
Observe that for any $\Lmc\in \Dln$ that does not admit RIs and $\Lmc$-ontology $\Omc$ one can construct in linear time an $\Lmc$-concept $D$ such that $\Omc$ and $\{\top \sqsubseteq D\}$ are equivalent in the sense that $\Omc \models \top\sqsubseteq D$ and $\{\top \sqsubseteq D \}\models \Omc$. 
Hence $\Lmc$ has the
CI-interpolation property
if for all $\Lmc$-ontologies $\Omc$ and $\Lmc$-CIs $C \sqsubseteq D$ such that $\Omc\models C \sqsubseteq D$ there exists an $\Lmc$-CI-interpolant for $\Omc$ and $\{C \sqsubseteq D\}$. It is known that $\mathcal{ALC}$ and its extensions with inverse roles, qualified number restrictions, and the universal role enjoy the CI-interpolation property~\cite{DBLP:series/lncs/KonevLWW09}. The following example shows that no DL in $\Dln$
that does not admit RIs enjoys the CI-interpolation property.
\begin{example}
We
 modify the ontology given in Example~\ref{exm:bdlnominals}. Let 
 \[
 \Omc'= \{ A \sqsubseteq \forall r.\neg A, \neg A \sqsubseteq \forall r. A \}.  
 \]
 We have that $\Omc'\models \{a\} \sqsubseteq \lnot \exists r.\{a\}$, but there does not exist
 an $\mathcal{ALCOI}^{u}$-CI-interpolant for $\Omc'$ and $\{a\}\sqsubseteq \lnot \exists r.\{a\}$, since one cannot express using an $\mathcal{ALCOI}^{u}(\{r\})$-CI that $\forall x \;\neg r(x,x)$.
Indeed, assume for a proof by contradiction that there exists an $\mathcal{ALCOI}^{u}$-CI-interpolant $\Omc''$ for $\Omc'$ and $\{ \{a\}\sqsubseteq \lnot \exists r.\{a\}\}$.
Consider the interpretations $\Imc_{1}, \Imc_{2}$ in Figure~\ref{fig:exCIinterpolant}, where $\Imc_{1} \models \Omc'$, $a^{\Imc_{2}} \in (\{a\} \sqcap \exists r.\{a\})^{\Imc_{2}}$, and $\Imc_{1}, d \sim_{\mathcal{ALCOI}^{u}, \{ r \}} \Imc_{2}, a^{\Imc_{2}}$.
Since $\Imc_{1} \models \Omc'$, we have that $\Imc_{1} \models \Omc''$, where $\Omc''$ can be assumed to be of the form $\{ \top \sqsubseteq D \}$, with $\text{sig}(D) \subseteq \{ r \}$.
Thus,
$d \in (\forall u.D)^{\Imc_{1}}$ and,
from $\Imc_{1}, d \sim_{\mathcal{ALCOI}^{u}, \{ r \}} \Imc_{2},
a^{\Imc_{2}}$, we obtain by Lemma~\ref{lem:equivalence} that
$a^{\Imc_{2}} \in (\forall u.D)^{\Imc_{2}}$.
This implies $\Imc_{2} \models \Omc''$, and hence $\Imc_{2} \models \{ a \} \sqsubseteq \lnot \exists r. \{ a \}$, contrary to the assumption that $a^{\Imc_{2}} \in (\{a\} \sqcap  \exists r.\{a\})^{\Imc_{2}}$.
\qedex
 

	\begin{figure}[th]
		\centering
		\begin{tikzpicture}
			\tikzset{
				dot/.style = {draw, fill=black, circle, inner sep=0pt, outer sep=1pt, minimum size=2pt}
			}

			\draw (-1,0.5) node[label=$\mathcal{I}_{1}$] (I1) {};
			
			\draw (0,-0.5) node[dot, label={[shift={(-0.25,-0.25)}]:$d$}, label={[shift={(0,-0.6)}]:$\lnot A$}, label={[shift={(-1,0.15)}]:$\mathcal{O}'$}] (d) {};
			
			\draw (0,0.5) node[dot, label={[shift={(-0.25,-0.25)}]:$e$}, label={[shift={(0,0)}]:$A$}] (e) {};
			
			
			\draw[->, >=stealth, bend right] (d) edge (e) node[label={[shift={(0.35,0.2)}]:$r$}] {};
			\draw[->, >=stealth, bend right] (e) edge (d) node[label={[shift={(-0.35,-0.8)}]:$r$}] {};

			\draw (5.6,0.5) node[label=$\mathcal{I}_{2}$] (I2) {};		
			
			\draw (5,0) node[dot, label={[shift={(0,-0.5)}]:$a$}, label={[shift={(1.35,-0.35)},align=center]:$\{ a \} \sqcap \exists r.\{ a \}$}] (a) {};
			
			
			\draw[->, >=stealth]  (a) edge [loop above] node[] {$r$} ();
			
			
			\path[black, dashed] (a) edge (e);
			\draw (2.5,0.15) node[label=$\sim_{\mathcal{ALCOI}^{u}, \Sigma}$] () {};

			\path[black, dashed] (a) edge (d);
			\draw (2.5,-0.9) node[label=$\sim_{\mathcal{ALCOI}^{u}, \Sigma}$] () {};
			
		\end{tikzpicture}
		
		\caption{Failure of the CI-interpolation property for $\Lmc \in \Dln$ without RIs.}
	\label{fig:exCIinterpolant}
	
\end{figure}

\end{example}

In this article we focus on interpolating concepts and not CIs. The main reasons are that the corresponding notion of an explicit definition of an ontology appears to be less useful than definitions of concepts and nominals and that while the CI-interpolation property is crucial for robust decompositions of ontologies and for robust forgetting~\cite{DBLP:series/lncs/KonevLWW09},
checking the existence of an interpolant or computing it for concrete ontologies and CIs appears to not have found any applications yet. Regarding the first point,
observe that CI-interpolants correspond to the following notion of an explicit CI-definition of an ontology. Let $\Sigma$ be a signature and $\Omc$ and $\Omc'$ ontologies. Then an $\Lmc(\Sigma)$-ontology $\Omc''$ is called an \emph{explicit $\Lmc(\Sigma)$-definition of $\Omc'$ under $\Omc$} if $\Omc\cup \Omc'\models \Omc''$ and $\Omc \cup \Omc''\models \Omc'$. In particular, if $\Omc$ is empty then one asks for an ontology using symbols in $\Sigma$ only that is equivalent to $\Omc$. While the existence of such ontologies is an interesting theoretical question that could well have applications in the future, investigating this problem is beyond the focus of this article. In what follows we only consider aspects of CI-interpolants that are closely related to concept interpolants, leaving their detailed investigation for future work.

\section{Main Results}
\label{sec:notions}

The failure of CIP and (P)BDP reported in Theorem~\ref{thm:cipandbdp}
imply that interpolant existence and projective and non-projective
definition existence cannot be directly polynomially reduced to subsumption checking.
This motivates studying the respective decision problems of
interpolant existence and projective and non-projective definition existence.
In this section we introduce the decision problems, formulate model-theoretic characterizations of the problems that play a fundamental role in our proofs, and we formulate the main results.

We start with interpolant existence for which we take the definition used in the 
formulation of the CIP.
\begin{definition}
  Let $\Lmc$ be a DL. Then \emph{$\Lmc$-interpolant existence} is the
  problem to decide for any $\Lmc$-ontologies $\Omc_{1},\Omc_{2}$ and $\Lmc$-concepts $C_{1},C_{2}$, whether there exists an
  $\Lmc(\Sigma)$-interpolant for $C_{1}\sqsubseteq C_{2}$ under $\Omc_{1}\cup \Omc_{2}$,
  where $\Sigma=\text{sig}(\Omc_{1},C_{1})\cap \text{sig}(\Omc_{2},C_{2})$. 
\end{definition}
In our proofs, we actually focus on a more general version of interpolant existence which has been discussed in the previous section and in which we do not split $\Omc$ into two ontologies and in which $\Sigma$ is arbitrary.
\begin{definition}
	Let $\Lmc$ be a DL. Then \emph{generalized $\Lmc$-interpolant existence} is the
	problem to decide for any $\Lmc$-ontology $\Omc$,
	$\Lmc$-concepts $C_{1},C_{2}$, and signature $\Sigma$ whether there exists an
	$\Lmc(\Sigma)$-interpolant for $C_{1}\sqsubseteq C_{2}$ under $\Omc$.
\end{definition}
We also consider (generalized) $\Lmc$-interpolant existence with empty
ontologies, called \emph{ontology-free (generalized) $\Lmc$-interpolant existence}, and with
RI-ontologies, called \emph{RI-ontology (generalized) $\Lmc$-interpolant existence},
both
defined in the obvious way. Observe that in the ontology-free case
there is no difference between generalized interpolant existence and
interpolant existence. In fact, also with ontologies generalized
interpolant existence and interpolant existence are interreducible.
\begin{lemma}\label{lem:mutual}
  Let $\Lmc\in \Dln$. There are mutual polynomial time reductions
  between generalized $\Lmc$-interpolant existence and
  $\Lmc$-interpolant existence.
\end{lemma}
\begin{proof} 
  The reduction from \Lmc-interpolant existence to generalized \Lmc-interpolant
  existence is trivial: for input $\Omc_1,\Omc_2,C_1,C_2$ to
  \Lmc-interpolant existence, set $\Omc=\Omc_1\cup\Omc_2$ and
  $\Sigma=\text{sig}(\Omc_1,C_1)\cap \text{sig}(\Omc_2,C_2)$. 

  For the converse reduction from generalized interpolant existence to
  interpolant existence, assume that an $\Lmc$-ontology $\Omc$, $\Lmc$-concepts
  $C_{1},C_{2}$, and $\Sigma$ are given. Then there exists an
  $\Lmc(\Sigma)$-interpolant for $C_{1}\sqsubseteq C_{2}$ under $\Omc$
  iff there exists an $\Lmc(\Sigma)$-interpolant for $C_{1}\sqsubseteq
  {C_{2}}_{\Sigma}$ under $\Omc\cup \Omc_{\Sigma}$, where
  $\Omc_{\Sigma}$ and ${C_{2}}_{\Sigma}$ are obtained from $\Omc$ and
  $C_{2}$ by replacing every non-$\Sigma$ symbol uniformly by a fresh
  symbol. The latter is an instance of $\Lmc$-interpolant
  existence.  
\end{proof} 
Note that the reduction above works for all
standard DLs including $\mathcal{ALC}$.  Recall that interpolant
existence reduces to checking $\Omc_{1}\cup \Omc_{2}\models
C_{1}\sqsubseteq C_{2}$ for logics with the CIP. Hence, for DLs which
enjoy the CIP such as $\mathcal{ALC}$, interpolant existence and
generalized interpolant existence are \ExpTime-complete and
ontology-free interpolant existence and generalized ontology-free
interpolant existence are \PSpace-complete. 

We next introduce the relevant definition existence problems.
\begin{definition}
  Let $\Lmc$ be a DL. \emph{Projective $\Lmc$-definition existence} is
  the problem to decide for any $\Lmc$-ontology $\Omc$, $\Lmc$-concepts
  $C$ and $C_{0}$, and signature $\Sigma$, whether
  there exists an explicit $\Lmc(\Sigma)$-definition of $C_{0}$ under
  $\Omc$ and $C$. 
  
  \emph{$($Non-projective$)$ $\Lmc$-definition existence of concept names (nominals)} is the
  sub-problem where $C_{0}$ ranges only over concept names $A$ (nominals $\{a\}$) and
  $\Sigma=\text{sig}(\Omc,C)\setminus\{A\}$ (and $\Sigma=\text{sig}(\Omc,C)\setminus\{a\}$, respectively).
\end{definition}
We also consider
the (projective) $\Lmc$-definition existence problems with empty ontologies, called \emph{ontology-free (projective) $\Lmc$-definition existence}, and with RI-ontologies, called \emph{RI-ontology (projective) $\Lmc$-definition existence}, both defined in the obvious way. Similar to the case of  interpolant existence, definition existence reduces to checking implicit
definability for logics with the PBDP. 
We provide model-theoretic characterizations for the non-existence of generalized interpolants and explicit definitions in terms of bisimulations.
\begin{definition}[Joint consistency] Let $\mathcal{L}\in \Dln$.
	Let $\Omc$ be an $\Lmc$-ontology, $C_{1},C_{2}$ be
	$\Lmc$-concepts, and $\Sigma$ a signature. Then
	$C_{1}$ and $C_{2}$ are called \emph{jointly
		consistent under \Omc modulo $\Lmc(\Sigma)$-bisimulations} if there exist
	pointed interpretations $\Imc_{1},d_{1}$ and $\Imc_{2},d_{2}$ such that
	$\Imc_{i}$ is a model of $\Omc$, $d_{i}\in C_{i}^{\Imc_{i}}$,
	for $i=1,2$, and $\Imc_{1},d_{1}\sim_{\Lmc,\Sigma} \Imc_{2},d_{2}$.
\end{definition}
The associated decision problem, \emph{joint consistency modulo
	$\Lmc$-bisimu-lations}, is defined in the expected way. The following
result characterizes the existence of interpolants using joint
consistency modulo $\Lmc(\Sigma)$-bisimulations. The proof uses
Lemma~\ref{lem:equivalence}.
\begin{restatable}{theorem}{thminterpolants}\label{thm:critint}
	Let $\mathcal{L}\in \Dln$.  Let $\Omc$ be an
	$\Lmc$-ontology, $C_{1},C_{2}$ be $\Lmc$-concepts, and
	$\Sigma$ a signature. Then the
	following conditions are equivalent:
	\begin{enumerate}
		
		\item there is no $\Lmc(\Sigma)$-interpolant for
		$C_{1}\sqsubseteq C_{2}$ under $\Omc$;
		
		\item $C_{1}$ and $\neg
		C_{2}$ are jointly consistent under \Omc modulo $\Lmc(\Sigma)$-bisimula-tions.
		
\end{enumerate} \end{restatable}
\begin{proof} \ The proof is standard and we refer the reader to~\cite{goranko20075} for similar proofs. We only provide a sketch.
	
	``1 $\Rightarrow$ 2''. Assume there is no $\Lmc(\Sigma)$-interpolant 
	for
	$C_{1}\sqsubseteq C_{2}$ under $\Omc$. Let
	\[
	\Gamma = \{ D \mid \Omc \models C_{1} \sqsubseteq D, D\in \Lmc(\Sigma)\}.
      \]
	Then $\Omc\not\models D\sqsubseteq C_{2}$, for any $D\in
	\Gamma$. As $\Gamma$ is closed under conjunction and by
	compactness (recall that $\mathcal{ALCHOI}^{u}$ is a fragment
	of first-order logic), there exists a model $\Jmc$ of $\Omc$
	and an element $d\in \Delta^{\Jmc}$ such that $d\in D^{\Jmc}$
	for all $D\in \Gamma$ but $d\not\in C_{2}^{\Jmc}$. Consider
	the set $t_{\Jmc}(d)= \{ D \in \Lmc(\Sigma) \mid d \in
	  D^{\Jmc}\}$. Then, using again compactness, there exists a
	  model $\Imc$ of $\Omc$ and an element $e\in \Delta^{\Imc}$ such that $e\in C_{1}^{\Imc}$ and $e\in D^{\Imc}$ for all $D\in t_{\Jmc}(d)$.
	Thus $\Imc,e\equiv_{\Lmc,\Sigma} \Jmc,d$. For every interpretation $\Imc$ there exists an $\omega$-saturated elementary extension $\Imc'$ of $\Imc$~\cite{modeltheory}. Thus, it follows from the fact that $\mathcal{ALCHOI}^{u}$ is a fragment of first-order logic that we may assume that both $\Imc$ and $\Jmc$
	are $\omega$-saturated. By Lemma~\ref{lem:equivalence}, $\Imc,e \sim_{\Lmc,\Sigma} \Jmc,d$.  
	
	\medskip
	
	``2 $\Rightarrow$ 1''. Assume an $\Lmc(\Sigma)$-interpolant $D$ for $C_{1}\sqsubseteq C_{2}$ under $\Omc$ exists.
	Assume that Condition~2 holds, that is, there are models $\Imc_{1}$ and $\Imc_{2}$ of $\Omc$ and $d_{i}\in \Delta^{\Imc_{i}}$ for $i=1,2$ such that
	$d_{1}\in C_{1}^{\Imc_{1}}$ and $d_{2}\not\in C_{2}^{\Imc_{2}}$ and
	$\Imc_{1},d_{1} \sim_{\Lmc,\Sigma} \Imc_{2},d_{2}$. Then, by Lemma~\ref{lem:equivalence},
	$\Imc_{1},d_{1} \equiv_{\Lmc,\Sigma} \Imc_{2},d_{2}$. But
	then from $d_{1}\in C^{\Imc_{1}}$ we obtain $d_{1}\in D^{\Imc_{1}}$ and
	so $d_{2}\in D^{\Imc_{2}}$ which implies $d_{2}\in C_{2}^{\Imc_{2}}$,
	a contradiction.
\end{proof}

\begin{example}\label{ex:3}
  Consider again $C_{1}=\{a\} \sqcap \exists r.\{a\}$ and $C_{2}=\{b\}
  \rightarrow \exists r.\{b\}$ from Example~\ref{ex:1} and set
  $\Sigma=\{r\}$. The interpretations $\Imc_{1}, \Imc_{2}$ depicted in
  Figure~\ref{fig:ex3} (where we set $a^{\Imc_i} = a$ and $b^{\Imc_i}
  = b$, for $i=1,2$) show that $C_{1}$ and $\neg C_{2}$ are jointly
  consistent modulo $\mathcal{ALCO}(\Sigma)$-bisimulations.  By
  extending the bisimulation in Figure~\ref{fig:ex3} to a relation $S$
  such that $(b^{\Imc_{1}},a^{\Imc_{2}}) \in S$ (so that the domain
  and range of $S$ contain $\Delta^{\Imc_{1}}$ and
  $\Delta^{\Imc_{2}}$, respectively), one can show that  $C_{1}$ and
  $\neg C_{2}$ are jointly consistent modulo
  $\mathcal{ALCO}^{u}(\Sigma)$-bisimulations.
  Moreover, by introducing an element $e$ in $\Imc_{2}$ so that $(e,
  b^{\Imc_{2}}) \in r^{\Imc_{2}}$ and $(e, e) \in r^{\Imc_{2}}$, and
  further extending $S$ by adding $(a^{\Imc_{1}}, e) \in S$, it can be
  seen that $C_{1}$ and $\neg C_{2}$ are jointly consistent modulo
  $\mathcal{ALCOI}^{u}(\Sigma)$-bisimulations (and hence
  $\mathcal{ALCHOI}^{u}(\Sigma)$-bisimulations).  \qedex
\end{example}
	
\begin{figure}[th]
  \centering
  \begin{tikzpicture}
    \tikzset{
      dot/.style = {draw, fill=black, circle, inner sep=0pt, outer sep=1pt, minimum size=2pt}
    }

    \draw (-2,0.5) node[label=$\Imc_{1}$] (I1) {};

    \draw (0,0) node[dot, label={[shift={(-0.25,-0.25)},align=center]:$a$}, label={[shift={(0,-0.75)},align=center]:$C_{1}$}] (a) {};

    \draw (-1,0) node[dot, label={[shift={(-0.25,-0.25)},align=center]:$b$}] (b1) {};


    \draw[->, >=stealth]  (a) edge [loop above] node[] {$r$} ();

    \draw (7,0.5) node[label=$\Imc_{2}$] (I2) {};

    \draw (6,-0.5) node[dot, label={[shift={(0.25,-0.25)},align=center]:$a$}] (a2) {};

    \draw (5,-0.5) node[dot, label={[shift={(0.25,-0.25)},align=center]:$b$}, label={[shift={(0,-0.75)},align=center]:$\lnot C_{2}$}] (b2) {};

    \draw (5,0.5) node[dot, label=east:$d$] (d) {};


    \draw[->, >=stealth] (b2) -- (d) node[midway, left] {$r$};

    \draw[->, >=stealth]  (d) edge [loop above] node[] {$r$} ();


    \path[black, dashed, bend left] (a) edge (b2);
    \draw (2.5,1) node[label=$\sim_{\mathcal{ALCO}, \Sigma}$] () {};

    \path[black, dashed, bend left] (a) edge (d);
    \draw (2.5,-0.25) node[label=$\sim_{\mathcal{ALCO}, \Sigma}$] () {};

  \end{tikzpicture}

  \caption{Interpretations $\Imc_{1}$ and $\Imc_{2}$ illustrating Example~\ref{ex:3}.}
  \label{fig:ex3}

\end{figure}

The following characterization of the existence of explicit
definitions can be proved similarly to Theorem~\ref{thm:critint}.
\begin{theorem}\label{thm:critdef}
  Let $\mathcal{L}\in \Dln$.  Let $\Omc$ be an $\Lmc$-ontology,
  $C$ and $C_{0}$ $\Lmc$-concepts, and $\Sigma\subseteq \text{sig}(\Omc,C)$ a
  signature. Then the following conditions are equivalent:
  \begin{enumerate} 

    \item there is no explicit
      $\Lmc(\Sigma)$-definition of $C_{0}$ under $\Omc$ and $C$; 

    \item $C\sqcap C_{0}$ and $C \sqcap \neg C_{0}$ are jointly consistent under \Omc modulo
      $\Lmc(\Sigma)$-bisimulations.  

  \end{enumerate}	
\end{theorem}
\begin{example}\label{ex:4}
	Consider $\Omc$, $C$, and $\Sigma$ from Example~\ref{ex:2}. The
	interpretations $\Imc_{1}, \Imc_{2}$ depicted in Figure~\ref{fig:ex4}
	show that $C \sqcap\exists r.\top$ and $C \sqcap \neg\exists r.\top$
	are jointly consistent under \Omc modulo
	$\mathcal{ALCH}(\Sigma)$-bisimulations.
	Note that the $\mathcal{ALCH}(\Sigma)$-bisimulation in
	Figure~\ref{fig:ex4} is also an
	$\mathcal{ALCH}^{u}(\Sigma)$-bisimula-tion, but it is not an
	$\mathcal{ALCHI}(\Sigma)$-bisimulation, since $e_{1}$ has both an
	$r_{1}$- and an $r_{2}$-predecessor, whereas $e_{2}$ and $e'_{2}$ lack
	an $r_{2}$- and an $r_{1}$-predecessor, respectively.  To repair this,
	we replace $\Imc_{2}$ with an interpretation $\Jmc$ that is obtained
	by taking the union of $\Imc_{2}$ with a copy $\overline{\Imc_{1}}$ of
	$\Imc_{1}$, and further adding $(\overline{d}_{1}, e_{2}) \in
	r_{2}^{\Jmc}$ and $(\overline{d}_{1}, e'_{2}) \in r_{1}^{\Jmc}$ (where
	$\overline{d}_{1}$ is the copy of $d_{1}$ in $\Jmc$).  Then, we extend
	the $\mathcal{ALCH}(\Sigma)$-bisimulation in Figure~\ref{fig:ex4} to a
	relation $S$ that also connects the elements of $\Imc_{1}$ with the
	respective copies in $\Jmc$.  It can be seen that $\Jmc$ is a model of
	$\Omc$, $d_{2} \in (C \sqcap \lnot \exists r.\top)^{\Jmc}$, and
	$(d_{1}, d_{2}) \in S$, where $S$ is an
	$\mathcal{ALCHI}^{u}(\Sigma)$-bisimulation.  \qedex
\end{example}

\begin{figure}[th]
	\centering
	\begin{tikzpicture}
		\tikzset{
			dot/.style = {draw, fill=black, circle, inner sep=0pt, outer sep=1pt, minimum size=2pt}
		}
		\draw (-2,0) node[label=$\Imc_{1}$] (I1) {};
		
		\draw (0,-0.5) node[dot, label=west:$d_{1}$, label={[shift={(-0.25,-0.85)},align=center]:$C \sqcap\exists r.\top$}] (d1) {};
		
		\draw (0,0.5) node[dot, label=west:$e_{1}$] (e1) {};
		
		
		\draw[->, >=stealth] (d1) -- (e1) node[midway, left] {$r, r_{1}, r_{2}$};

		\draw (7,0) node[label=$\Imc_{2}$] (I2) {};
		
		\draw (5,-0.5) node[dot, label=east:$d_{2}$, label={[shift={(0.25,-0.85)},align=center]:$C \sqcap \lnot \exists r.\top$}] (d2) {};
		
		\draw (4.5,0.5) node[dot, label=east:$e_{2}$] (e2) {};
		
		\draw (5.5,0.5) node[dot, label=east:$e'_{2}$, label=north:$A$] (e'2) {};

		
		\draw[->, >=stealth] (d2) -- (e2) node[midway, left] {$r_{1}$};
		
		\draw[->, >=stealth] (d2) -- (e'2) node[midway, right] {$r_{2}$};
		
		
		\path[black, dashed, bend right] (d1) edge (d2);
		\draw (2.5,-1.35) node[label=$\sim_{\mathcal{ALCH}, \Sigma}$] () {};

		\path[black, dashed, bend left] (e1) edge (e'2);
		\draw (2.5,0.65) node[label=$\sim_{\mathcal{ALCH}, \Sigma}$] () {};

		\path[black, dashed, bend right] (e1) edge (e2);
		\draw (2.5,-0.25) node[label=$\sim_{\mathcal{ALCH}, \Sigma}$] () {};

	\end{tikzpicture}

	\caption{Interpretations $\Imc_{1}$ and $\Imc_{2}$ illustrating Example~\ref{ex:4}.}
\label{fig:ex4}

\end{figure}
	%
	%
	%
%
%
%

 Interpolant existence and explicit definition existence are closely linked. We use Theorems~\ref{thm:critint} and~\ref{thm:critdef} to show 
the following reductions.
	
\begin{lemma}\label{lem:reduce1}
	Let $\mathcal{L}\in \Dln$, $\Omc$ be an $\Lmc$-ontology,
	$C, C_{0}, C_{1}$, and $C_{2}$ be $\Lmc$-concepts, and $\Sigma$ a signature.
	Then the following conditions are equivalent:
\begin{enumerate}
	\item there is an explicit $\mathcal{L}(\Sigma)$-definition of $C_{0}$ under $\Omc$ and $C$;
	\item there is an $\Lmc(\Sigma)$-interpolant for
	$C\sqcap C_{0} \sqsubseteq C \rightarrow C_{0}$ under $\Omc$.
\end{enumerate}
Conversely, the following conditions are also equivalent: 
\begin{enumerate}
	
	\item there is an $\Lmc(\Sigma)$-interpolant for
	$C_{1}\sqsubseteq C_{2}$ under $\Omc$;
	
	\item $\Omc\models C_{1}\sqsubseteq C_{2}$ and there is an explicit $\mathcal{L}(\Sigma)$-definition of $C_{2}$ under $\Omc$ and $C_{2} \rightarrow C_{1}$.
	
\end{enumerate} 
\end{lemma}	
\begin{proof}
We show the second equivalence. Assume that (1) does not hold. To show that (2) does not hold, assume $\Omc\models C_{1} \sqsubseteq C_{2}$ (otherwise we are done). By Theorem~\ref{thm:critint} there exist
pointed interpretations $\Imc_{1},d_{1}$ and $\Imc_{2},d_{2}$ such that
$\Imc_{i}$ is a model of $\Omc$, $d_{1}\in C_{1}^{\Imc_{1}}$, $d_{2}\not\in C_{2}^{\Imc_{2}}$, and $\Imc_{1},d_{1}\sim_{\Lmc,\Sigma} \Imc_{2},d_{2}$.
But then $d_{1}\in ((C_{2}\rightarrow C_{1})\sqcap C_{2})^{\Imc_{1}}$ and $d_{2}\in ((C_{2}\rightarrow C_{1})\sqcap \neg C_{2})^{\Imc_{2}}$ which shows that (2) does not hold by Theorem~\ref{thm:critdef}. The other direction is shown similarly.
\end{proof}

Hence we obtain the following corollary.
\begin{theorem}
Let $\mathcal{L}\in \Dln$. Then there is a polynomial time reduction of
projective $\Lmc$-definition existence to $\Lmc$-interpolant
existence (and thus to generalized \Lmc-interpolant existence). Conversely, there
is a polynomial time reduction of generalized $\Lmc$-interpolant
existence (and thus $\Lmc$-interpolant existence) to projective
$\Lmc$-definition existence if an oracle for \Lmc-subsumption is
admitted. 

Both reductions also exist for the ontology-free case and for RI-ontologies.
\end{theorem}

We now formulate the main complexity results proved in this article.

\begin{theorem}\label{thm:main1}
	Let $\mathcal{L}\in \Dln$. Then \Lmc-interpolant existence, generalized $\Lmc$-interpolant existence, and projective $\Lmc$-definition existence are
		all \TwoExpTime-complete.
\end{theorem}
It follows that interpolant existence and projective definition existence are one exponential harder than subsumption for logics in $\Dln$. Our lower bound proofs rely on the presence of ontologies. To understand the ontology-free case (and the case with RI-ontologies) we first recall from our introduction of DLs in $\Dln$ above that for DLs with the universal role or with both inverse roles and nominals, the ontology can be encoded in a concept and so interpolant existence and projective definition existence are still $\TwoExpTime$-complete with empty ontologies and RI-ontologies, respectively. 
For the remaining DLs in $\Dln$, interpolant existence and projective definition existence become $\coNExpTime$-complete, however. Thus, less complex than with ontologies, but still
harder than subsumption (which is $\PSpace$-complete), under standard complexity theoretic assumptions.
\begin{theorem}\label{thm:main2}
Let $\Lmc\in \Dln$.
\begin{enumerate}
	\item If $\Lmc$ admits nominals and the universal role, or nominals and inverse roles, then ontology-free \Lmc-interpolant existence, generalized $\Lmc$-interpolant existence, and projective 
    $\Lmc$-definition existence are all \TwoExpTime-complete.
    \item If $\Lmc$ admits the universal role and RIs, then RI-ontology \Lmc-interpo-lant existence, generalized $\Lmc$-interpolant existence, and projective 
    $\Lmc$-definition existence are all \TwoExpTime-complete.
    \item If $\Lmc\in \{\mathcal{ALCO},\mathcal{ALCHO}\}$, then ontology-free and RI-ontology
    \Lmc-interpolant existence, generalized $\Lmc$-interpolant existence, and projective 
    $\Lmc$-definition existence are all \coNExpTime-complete. 
    \item If $\Lmc\in \{\mathcal{ALCH},\mathcal{ALCHI}\}$, then RI-ontology
    \Lmc-interpolant existence, generalized $\Lmc$-interpolant existence, and projective 
    $\Lmc$-defini-tion existence are all \coNExpTime-complete.
\end{enumerate}
\end{theorem}
%
We have seen that with the exception of $\mathcal{ALCO}$ and $\mathcal{ALCHO}$, all DLs in $\Dln$ enjoy the non-projective Beth definability property. Hence checking the existence of a non-projective definition of a concept name is polynomial time reducible to subsumption checking and so \ExpTime-complete in the presence of an ontology. The following result states that even for $\mathcal{ALCO}$ and $\mathcal{ALCHO}$ checking the existence of non-projective definitions of concept names is not harder than subsumption. 
\begin{theorem}\label{thm:main3}
Let $\Lmc\in \{\mathcal{ALCO},\mathcal{ALCHO}\}$. Then non-projective \Lmc-definition existence of concept names is \ExpTime-complete.  
\end{theorem}
Interestingly, Theorem~\ref{thm:main3} is the only result where the lack of either the CIP of (P)BDP does not lead to an increase in complexity of the interpolant/explicit definition existence problem. {We show Theorem~\ref{thm:main3} in the appendix provided as supplementary material as it uses techniques that are slightly different from our other main results.}

We next consider the non-projective explicit definability of nominals. We have seen in Example~\ref{exm:bdlnominals} above that for nominals even the non-projective Beth definability property does not hold for any DL in $\Dln$. In fact, the following result states that the non-projective definability of nominals is as hard as their projective definability.
\begin{theorem}\label{thm:main4}
  Let $\Lmc\in \Dln$ admit nominals. Then non-projective \Lmc-definition existence of nominals is \TwoExpTime-complete.
\end{theorem}
Observe that the characterizations given in Theorems~\ref{thm:critint}
and~\ref{thm:critdef} provide mutual polynomial time reductions of generalized interpolant and
definition existence to the complement of joint consistency modulo
\Lmc-bisimulations. Hence, to prove Theorems~\ref{thm:main1} to ~\ref{thm:main4},
it suffices to prove the corresponding complexity bounds for joint consistency.

We finally discuss an interesting consequence for CI-interpolants. Let
$\Lmc$ be a DL in \Dln that does not admit RIs. The
\emph{CI-interpolant existence problem in $\Lmc$} is the problem to
decide for  $\Lmc$-ontologies $\Omc$ and $\Omc'$ whether there exists
an $\Lmc$-CI-interpolant for $\Omc$ and $\Omc'$. 
	
\begin{theorem}\label{thm:ciinterpolants} 
  Let $\Lmc \in \{\mathcal{ALCO}^{u},\mathcal{ALCOI}^{u}\}$. Then
  CI-interpolant existence in $\Lmc$ is \TwoExpTime-complete.
\end{theorem}
Observe that the \TwoExpTime upper bound is an immediate consequence of Point~1 of Theorem~\ref{thm:main2} as we can give a polynomial time reduction of CI-interpolant existence to ontology-free interpolant existence. Assume $\Lmc$-ontologies $\Omc$ and $\Omc'$ are given.
Let $\Sigma=\text{sig}(\Omc) \cap \text{sig}(\Omc')$. We find $\Lmc$-concepts $D$ and $D'$ such that $\Omc$ is equivalent to $\{\top \sqsubseteq D\}$ and $\Omc'$ is equivalent to $\{\top \sqsubseteq D'\}$, respectively. Then there exists a $\Lmc$-CI-interpolant for $\Omc$ and $\Omc'$ iff there exists an $\Lmc$-interpolant for $\forall u.D\sqsubseteq \forall u.D'$. 

The \TwoExpTime lower bound is proved in Section~\ref{sec:lowerbound} (Lemma~\ref{lem:twoexp11}) by adapting the \TwoExpTime lower bound proof for interpolant existence in $\Lmc$.

%
%



\section{Upper Bound Proofs With Ontology} \label{sec:upperbound}

We show the double exponential upper bound of Theorem~\ref{thm:main1}
(and thus of Theorem~\ref{thm:main4})
using a new mosaic elimination procedure that decides joint
consistency modulo $\Lmc$-bisimulations, for all $\Lmc\in \Dln$.  
\begin{theorem}\label{thm:2expupper}
  Let $\mathcal{L}\in \Dln$. Then joint consistency modulo
  $\Lmc$-bisimulations is in \TwoExpTime.
\end{theorem}
To motivate our approach, reconsider Example~\ref{ex:3}. Notice that
in interpretations $\Imc_1,\Imc_2$ witnessing joint consistency of
$C_1$ and $\neg C_2$, $a^{\Imc_1}$ is bisimilar to both $b^{\Imc_2}$
and $d$.  Moreover, it can be easily verified that there are no
witnessing interpretations where $a^{\Imc_1}$ is bisimilar to a single
element in $\Imc_2$. Using an ontology, one can extend this example so
that $a^{\Imc_1}$ is enforced to be bisimilar to exponentially many
elements in $\Imc_2$ in any interpretations $\Imc_1,\Imc_2$ witnessing
joint consistency of two 
concepts (in fact, this will be the basis
for showing the lower bound in the subsequent section).  
Thus, we
cannot consider (pairs of) elements in isolation, but instead need to
consider sets of elements. As usual in DLs, we abstract elements in
interpretations by \emph{types}, which syntactically describe the behavior of these
elements by listing the relevant concepts that are satisfied there.
Correspondingly, sets of elements are abstracted to sets of
types. Since we need to coordinate two interpretations
$\Imc_1,\Imc_2$, we thus consider \emph{mosaics}, which are pairs
$(T_1,T_2)$ of sets of types.  The intuitive meaning of such a
pair is that it describes collections of elements in two
interpretations $\Imc_1$ and $\Imc_2$ which realize precisely the
types in $T_1$ and $T_2$, respectively, and are all mutually bisimilar.
Naturally, not all possible mosaics $(T_1,T_2)$ can be realized in
this way and the goal is to determine the realizable ones. For this
task, we use an elimination procedure. We start with the set
of all possible mosaics and drop the `bad' ones until a fixed point is
reached. We will see that the elimination conditions extend the
conditions known from standard type elimination procedures in a relatively natural way to
mosaics.  Then, concepts $C_1,C_2$ will be jointly consistent
under an ontology \Omc modulo bisimulations if there is a
surviving mosaic $(T_1,T_2)$ such that $C_1$ is contained in some type
in $T_1$ and $C_2$ is contained in some type in $T_2$. 

We will now formalize our approach and start by introducing the
relevant notions. Assume $\mathcal{L}\in \Dln$ and consider an $\Lmc$-ontology $\Omc$, $\Lmc$-concepts
$C_{1},C_{2}$, and a signature $\Sigma$.  Let
$\Xi=\text{sub}(\Omc,C_{1},C_{2})$ denote the closure under single
negation of the set of subconcepts of concepts in $\Omc,C_{1},C_{2}$.
A \emph{$\Xi$-type $t$} is a subset of $\Xi$ such that there exists a
model $\Imc$ of \Omc and $d\in \Delta^{\Imc}$ with
$t=\text{tp}_{\Xi}(\Imc,d)$, where 
\[\text{tp}_{\Xi}(\Imc,d) = \{ C\in \Xi\mid d\in C^{\Imc}\}\] 
is the \emph{$\Xi$-type realized at $d$ in $\Imc$}. Let
$\text{Tp}(\Xi)$ denote the set of all $\Xi$-types. We remark that
the number of $\Xi$-types is at most exponential in
$||\Omc||+||C_1||+||C_2||$ and, moreover, the 
set of all $\Xi$-types can be computed in time 
exponential in $||\Omc||+||C_1||+||C_2||$ for all considered
logics~\cite{handbookDL,DBLP:conf/csl/ArecesBM99}.  A \emph{mosaic} is
a pair $(T_1,T_2)$ of sets of types $T_1,T_2\subseteq \text{Tp}(\Xi)$.
For interpretations $\Imc_1,\Imc_2$ and $i\in \{1,2\}$, the
\emph{mosaic defined by $d\in \Delta^{\Imc_{i}}$} in
$\Imc_{1},\Imc_{2}$ is the pair $(T_{1}(d),T_{2}(d))$ where
\[T_{j}(d) = \{ \text{tp}_{\Xi}(\Imc_{j},e) \mid e\in
\Delta^{\Imc_{j}}, \Imc_{i},d \sim_{\Lmc,\Sigma}\Imc_{j},e\}, \] 
for $j=1,2$. 
 We say that a pair $(T_{1},T_{2})$ of sets $T_{1},T_{2}$
of types is a \emph{mosaic defined by $\Imc_{1},\Imc_{2}$} if there
exists $d \in \Delta^{\Imc_{1}}\cup \Delta^{\Imc_{2}}$ such that
$(T_{1,}T_{2}) = (T_{1}(d),T_{2}(d))$. Clearly, there are at most
doubly exponentially many mosaics.
\begin{example}\label{ex:5bis}
  Recall $C_{1}$, $C_{2}$, $\Sigma$, and $\Imc_{1}, \Imc_{2}$ from
  Example~\ref{ex:3}, and let $\Omc=\emptyset$.  The set $\Xi$
  consists of the concepts $\{a\}$, $\exists r.\{a\}$, $\{b\}$,
  $\exists r.\{b\}$, $C_{1}$, $C_{2}$, and negations thereof.  We have
  that:
  \begin{align*}
    \textnormal{tp}_{\Xi}(\Imc_{1},a^{\Imc_{1}}) & = \{ \{a\},
      \exists r.\{a\}, \lnot \{b\}, \lnot \exists r.\{b\}, C_{1},
      C_{2} \} \\
    \textnormal{tp}_{\Xi}(\Imc_{2},b^{\Imc_{2}}) & = \{ \lnot
	\{a\}, \lnot \exists r.\{a\}, \{b\}, \lnot \exists r.\{b\},
	\lnot C_{1}, \lnot C_{2} \} \\
    \textnormal{tp}_{\Xi}(\Imc_{2},d) & = \{ \lnot \{a\}, \lnot
      \exists r.\{a\}, \lnot \{b\}, \lnot \exists r.\{b\},  \lnot
      C_{1}, C_{2} \}
  \end{align*}
  The mosaic defined by $a^{\Imc_{1}}$ in $\Imc_{1}, \Imc_{2}$ is $(T_{1}(a^{\Imc_{1}}), T_{2}(a^{\Imc_{1}}))$, where 
  \[T_{1}(a^{\Imc_{1}})  = \{
    \textnormal{tp}_{\Xi}(\Imc_{1},a^{\Imc_{1}}) \}
    \quad\text{and}\quad %
    T_{2}(a^{\Imc_{1}})  = \{
      \textnormal{tp}_{\Xi}(\Imc_{2},b^{\Imc_{2}}),
    \textnormal{tp}_{\Xi}(\Imc_{2},d) \}.\]
    \qedex
\end{example}
As announced above, the aim of the mosaic elimination
procedure is to determine all mosaics $(T_{1},T_{2})$ such that all $t\in
T_{1}\cup T_{2}$ can be realized in mutually $\Lmc(\Sigma)$-bisimilar
elements of models $\Imc_1,\Imc_2$ of $\Omc$. In order to formulate the elimination
conditions, we define several compatibility conditions between types
and between mosaics, similar to the compatibility conditions that are
used in standard type elimination procedures.
Throughout the rest of the section, we treat the
universal role $u$ as a role name contained in $\Sigma$, in case
$\Lmc$ admits the universal role. Note that $u^-$ is equivalent to
$u$, and that $\Omc\models r\sqsubseteq u$, for every role $r$.

Let $t_1,t_2$ be $\Xi$-types. We call $t_1,t_2$ \emph{$u$-equivalent}
if  $\exists u.C\in t_1$ iff
$\exists u.C\in t_2$, for every $\exists u.C\in \Xi$. Notice that the
condition is trivially satisfied if \Lmc does not 
admit the universal role.
For a role $r$, we call $t_{1},t_{2}$ 
\emph{$r$-coherent for $\Omc$}, in symbols $t_{1}
\rightsquigarrow_{r} t_{2}$, if $t_1,t_2$ are $u$-equivalent and
the following conditions hold for
all roles $s$ with $\Omc \models r \sqsubseteq s$: (1) if $\neg
\exists s.C\in t_{1}$, then $C\not\in t_{2}$ and (2) if $\neg \exists
s^{-}.C\in t_{2}$, then $C\not\in t_{1}$.  Note that
$t\rightsquigarrow_{r}t'$ iff 
$t'\rightsquigarrow_{r^-}t$.
We lift the definition of
$r$-coherence from types to mosaics $(T_1,T_2),(T_1',T_2')$.
Specifically, we call 
$(T_{1},T_{2})$, $(T_{1}',T_{2}')$ \emph{$r$-coherent}, in symbols
$(T_{1},T_{2})\rightsquigarrow_{r} (T_{1}',T_{2}')$, if for $i=1,2$: 
\begin{itemize}

  \item for every $t\in T_{i}$ there exists a $t'\in T_{i}'$ such that
    $t\rightsquigarrow_{r}t'$, and

  \item if $\Lmc$ admits inverse roles, then for every $t'\in
    T_{i}'$, there is a $t\in T_{i}$ such that
    $t\rightsquigarrow_{r}t'$.

\end{itemize}
Note that $(T_1,T_2)\rightsquigarrow_r(T_1',T_2')$ iff
$(T_1',T_2')\rightsquigarrow_{r^-}(T_1,T_2)$ in case \Lmc admits 
inverse roles. Also notice that $(T_1,T_2)\rightsquigarrow_r
(T_1',T_2')$ implies $(T_1,T_2)\rightsquigarrow_u (T_1',T_2')$, for
every role $r$.
\begin{example}
  Consider again interpretations $\Imc_1,\Imc_2$ from
  Example~\ref{ex:3} and the types $t_1 =
  \textnormal{tp}_{\Xi}(\Imc_1,a^{\Imc_1})$, $t_2 =
  \textnormal{tp}_{\Xi}(\Imc_2,b^{\Imc_2})$, and $t_3 =
  \textnormal{tp}_{\Xi}(\Imc_2,d)$. Then,
  $t_1\rightsquigarrow_{r} t_1$,
  $t_2\rightsquigarrow_{r}t_3$, and
  $t_3\rightsquigarrow_{r}t_3$. Moreover, the mosaic $(T_1,T_2)$
  defined by $a^\Imc_1$ in $\Imc_1,\Imc_2$ satisfies
  $(T_1,T_2)\rightsquigarrow_r (T_1,T_2)$.
  \qedex
\end{example}
We are now in the position to formulate the mosaic elimination
conditions.  Let $\Smc\subseteq 2^{\text{Tp}(\Xi)} \times
2^{\text{Tp}(\Xi)}$ be a set of mosaics. We call
$(T_{1},T_{2})\in\Smc$ \emph{bad} if it violates one of the following
conditions.  
\begin{description}
	
  \item[$\Sigma$-concept name coherence] $A\in t$ iff
    $A\in t'$, for every concept name
    $A\in \Sigma$ and every $t,t'\in T_{1}\cup T_{2}$;

  \item [Existential saturation] for $i=1,2$ and $\exists r.C\in
    t\in T_{i}$, there exists $(T_{1}',T_{2}')\in \mathcal{S}$ such
    that (1) there exists $t'\in T_{i}'$ with $C\in t'$ and
    $t\rightsquigarrow_{r}t'$ and (2) if $\Omc\models
    r\sqsubseteq s$ for a $\Sigma$-role $s$, then
    $(T_{1},T_{2})\rightsquigarrow_{s} (T_{1}',T_{2}')$.


\end{description}
For didactic purposes and because we need it later in
Section~\ref{sec:computation}, we first give the mosaic elimination procedure
for logics \Lmc that do not admit nominals. The procedure starts with the 
set $\Smc_0$ of all mosaics.
Then obtain, for $i\geq 0$, $\Smc_{i+1}$ from $\Smc_i$ by eliminating
all mosaics $(T_1,T_2)$ that are bad in $\Smc_i$. Let $\Smc^*$ be where the
sequence stabilizes. The elimination procedure decides joint
consistency in the following sense.  
\begin{lemma}
  \label{lem:mosaic-elim-wo-nominals} 
  If $\Lmc$ does not admit nominals, the following conditions are equivalent:
   \begin{enumerate}

    \item $C_1,C_2$ are jointly consistent under \Omc modulo $\Lmc(\Sigma)$-bisimulations;

    \item there exist $(T_1,T_2)\in \Smc^*$ and $\Xi$-types $t_{1}\in
      T_1,t_{2}\in T_2$ with $C_{1}\in t_{1}$ and $C_{2}\in t_{2}$.

  \end{enumerate} 	
\end{lemma}
We refrain from giving the proof of
Lemma~\ref{lem:mosaic-elim-wo-nominals} since it will follow from
Lemma~\ref{lem:mosaic-elim} below. We note, however, that for \Lmc as
in the lemma, Theorem~\ref{thm:2expupper} is an immediate consequence
of the procedure: there are only double exponentially many mosaics, so
the elimination terminates after at most double exponentially steps.
It remains to observe that every elimination step can be executed in
double exponential time.

This relatively straightforward elimination procedure does not quite
work in the presence of nominals.  Intuitively, the reason is that in
any two interpretations $\Imc_1,\Imc_2$, every nominal $a$ is realized
(modulo bisimulation) in exactly one mosaic. Now, if the set $\Smc$
contains several mosaics mentioning $a$, they possibly witness
existential saturation of each other which, however, cannot be
reflected in an interpretation.  Thus, for the mosaic elimination
procedure to work (in the sense of
Lemma~\ref{lem:mosaic-elim-wo-nominals}) one has to ``guess'' for
every nominal $a$ exactly one mosaic that describes $a$. 

To formalize this idea, let us call a set $\mathcal{S}$ of mosaics
\emph{good for nominals} if for every individual name $a\in
\text{sig}(\Xi)$ and $i=1,2$ there exists exactly one $t_{a}^{i}$ with
$\{a\}\in t_{a}^{i}\in \bigcup_{(T_{1},T_{2})\in \mathcal{S}}T_{i}$
and exactly one pair $(T_{1},T_{2})\in \mathcal{S}$ with $t_{a}^{i}\in
T_{i}$. Moreover, if $a\in \Sigma$, then that pair takes the form
\begin{itemize}

  \item $(\{t_{a}^1\},\{t_a^2\})$ in case \Lmc admits the universal role, and

  \item $(\{t_{a}^{1}\},\{t_{a}^{2}\})$, $(\{t_{a}^{1}\},\emptyset)$,
    or $(\emptyset, \{t_{a}^{2}\})$, otherwise.

\end{itemize}
%
We can now formulate the more general lemma.
\begin{restatable}{lemma}{lemmaupper} \label{lem:mosaic-elim} 
  The following conditions are equivalent:
   \begin{enumerate}

    \item $C_1,C_2$ are jointly consistent under \Omc
      modulo $\Lmc(\Sigma)$-bisimulations;

    \item there exists a set $\Smc^*$ of mosaics that is good for
      nominals and does not contain a bad mosaic, such that there exist $(T_1,T_2)\in \Smc^*$ and $\Xi$-types $t_{1}\in
      T_1,t_{2}\in T_2$ with $C_1\in t_{1}$ and $C_{2}\in t_{2}$.

  \end{enumerate} 	
\end{restatable}
\begin{proof}
  ``1 $\Rightarrow$ 2''. Let $\Imc_{1},d_{1}\sim_{\Lmc,\Sigma}
  \Imc_{2},d_{2}$ for models $\Imc_{1}$ and $\Imc_{2}$
  of $\Omc$ such that $d_{1},d_{2}$ realize $\Xi$-types
  $t_{1},t_{2}$ and $C_{1}\in t_{1}, C_{2}\in t_{2}$. 
  Let $\Smc^*$ be
  the set of all mosaics defined by $\Imc_1,\Imc_2$.
%
  It is routine to show that no $(T_1,T_2)$ in $\Smc^*$ is bad and that
  $\Smc^*$ is good for nominals. Now, the mosaic $(T_1,T_2)$ defined by $d_1^\Imc$ in
  $\Imc_1,\Imc_2$ witnesses Condition~(2).

  \smallskip ``2 $\Rightarrow$ 1''.  Suppose there exist a good set
  $\Smc^*$ of mosaics and $(S_1,S_2)\in
  \Smc^*$ and $\Xi$-types $s_{1}\in S_1,s_{2}\in S_2$ with $C_{1}\in
  s_{1}$ and $C_{2}\in s_{2}$. Let $\Imc_i$, for $i=1,2$ be interpretations defined by setting:
  \begin{align*}
    \Delta^{\Imc_{i}} &:= \{ (t,(T_{1},T_{2})) \mid (T_{1},T_{2})\in
    \Smc^*, t\in T_i,\text{ and }\\
    & \hspace{0.75cm}\text{ if \Lmc admits the universal role, then }
    (S_1,S_2)\rightsquigarrow_u(T_1,T_2) \\
    & \hspace{0.75cm} \text{
    and $t,s_i$ are $u$-equivalent}\}\\ 
    r^{\Imc_{i}} &:=\{((t,p),(t',p'))\in \Delta^{\Imc_{i}}\times
    \Delta^{\Imc_{i}} \mid t\rightsquigarrow_{r}t'\text{ and }
    \text{for all $\Sigma$-roles } s:\\
    & \hspace{0.75cm} ((\Omc\models r
    \sqsubseteq s)\Rightarrow p\rightsquigarrow_{s}p')\} \\
    A^{\Imc_{i}} &:= \{(t,p)\in \Delta^{\Imc_{i}}\mid A\in t\} \\
    a^{\Imc_{i}} & := (t,(T_{1},T_{2}))\in \Delta^{\Imc_{i}}, \{a\}\in t\in T_{i}
  \end{align*}
  Note that the interpretation of nominals is well-defined since
  $\Smc^*$ is good for nominals. 
  
  We verify that interpretations $\Imc_{1}$ and $\Imc_{2}$
  witness Condition~(1).  

  \smallskip\noindent\textit{Claim~1.} For $i=1,2$, all $C\in \Xi$,
  and all $(t,p)\in \Delta^{\Imc_i}$, we have $(t,p)\in C^{\Imc_i}$ iff
  $C\in t$.

  \smallskip\noindent\textit{Proof of Claim~1.} Let $i\in\{1,2\}$. The proof is by
  induction on the structure of concepts in $\Xi$.
  \begin{itemize}

    \item The claim holds for concept names $C=A$ and all
      nominals $C=\{a\}$, by definition of
      $\Imc_i$.

    \item The Boolean cases, $\neg C$ and $C\sqcap C'$, are
      immediate consequences of the hypothesis. 

    \item Let $C = \exists r.D$. (Recall that $r$ is possibly the
      universal role $u$.)

      ``if'': Suppose $\exists r.D\in t$. By existential saturation,
      there is a $p'=(T_1',T_2')\in\Smc^*$ such that~(1) there exists
      $t'\in T_i'$ with $D\in t'$ and $t\rightsquigarrow_{r},t'$
      and~(2) if $\Omc\models r\sqsubseteq s$ for some $\Sigma$-role
      $s$, then $p\rightsquigarrow_s p'$. Note that $t,t'$ are thus
      also $u$-equivalent, so $(t',p')\in \Delta^{\Imc_i}$. We
      distinguish cases: 

      \begin{itemize}

	\item If $r$ is a role name, then by definition of
	  $r^{\Imc_i}$, we have that $( (t,p),(t',p'))\in r^{\Imc_i}$.
	  Since $D\in t'$, induction yields $(t',p')\in D^{\Imc_i}$.
	  Overall, we get $(t,p)\in (\exists r.D)^{\Imc_i}$.

	\item If $r=r_0^-$ is an inverse role, then~(1) and~(2) above
	  imply~(1') $t'\rightsquigarrow_{r_0}t$ and~(2') if
	  $\Omc\models r_0\sqsubseteq s$ for some $\Sigma$-role $s$,
	  then $p\rightsquigarrow_s p'$. As before, we can then
	  conclude that $( (t',p'),(t,p))\in r_0^{\Imc_i}$.  Since
	  $D\in t'$, induction yields $(t',p')\in D^{\Imc_i}$.
	  Overall, we get $(t,p)\in (\exists r_0^-.D)^{\Imc_i}$.

      \end{itemize}

      ``only if'': Suppose $(t,p)\in (\exists r.D)^{\Imc_i}$. Then,
      there is $(t',p')\in \Delta^{\Imc_i}$ with 
      $( (t,p),(t',p'))\in r^{\Imc_i}$ and $(t',p')\in D^{\Imc_i}$.
      By induction, the latter implies $D\in t'$.  We distinguish cases: 

      \begin{itemize}

	\item If $r$ is a role name, then by definition of
	  $r^{\Imc_i}$, $t\rightsquigarrow_{r} t'$ and thus
	  $\exists r.D\in t$.

	\item If $r=r_0^-$ is an inverse role, then by definition of
	  $r_0^{\Imc_i}$, $t'\rightsquigarrow_{r_0}t$. Thus, also
	  $\exists r^-_0.D\in t$.

      \end{itemize}

  \end{itemize}
  This finishes the proof of Claim~1. Claim~1 implies that
  $(s_1,(S_1,S_2))\in C_1^{\Imc_1}$ and $(s_2,(S_1,S_2))\in
  C_2^{\Imc_2}$. Claim~1 also implies that the type realized by
  $(t,p)$ in $\Imc_i$ is $t$, for all $(t,p)\in \Delta^{\Imc_i}$.
  Since types are, by definition, realized in models of \Omc, it
  follows that both $\Imc_1$ and $\Imc_2$ are models of \Omc.

  \smallskip\noindent\textit{Claim~2.} The relation $R$ defined by 
  \[R = \{ ((t,p), (t',p))\mid (t,p)\in \Delta^{\Imc_1}, (t',p)\in
  \Delta^{\Imc_2}\}\] 
  is an $\Lmc(\Sigma)$-bisimulation. 

  \smallskip\noindent\textit{Proof of Claim~2.} Clearly, $R$
  satisfies Condition~[AtomC] due to $\Sigma$-concept name coherence.
  Condition~[AtomI] follows from the fact that $\Smc^*$ is good for
  nominals in case \Lmc admits nominals. 

  For Condition~[Forth], let $((t,p), (t',p))\in R$ and
  $((t,p),(t_1,p_1))\in r^{\Imc_1}$, for some $\Sigma$-role $r$, and
  let $p=(T_1,T_2)$ and $p_1=(T_1',T_2')$. We
  distinguish cases: 
  \begin{itemize}

    \item If $r$ is a role name, then by definition of $r^{\Imc_1}$,
      we have~(1) $t\rightsquigarrow_{r} t_1$ and~(2) for all
      $\Sigma$-roles $s$ with $\Omc\models r\sqsubseteq s$, we have
      $p\rightsquigarrow_s p_1$. Since $t'\in T_2$ and
      $p\rightsquigarrow_r p_1$ there is some $t''\in T_2'$ with
      $t'\rightsquigarrow_{r} t''$. Thus, in particular, $t''$ is
      $u$-equivalent to $t'$ (and thus to $s_2$), which implies
      $(t'',p_1)\in \Delta^{\Imc_2}$. The definition of $r^{\Imc_2}$
      then implies that $( (t',p), (t'',p_1))\in r^{\Imc_2}$. It
      remains to note that the definition of $R$ yields $(
      (t_1,p_1),(t'',p_1))\in R$.

    \item If $r=r_0^-$ is an inverse role, then by definition of
      $r_0^{\Imc_1}$, we have~(1) $t_1\rightsquigarrow_{r_0} t$
      and~(2) for all $\Sigma$-roles $s$ with $\Omc\models
      r_0\sqsubseteq s$, we have $p_1\rightsquigarrow_s p$. Since
      $t'\in T_2$ and $p_1\rightsquigarrow_{r_0} p$ there is some
      $t''\in T_2'$ with $t''\rightsquigarrow_{r_0} t'$. Thus, in
      particular, $t''$ is $u$-equivalent to $t'$ (and thus to $s_2$),
      which implies $(t'',p_1)\in \Delta^{\Imc_2}$. The definition of
      $r_0^{\Imc_2}$ then implies that $( (t'',p_1), (t',p))\in
      r_0^{\Imc_2}$. It remains to note that the definition of $R$
      yields $( (t_1,p_1),(t'',p_1))\in R$.

  \end{itemize}
  Condition~[Back] is dual. 

  Finally, we verify that $R$ and $R^-$ are surjective if \Lmc admits the universal role. Let
  $(t,(T_1,T_2))\in \Delta^{\Imc_1}$. Then,
  $(S_1,S_2)\rightsquigarrow_u(T_1,T_2)$, by definition of
  $\Delta^{\Imc_1}$. This implies that there is a type $t'\in
  T_2$ which is $u$-equivalent to $s_2$ and thus $(t',(T_1,T_2))\in
  \Delta^{\Imc_2}$. The definition of $R$ implies
  $((t,(T_1,T_2)),(t',(T_1,T_2)))\in R$. The other direction is dual.

  \smallskip This finishes the proof of Claim~2. 
  By definition of $R$, we have 
  $( (s_1,(S_1,S_2)), (s_2,(S_1,S_2)))\in R$, 
  and thus
  $\Imc_1,(s_1,(S_1,S_2))\sim_{\Lmc,\Sigma}\Imc_2,(s_2,(S_1,S_2))$.
\end{proof}

It remains to argue that we can find in double exponential time a set
$\Smc^*$ as in Condition~(2) of Lemma~\ref{lem:mosaic-elim}. We use a
suitable variant of the elimination procedure described after
Lemma~\ref{lem:mosaic-elim-wo-nominals}.
\begin{restatable}{lemma}{lemmaupperz} \label{lem:dexp}
  Let $\mathcal{L}\in \Dln$. Then it is decidable in 
  time double exponential in $||\Omc||+||C_1||+||C_2||$ whether for an $\Lmc$-ontology $\Omc$,
  $\Lmc$-concepts $C_{1},C_{2}$, and a signature $\Sigma\subseteq
  \text{sig}(\Xi)$ there exists some $\mathcal{S}^*$
  satisfying Condition~(2) of Lemma~\ref{lem:mosaic-elim}.
\end{restatable}

\begin{proof} 
  Let $\Lmc\in\Dln$, and assume $\Omc$, $C_{1},C_{2}$, and $\Sigma$
  are given. We can enumerate in double exponential time the maximal
  good sets $\mathcal{U}\subseteq 2^{T(\Xi)}\times 2^{T(\Xi)}$ by
  picking, for each nominal $a\in
  \text{sig}(\Xi)$ and $i=1,2$, a type $t_a^{i}$, and a mosaic
  $(T_1,T_2)$ with $t_a^i\in T_i$. In doing so, we make sure that
  $(\{t_{a}^{1}\},\{t_{a}^{2}\})$ is selected in case $a\in \Sigma$.
  Crucially, there are only double exponentially many possibilities to make this
  choice. Remove all mosaics that mention a nominal and have not been
  selected. The resulting set is good for nominals. 
%

  Then we eliminate from any set $\mathcal{U}$ obtained in that
  process recursively all bad mosaics. 
  Let $\mathcal{S}_{\Umc}\subseteq \mathcal{U}$ be the largest fixpoint of
  that procedure. Then one can easily show that there exists a set
  $\mathcal{S}^*$ satisfying Condition~(2) of Lemma~\ref{lem:mosaic-elim}
  iff there exists a set $\mathcal{U}$ that can be obtained by the
  process described above
  such that the largest fixpoint $\mathcal{S}_{\Umc}$ satisfies
  Condition~(2) of Lemma~\ref{lem:mosaic-elim}. Since elimination
  terminates after double exponential time, and there are only double
  exponentially many possible choices for \Umc, the lemma follows. 
\end{proof}

Theorem~\ref{thm:2expupper} is a direct consequence of
Lemmas~\ref{lem:mosaic-elim} and \ref{lem:dexp}.

\section{Lower Bound Proofs With Ontology}
\label{sec:lowerbound}

The goal of this section is to provide the proofs of the lower bounds in
Theorems~\ref{thm:main1},~\ref{thm:main4},
and~\ref{thm:ciinterpolants}. We start with the former two.  By
Lemma~\ref{lem:reduce1} and Theorem~\ref{thm:critdef}, it suffices to
consider joint consistency. 
We will provide two reductions: in
Section~\ref{sec:lowernom}, we provide the reduction for DLs in $\Dln$ that
admits nominals and, in Section~\ref{sec:lower-role-inclusions}, the one for
DLs that admits role inclusions. In
Section~\ref{sec:shapeofinterpolant}, we will investigate the shape of
the interpolants / explicit definitions that arise in the preceding lower bound
proof. 
In Section~\ref{sec:lowerci}, we then show how to adapt the lower
bound proof from Section~\ref{sec:lowernom} to the case of
CI-interpolant existence. 
In all cases we reduce the word problem for languages recognized by
exponentially space bounded alternating Turing machines, which we
introduce next.

An \emph{alternating Turing
machine (ATM)} is a tuple
\[
M=(Q,\Theta,\Gamma,q_0,\Delta)
\]
where
$Q=Q_{\exists}\uplus Q_{\forall}$
is a finite
set of states partitioned into \emph{existential states}~$Q_{\exists}$
and \emph{universal states}~$Q_{\forall}$.
Further, $\Theta$ is the input alphabet and $\Gamma$ is the tape
alphabet that contains a \emph{blank symbol} $\Box \notin \Theta$,
$q_0\in Q_{\forall}$ is the \emph{initial state}, and $\Delta\subseteq
Q\times \Gamma\times Q\times \Gamma \times \{L,R\}$ is the
\emph{transition relation}.  We assume without loss of generality that
the set $\Delta(q,a):=\{(q',a',M)\mid (q,a,q',a',M)\in\Delta\}$
contains exactly two or zero elements for every $q\in Q$ and $a \in
\Gamma$. Moreover, the state $q'$ must be from $Q_\forall$ if $q \in
Q_\exists$ and from $Q_\exists$ otherwise, that is, existential and
universal states alternate. 
Acceptance of ATMs is defined in a slightly unusual way, without using
accepting states. Intuitively, an ATM accepts if it runs forever on
all branches and rejects otherwise.  More formally, a
\emph{configuration} of an ATM is a word $wqw'$ with
\mbox{$w,w'\in\Gamma^*$} and $q\in Q$. 
We say that $wqw'$ is \emph{existential} if~$q$ is, and likewise for
\emph{universal}.  \emph{Successor configurations} are defined in the
usual way.  Note that every configuration has exactly zero or two
successor configurations. A \emph{computation tree} of an ATM $M$ on
input $w$ is a (possibly infinite) tree whose nodes are labeled with
configurations of $M$ such that
\begin{itemize}

  \item the root is labeled with the initial configuration
    $q_0w$;

  \item if a node is labeled with an existential configuration 
    $wqw'$, then it has a single successor which is labeled 
    with a successor configuration of $wqw'$;

  \item if a node is labeled with a universal configuration $wqw'$,
    then it has two successors which are labeled with the two
    successor configurations of~$wqw'$.

\end{itemize}
An ATM $M$ \emph{accepts} an input $w$ if there is a computation tree
of $M$ on $w$. Note that we can convert any ATM $M$ in which
acceptance is based on accepting states to our model by assuming that
$M$ terminates on any input and then modifying it to enter an infinite
loop from the accepting states. It is well-known that there are
$2^n$-space bounded ATMs which recognize a \TwoExpTime-hard
language~\cite{chandraAlternation1981}, where $n$ is the length of the input $w$.

\subsection{DLs with Nominals}\label{sec:lowernom}
We start with DLs supporting nominals. 
By Theorem~\ref{thm:critdef}, it suffices to prove the following result.
\begin{lemma}\label{lem:lowernom}
  Let $\mathcal{L}\in \Dln$ admit nominals.
  It is \TwoExpTime-hard to
  decide for an $\Lmc$-ontology $\Omc$, individual name $b$, and
  signature $\Sigma\subseteq \text{sig}(\Omc)\setminus \{b\}$ whether
  $\{b\}$ and $\neg\{b\}$ are jointly consistent under \Omc  modulo
  $\Lmc(\Sigma)$-bisimulations. This is true even if $b$ is the only
  individual in $\Omc$ and $\Sigma=\text{sig}(\Omc)\setminus \{b\}$.
\end{lemma}

As announced, we reduce the word problem for $2^n$-space bounded ATMs.
Let us fix such an ATM $M=(Q,\Theta,\Gamma,q_0,\Delta)$ and an input
$w=a_0\ldots a_{n-1}$ of length $n$.  We first provide the reduction
for $\Lmc=\ALCO$ using an ontology $\Omc$ and a signature $\Sigma$ such
that $\Omc$ contains concept names that are not in $\Sigma$ and uses
two role names $r,s$, and show later how to adapt this proof to
$\Sigma=\text{sig}(\Omc)\setminus\{b\}$ and DLs
supporting inverses and/or the universal role. 

The idea of the reduction is as follows. We aim to construct an
ontology \Omc such that $M$ accepts $w$ iff $\{b\}$ and
$\neg\{b\}$ are jointly consistent under \Omc modulo
$\Lmc(\Sigma)$-bisimulations, where
\[\Sigma  = \{r,s,Z,B_{\forall},B_{\exists}^1,B_{\exists}^2\}\cup
\{A_\sigma\mid \sigma\in \Gamma\cup (Q\times \Gamma)\}.\]
The ontology $\Omc$ enforces that $r(b,b)$ holds in any model \Omc
using the concept inclusion $\{b\}\sqsubseteq \exists r.\{b\}$.
Moreover, it enforces that any element distinct from $b^\Imc$ with
an $r$-successor lies on an infinite $r$-path $\rho$ enforced by the concept
inclusions:
\begin{align*}
  \neg \{b\} \sqcap \exists r.\top & \sqsubseteq I_{s} &
  I_s \sqsubseteq \exists r.\top\sqcap \forall r.I_s
\end{align*}
with $I_{s}$ a concept name.
Thus, if there exist models \Imc, \Jmc of \Omc with
$\Imc,b^\Imc\sim_{\ALCO,\Sigma}\Jmc,d$ for some $d\neq b^{\Jmc}$ and
$d\in (\exists r.\top)^\Jmc$, it follows that all elements on the path
$\rho$ are
$\ALC(\Sigma)$-bisimilar to $b^\Imc$ and thus mutually
$\ALC(\Sigma)$-bisimilar. The situation is depicted in
Figure~\ref{fig:loweronto}, where the trees $T_{*}$ and $T_i$,
$i\geq 0$ starting in $b^\Imc$ and on the path elements, respectively,
are also mutually $\ALC(\Sigma)$-bisimilar. These trees shall
represent the computation tree of $M$ on input $w$ (using symbols from
$\Sigma$) as follows, cf.~Figure~\ref{fig:computation-tree} which
shows the skeleton of a single tree $T_i$. Configurations of $M$ are
represented as paths of length $2^n$ over a role $s$ in which
every element is labeled with a symbol $A_\sigma$,
$\sigma\in\Gamma\cup(Q\times \Gamma)$ that represents the content of a
single tape cell (omitted in the figure for the sake of readibility). In Figure~\ref{fig:computation-tree}, the start of a
configuration is indicated by $\circ$ and the $s$-path between
consecutive $\circ$ has length $2^n$.
Every configuration is marked as existential or universal using concept
names $B_\forall,B_\exists^1,B_\exists^2$; the superscript
$\cdot^1$/$\cdot^2$ indicates which successor is chosen for an
existential configuration. Existential configurations have a single
successor configuration and universal configurations have two
successor configurations.

\begin{figure}[th]
\centering
\begin{tikzpicture}
\tikzset{
dot/.style = {draw, fill=black, circle, inner sep=0pt, outer sep=0pt, minimum size=2pt}
}

\node (b) at (0,0) {};
\draw (0,0) node[dot, label=north:$b^\Imc$] {};
\draw[->, >=stealth]  (b) edge [loop left] node[] {$r$} ();

\draw (0,-0.1) -- (-1,-2.1) -- (1,-2.1)-- (0,-0.1);
\draw (0,-1.4) node {$T_{\ast}$};

\draw (4,0) node[dot, label=north:{$d\neq b^\Jmc$}, label=south east:$0$] (0) {};

\draw (4.1,0) -- (6,-0.4) -- (6,0.4)-- (4.1,0);
\draw (5.5,0) node {$T_{0}$};

\draw[->, >=stealth] (4,-0.1) -- (4,-1.2) node[midway, left] {$r$};

\draw (4,-1.3) node[dot, label=south east:$1$] (1) {};
\draw (4.1,-1.3) -- (6,-1.7) -- (6,-0.9)-- (4.1,-1.3);
\draw (5.5,-1.3) node {$T_{1}$};

\draw[->, >=stealth] (4,-1.4) -- (4,-2.5) node[midway, left] {$r$};

\draw (4,-2.6) node[dot, label=south east:$2$] (2) {};
\draw (4.1,-2.6) -- (6,-2.2) -- (6,-3)-- (4.1,-2.6);
\draw (5.5,-2.6) node {$T_{2}$};

\draw[->, >=stealth] (4,-2.7) -- (4,-3.8) node[midway, left] {$r$};

\draw (4,-4.1) node {$\vdots$};

\draw (5.5,-3.5) node {$\vdots$};

\draw (3.4,-4) node {$\rho$};

%
%




\path[black, dashed, out=20,in=170] (b) edge (0);
\path[black, dashed, bend left] (b) edge (1);
\path[black, dashed, out=-15, in=125] (b) edge (2);

\draw (3,-2.1) node {\color{black}{$\vdots$}};



\end{tikzpicture}
\caption{Enforced bisimulation in lower bound}\label{fig:loweronto}

\end{figure}



The structure of this computation
tree can easily be enforced in \ALC using standard techniques (as we
detail below). The difficulty is to achieve synchronization between
successor configurations in the tree. That is, 
if a configuration $c$ in the computation tree is followed by another
configuration $c'$, then $c'$ is actually a successor configuration of
$c$ according to $M$. 
To achieve this, we first ensure that in $T_i$, for 
$i\leq 2^n$, the
$(2^n-i)$-th cell of each configuration in the computation tree is synchronized with the
$(2^n-i)$-th cell of the next configuration(s), as indicated by the
dotted lines in Figure~\ref{fig:computation-tree}. This can be
realized in \ALC using a set of concept
names not in $\Sigma$. Then we exploit the fact
that the trees $T_i$ are mutually $\ALC(\Sigma)$-bisimilar which implies that in \emph{all}
$T_i$ \emph{all} cells of \emph{all} configurations are
synchronized.  In more detail, we do this by using several counters modulo
$2^n$ as follows.

\begin{figure}[th]
	\centering
	\begin{tikzpicture}
		\tikzset{
			dot/.style = {draw, fill=black, circle, inner
			sep=0pt, outer sep=1pt, minimum size=3pt},
			wdot/.style = {draw, fill=white, circle, inner
			sep=0pt, outer sep=1pt, minimum size=3pt}
		}
		

\draw (0,0.75) node[label=$T_{i}$] (Ti) {};


\draw (0,0) node[wdot, label=south:$B_\forall$] (a2) {};
\draw (0.75,0) node[dot, label=$$] (a3) {};
\draw (1.125,-0.275) node[label=$\ldots$] (dot1) {};
\draw (1.5,0) node[dot, label=south:$\!\!\!\!\!\!2^n -i$] (f) {};
\draw (1.875,-0.275) node[label=$\ldots$] (dot2) {};
\draw (2.25,0) node[dot, label=$$] (a4) {};
\draw (3,0) node[dot, label=$$] (a5) {};

\draw (3.75,-1) node[wdot, label=north:$B_\exists^{*}$] (aa1) {};
\draw (4.5,-1) node[dot, label=$$] (ii1) {};
\draw (4.875,-1.25) node[label=$\ldots$] (dota1) {};
\draw (5.25,-1) node[dot, label=south:$2^n-i$] (aa2) {};
\draw (5.625,-1.25) node[label=$\ldots$] (dota2) {};
\draw (6,-1) node[dot,label=$$] (aa3) {};
\draw (6.75,-1) node[wdot,label=north:$B_\forall$] (aa4) {};
\draw (7.5,-1) node[dot, label=$$] (kk1) {};
\draw (7.875,-1.25) node[label=$\ldots$] (dota3) {};
\draw (8.25,-1) node[dot, label=south:$2^n-i$] (abk1) {};
\draw (8.625,-1.25) node[label=$\ldots$] (dota4) {};
\draw (9,-1) node[dot, label=$$] (aa5) {};
\draw (9.75,-1) node[dot, label=$$] (aa6) {};
\draw (10.5,-1.5) node[wdot, label=$$] (aaa1) {};
\draw (10.875,-1.775) node[label=$\ldots$] (dotaaf) {};
\draw (10.5,-0.5) node[wdot, label=$$] (aab1) {};
\draw (10.875,-0.75) node[label=$\ldots$] (dotabf) {};

\draw (3.75,1) node[wdot, label=south:$B_\exists^{*}$] (ab1) {};
\draw (4.5,1) node[dot, label=$$] (ii2) {};
\draw (4.875,0.725) node[label=$\ldots$] (dotb1) {};
\draw (5.25,1) node[dot, label=south:$2^n-i$] (ab2) {};
\draw (5.625,0.725) node[label=$\ldots$] (dotb2) {};
\draw (6,1) node[dot,label=$$] (ab3) {};
\draw (6.75,1) node[wdot,label=south:$B_\forall$] (ab4) {};
\draw (7.5,1) node[dot, label=$$] (kk2) {};
\draw (7.875,0.725) node[label=$\ldots$] (dotb3) {};
\draw (8.25,1) node[dot,label=south:$2^n-i$] (abk2) {};
\draw (8.625,0.725) node[label=$\ldots$] (dotb4) {};
\draw (9,1) node[dot, label=$$] (ab5) {};
\draw (9.75,1) node[dot, label=$$] (ab6) {};
\draw (10.5,0.5) node[wdot, label=$$] (aba1) {};
\draw (10.875,0.225) node[label=$\ldots$] (dotbaf) {};
\draw (10.5,1.5) node[wdot, label=$$] (abb1) {};
\draw (10.875,1.25) node[label=$\ldots$] (dotbbf) {};


\draw[->, >=stealth] (a2) edge (a3)
node[label={[shift={(0.35,0)}]:$s$}] {};
\draw[->, >=stealth] (a4) edge (a5) node[label={[shift={(0.35,0.2)}]:$$}] {};
\draw[->, >=stealth] (a5) edge (aa1) node[label={[shift={(0.35,0.2)}]:$$}] {};
\draw[->, >=stealth] (a5) edge (ab1) node[label={[shift={(0.35,0.2)}]:$$}] {};
\draw[->, >=stealth] (aa1) edge (ii1) node[label={[shift={(0.35,0.2)}]:$$}] {};
%
\draw[->, >=stealth] (aa3) edge (aa4) node[label={[shift={(0.35,0.2)}]:$$}] {};
\draw[->, >=stealth] (aa4) edge (kk1) node[label={[shift={(0.35,0.2)}]:$$}] {};
\draw[->, >=stealth] (aa5) edge (aa6) node[label={[shift={(0.35,0.2)}]:$$}] {};
\draw[->, >=stealth] (aa6) edge (aaa1) node[label={[shift={(0.35,0.2)}]:$$}] {};
\draw[->, >=stealth] (aa6) edge (aab1) node[label={[shift={(0.35,0.2)}]:$$}] {};

\draw[->, >=stealth] (ab1) edge (ii2) node[label={[shift={(0.35,0.2)}]:$$}] {};
%
\draw[->, >=stealth] (ab3) edge (ab4) node[label={[shift={(0.35,0.2)}]:$$}] {};
\draw[->, >=stealth] (ab4) edge (kk2) node[label={[shift={(0.35,0.2)}]:$$}] {};
\draw[->, >=stealth] (ab5) edge (ab6) node[label={[shift={(0.35,0.2)}]:$$}] {};
\draw[->, >=stealth] (ab6) edge (aba1) node[label={[shift={(0.35,0.2)}]:$$}] {};
\draw[->, >=stealth] (ab6) edge (abb1) node[label={[shift={(0.35,0.2)}]:$$}] {};

			
\path[black, dashed, bend right] (f) edge (aa2);
\draw (2.5,0.15) node[label=$$] () {};
\path[black, dashed, bend right] (aa2) edge (abk1);
\draw (2.5,0.15) node[label=$$] () {};

\path[black, dashed, bend left] (f) edge (ab2);
\draw (2.5,0.15) node[label=$$] () {};
\path[black, dashed, bend left] (ab2) edge (abk2);
\draw (2.5,0.15) node[label=$$] () {};

\end{tikzpicture}	
		
\caption{Computation tree of $M$}
\label{fig:computation-tree}
	
\end{figure}

The first counter counts modulo $2^{n}$ along the path $\rho$ using
concept names not in $\Sigma$. As announced, each point of $\rho$
starts an infinite tree along role $s$ that is supposed to mimick the
computation tree of $M$ on input $w$. Along this tree, two more counters are maintained:
\begin{itemize}

  \item one counter starting at $0$ and counting modulo $2^n$, and 
 
  \item another counter starting at \emph{the value of the counter on
    $\rho$} and also counting modulo $2^{n}$.

\end{itemize}
The first counter is used to divide the tree into
configurations of length $2^n$ and the second counter is used to link
the $(2^n-i)$-th cell of successive configurations in $T_i$ as described
above.
%

We will next provide the concept inclusions in $\Omc$ in more detail.
The counter along $\rho$ is realized using concept names $A_i$,
$0\leq i<n$ and by including the following (standard) concept inclusions, for every $i$ with
$0\leq i<n$:
\begin{align*}
  I_s \sqcap A_i \sqcap \bigsqcap_{j < i} A_j & \sqsubseteq \forall r.
  \neg A_i &
  I_s \sqcap \neg A_i \sqcap \bigsqcap_{j < i} A_j & \sqsubseteq \forall r.
  A_i\\
  I_s \sqcap A_i \sqcap \bigsqcup_{j < i} \neg {A}_j & \sqsubseteq \forall r.
  A_i
  & I_s\sqcap \neg{A}_i \sqcap \bigsqcup_{j < i} \neg{A}_j & \sqsubseteq
  \forall r. \neg{A}_i
\end{align*}
Using again the concept name $I_s$, we start the $s$-trees with two counters, realized using
concept names $U_i$ and $V_i$, $0\leq i<n$,
and initialized to $0$ and the value of the $A$-counter, respectively,
by including the following concept inclusions for every $j$ with $0\leq j< n$:
\begin{align*}
  I_s & \sqsubseteq (U=0) \\
  I_s & \sqsubseteq A_j\leftrightarrow V_j \\
%
%
  \top & \sqsubseteq \exists s.\top
\end{align*}
Here, $(U=0)$ is an abbreviation for the concept
$\bigsqcap_{i=0}^{n-1}\neg U_i$; we use similar abbreviations below
without further notice.  The counters $U_i$ and $V_i$ are incremented
along $s$ in the same way as $A_i$ is incremented along $r$, so we omit details. Configurations
of $M$ are represented between two consecutive points having
$U$-counter value $0$. We next enforce the structure of the
computation tree (recall that $q_0\in Q_\forall$):
\begin{align*}
  I_s & \sqsubseteq B_{\forall} \tag{$\dagger$}\label{eq:startrolei}\\
  %
  %
  (U<2^n-1) \sqcap B_\forall & \sqsubseteq \forall s.B_{\forall} \\
  (U<2^n-1) \sqcap B_\exists^i & \sqsubseteq \forall
  s.B_{\exists}^i && i\in\{1,2\} \\
  (U=2^n-1) \sqcap B_\forall & \sqsubseteq \forall
  s.(B_{\exists}^1\sqcup B_{\exists}^2) \\
  (U=2^n-1) \sqcap B_\exists^i& \sqsubseteq \forall
  s.B_\forall && i\in\{1,2\} \\
  (U=2^n-1) \sqcap B_{\forall} & \sqsubseteq \exists s.Z\sqcap \exists
  s.\neg Z
   %
\end{align*}
These concept inclusions enforce that all points which represent a configuration
satisfy one of $B_{\forall},B_{\exists}^1,B_{\exists}^2$
indicating the kind of configuration and, if existential, also a choice of the transition
function. The symbol $Z\in \Sigma$ enforces the branching.

We next set the initial configuration, for
input $w=a_0,\dots,a_{n-1}$.
\begin{align*}
  I_{s} & \sqsubseteq  A_{q_{0},a_0}\\
  I_{s} & \sqsubseteq  \forall s^{k}.A_{a_{k}} & 0<i<n \\
  I_{s} & \sqsubseteq  \forall s^{n}.\mn{Blank} \\
  \mn{Blank} & \sqsubseteq A_{\Box} \\
  \mn{Blank} \sqcap (U<2^n -1) & \sqsubseteq \forall s. \mn{Blank}
\end{align*}
To coordinate successor configurations, we associate with $M$
functions $f_i$, $i\in \{1,2\}$ that map the content of three
consecutive cells of a configuration to the content of the middle cell
in the $i$-the successor configuration (assuming an arbitrary order on
the set $\Delta(q,a)$, for all $q,a$). In what follows, we ignore the corner cases that occur
at the border of configurations; they can be treated in a similar way. 
Clearly, for each possible triple
$(\sigma_1,\sigma_2,\sigma_3)\in (\Gamma\cup(Q\times \Gamma))^3$, the
\ALC-concept $C_{\sigma_1,\sigma_2,\sigma_3}=A_{\sigma_1}\sqcap
\exists s.(A_{\sigma_2}\sqcap \exists s.A_{\sigma_3})$ is true at an
element $a$ of the computation tree iff $a$ is labeled with
$A_{\sigma_1}$, an $s$-successors $b$ of $a$ is labeled with
$A_{\sigma_2}$, and an $s$-successors $c$ of $b$ is labeled with
$A_{\sigma_3}$. In each configuration, we synchronize elements
with $V$-counter $0$ by including for every $(\sigma_1,\sigma_2,\sigma_3)$ and $i\in\{1,2\}$ the following
concept inclusions: 
\begin{align*}
  (V=2^n-1) \sqcap (U<2^n-2) \sqcap C_{\sigma_1,\sigma_2,\sigma_3}
  \sqcap B_{\forall} & \sqsubseteq
  \forall s. A^1_{f_1(\sigma_1,\sigma_2,\sigma_3)} \ \sqcap \\
  & \ \ \ \ \forall
  s.A^2_{f_2(\sigma_1,\sigma_2,\sigma_3)}\\
  (V=2^n-1) \sqcap (U<2^n-2) \sqcap C_{\sigma_1,\sigma_2,\sigma_3} \sqcap  B_\exists^i
  & \sqsubseteq \forall s. A^i_{f_i(\sigma_1,\sigma_2,\sigma_3)}
  %
\end{align*}
 At this point, the importance of the superscript in
$B_\exists^\ast$ becomes apparent: since different cells of a
configuration are synchronized
in different trees $T_k$ the superscript makes sure that all trees
rely on the same choice for existential configurations. The concept names $A^i_{\sigma}$ are used as markers (not in $\Sigma$) and
are propagated along $s$ for $2^n$ steps, exploiting the $V$-counter.
The superscript $i\in\{1,2\}$ determines the successor configuration
that the symbol is referring to. After crossing the end of a
configuration, the symbol $\sigma$ is propagated using concept names
$A_{\sigma}'$ (the superscript is not needed anymore because the
branching happens at the end of the configuration, based on $Z$).
\begin{align*}
  (U<2^n-1) \sqcap A_\sigma^i & \sqsubseteq \forall s. A_{\sigma}^i \\
  (U=2^n-1) \sqcap B_\forall \sqcap A_\sigma^1 & \sqsubseteq \forall
  s.(Z\to A'_\sigma)\\
  (U=2^n-1) \sqcap B_\forall \sqcap A_\sigma^2 & \sqsubseteq \forall
  s.(\neg Z\to A'_\sigma)\\
  (U=2^n-1) \sqcap B_\exists^i \sqcap A_\sigma^i & \sqsubseteq \forall
  s.A'_\sigma && i\in\{1,2\}\\
  (V<2^n-1) \sqcap A'_\sigma & \sqsubseteq \forall s.A'_\sigma \\
  (V=2^n-1) \sqcap A'_\sigma & \sqsubseteq \forall s.A_\sigma
\end{align*}
%
For those $(q,a)$ with $\Delta(q,a)=\emptyset$, we add the concept inclusion
\[A_{q,a} \sqsubseteq \bot. \]
The following lemma establishes correctness of the reduction.

\begin{lemma}\label{lem:sim1}
  The following conditions are equivalent:
  \begin{enumerate} 

    \item $M$ accepts $w$;

    \item there exist models $\Imc$ and $\Jmc$ of $\Omc$ such that
      $\Imc,b^{\Imc}\sim_{\ALCO,\Sigma} \Jmc,d$, for some
      $d\not=b^{\Jmc}$.  

  \end{enumerate} 		

\end{lemma}

\begin{proof} ``1 $\Rightarrow$ 2''. If $M$ accepts $w$, there is a
  computation tree of $M$ on $w$. We construct a single interpretation
  $\Imc$ with $\Imc,b^{\Imc}\sim_{\ALCO,\Sigma} \Imc,d$ for some
  $d\not=b^{\Imc}$ as follows. Let $\widehat\Jmc$ be the infinite tree-shaped
  interpretation that represents the computation tree of $M$ on $w$ as
  described above, that is, configurations are represented by
  sequences of $2^n$ elements linked by role $s$ and labeled by
  $B_\forall,B_\exists^1,B_\exists^2$ depending on whether the
  configuration is universal or existential, and in the latter case
  the superscript indicates which choice has been made for the
  existential state. Finally, the first element of the first successor
  configuration of a universal configuration is labeled with
  $Z$. Observe that $\widehat\Jmc$
  interprets only the symbols in $\Sigma$ as non-empty. Now, we
  obtain interpretations $\Imc_k$, $k<2^n$ from $\widehat\Jmc$ by interpreting
  non-$\Sigma$-symbols as follows: 
  \begin{itemize}

    \item the root of $\Imc_k$ satisfies $I_s$; 

    \item the $U$-counter starts at $0$ at the root and counts modulo
      $2^n$ along each $s$-path;

    \item the $V$-counter starts at $k$ at the root and counts modulo
      $2^n$ along each $s$-path;

    \item the auxiliary concept names of the shape $A_\sigma^i$ and
      $A_\sigma'$ are interpreted in a minimal way so as to satisfy
      the concept inclusions starting from
      concept inclusion~\eqref{eq:startrolei}. Note that, by definition of
      these concept inclusions, there is a unique result.

  \end{itemize}
  Now obtain $\Imc$ from $\widehat\Jmc$ and the $\Imc_k$ by creating an infinite outgoing
  $r$-path $\rho$ from some element $d\neq b^\Imc$ (with the corresponding $A$-counter)
  and adding $\Imc_k$, $k<2^n$ to
  every element with $A$-counter value $k$ on the $r$-path, identifying the roots of the $\Imc_k$ with
  the element on the path. Additionally, include
  $(b^\Imc,b^\Imc)\in r^\Imc$ and add $\widehat\Jmc$ to \Imc
  by identifying $b^\Imc$ with the root of $\widehat\Jmc$. It should be clear that $\Imc$ is as
  required. In particular, the reflexive, transitive, and symmetric closure of 
  \begin{itemize}

    \item all pairs $(b^\Imc,e)$, with $e$ on $\rho$, and 

    \item all pairs $(e,e')$, with $e$ in $\widehat\Jmc$ and $e'$ a
      copy of $e$ in some tree $\Imc_{k}$

  \end{itemize}
  is an $\ALCO(\Sigma)$-bisimulation $S$ on $\Imc$ with $(b^\Imc,d)\in S$.

  ``2 $\Rightarrow$ 1''.  Assume that
  $\Imc,b^{\Imc}\sim_{\ALCO,\Sigma} \Jmc,d$ for some $d\not=b^{\Jmc}$.
  As argued above, due to the $r$-self loop at $b^\Imc$, from $d$
  there has to be an outgoing infinite $r$-path on which all $s$-trees
  are $\ALCO(\Sigma)$-bisimilar. Since $\Imc$ is a model of $\Omc$,
  all these $s$-trees are additionally labeled with some auxiliary
  concept names not in $\Sigma$, depending on the distance from their
  roots on $\rho$. Using the concept inclusions in \Omc and the
  arguments given in their description, it can be shown that all
  $s$-trees contain a computation tree of $M$ on input $w$ (which is
  solely represented with concept names in $\Sigma$).  
\end{proof}
The same ontology $\Omc$ can be used for the remaining DLs with
nominals.  For $\ALCO^{u}$, exactly the same proof works; in
particular, note that both the bisimulation $S$ constructed in
``$1\Rightarrow 2$'' and its inverse are surjective. For the DLs with inverse
roles the (one-way) infinite $r$-path $\rho$ has to replaced by a two-way
infinite path in ``$1\Rightarrow 2$''.

\medskip
Using the ontology $\Omc$ defined above we define a new ontology $\Omc'$
to obtain the \TwoExpTime lower bound for signatures $\Sigma'=\text{sig}(\Omc')\setminus\{b\}$. Fix a role name $r_{E}$ for any concept name $E\in \text{sig}(\Omc)\setminus \Sigma$.
Now replace in $\Omc$ any occurrence of $E\in \text{sig}(\Omc)\setminus \Sigma$ by $\exists r_{E}.\{b\}$ and denote the resulting ontology 
by $\Omc'$. 
%
\begin{lemma}\label{lem:correctness_nom} The following conditions are equivalent:
  \begin{enumerate}

    \item $M$ accepts $w$;

    \item there exist models $\Imc$ and $\Jmc$ of $\Omc'$ such that
      $\Imc,b^{\Imc}\sim_{\ALCO,\Sigma'} \Jmc,d$, for some $d\not=b^{\Jmc}$.

  \end{enumerate} 		
\end{lemma}

\begin{proof} ``1 $\Rightarrow$ 2''. We modify the interpretation
  \Imc defined in the proof of Lemma~\ref{lem:sim1} in such a way that
  we obtain a model of $\Omc'$ and such that the
  $\ALCO(\Sigma)$-bisimulation $S$ on $\Imc$ defined in that proof is,
  in fact, an $\ALCO(\Sigma')$-bisimulation on the new interpretation.
  Formally, obtain $\Imc'$ from $\Imc$ by interpreting every $r_E$,
  $E\in \text{sig}(\Omc)\setminus\Sigma$ as follows: 
  \begin{enumerate}[label=(\roman*)]

    \item there is an $r_{E}$-edge from $e$ to $b^{\Imc}$, for all
      $e\in E^{\Imc}$;

    \item there is an $r_{E}$-edge from $e$ to all elements on the path
      $\rho$, for all $(e,e')\in S$ and $e'\in E^{\Imc}$;  

    \item there are no more $r_E$-edges.



  \end{enumerate} 
  Note that, by~(i), $\Imc'$ is a model of $\Omc'$. By~(ii), the
  relation $S$ defined in the proof of Lemma~\ref{lem:sim1} is an
  $\ALCO(\Sigma')$-bisimulation. In particular, by~(i), elements $e'\in
  E^\Imc$ have now an $r_E$-edge to $b^\Imc$, so any element $e$ bisimilar to
  $e'$, that is, $(e,e')\in S$, needs an $r_E$-successor to some
  element bisimilar to $b^\Imc$. Since all elements on the path
  $\rho$ are bisimilar to $b^\Imc$, these $r_E$-successors exist due
  to~(ii).

  \medskip ``2 $\Rightarrow$ 1''. This direction remains the same as
  in the proof of Lemma~\ref{lem:sim1}.
\end{proof} 
The extension to DLs with inverse roles and the universal role and
the restriction to a single role name are again straightforward. 

\medskip We conclude the section with an observation that will be
relevant for the application of our results to modal logic in
Section~\ref{sec:modal}. More specificially, we strengthen the lower
bound for the case of $\Lmc=\ALCO^u$ as follows: 
\begin{lemma} \label{lem:lower-single-role}
  Let $\Lmc\in\{\ALCO,\ALCO^u\}$. Then, it is \TwoExpTime-hard to
  decide for an $\Lmc$-ontology $\Omc$,
  individual name $b$, and signature $\Sigma\subseteq
  \text{sig}(\Omc)\setminus \{b\}$ whether $\{b\}$ and
  $\neg\{b\}$ are jointly consistent under \Omc modulo 
  $\Lmc(\Sigma)$-bisimulations, even if $\Omc$ is allowed to use only
  a single role name. 
\end{lemma} 
\begin{proof}
We modify the ontology \Omc and signature $\Sigma$ used in the proof of Lemma~\ref{lem:lowernom}. Let $\Omc'$ be the ontology
obtained from $\Omc$ by: 
\begin{itemize}

  \item replacing every subconcept of the shape $\exists r.C$ with
    $\exists r.(X_r\sqcap C)$ and

  \item replacing every subconcept of the shape $\exists s.C$ with
    $\exists r.(X_s\sqcap C)$,

\end{itemize}
for fresh concept names $X_r,X_s$, and set
$\Sigma'=\Sigma\cup\{X_r,X_s\}$. It is routine to verify that
Lemma~\ref{lem:sim1} holds for $\Omc',\Sigma'$ instead of $\Omc,\Sigma$. In
particular, we can obtain an interpretation $\Imc'$ from \Imc as
constructed in ``$1\Rightarrow 2$'' as follows.
\begin{itemize}

  \item replace all $s$-connections by $r$-connections; 

  \item every element that has an $s$-predecessor in \Imc satisfies
    $X_s$ in $\Imc'$, that is, $X_s^{\Imc'}=(\exists s^-.\top)^\Imc$;

  \item $b^\Imc$ and every element on the infinite $r$-path $\rho$ in
    \Imc satisfy
    $X_r$ in $\Imc'$, that is, $X_r^{\Imc'}=(\exists
    r.\top)^\Imc$ (the root of the infinite path has to satisfy
    $X_r$ since it is bisimilar to $b^\Imc$ which satisfies $X_r$).

\end{itemize}
\end{proof}

\subsection{Shape of Explicit Definitions in the Lower Bound}\label{sec:shapeofinterpolant}

The goal of this subsection is to provide some intuition on the shape
of the explicit definitions that arise in the proof of
Lemma~\ref{lem:lowernom}. We note first that $r(x,x)$ is an explicit
FO$(\Sigma)$-definition of $\{b\}$ under $\Omc$, regardless of whether the ATM
accepts its input or not. This means that interpolant and explicit
definition existence is \TwoExpTime-hard even under the promise that a
fixed FO-definition / FO-interpolant exists.

We now analyze the $\ALCO(\Sigma)$-definitions that arise in the proof
of Lemma~\ref{lem:lowernom}. Recall that such a definition exists iff
the ATM $M$ does not accept its input $w$. So, for the rest of the
dicussion we assume the latter. Instead of directly
providing an explicit $\ALCO(\Sigma)$-definition of $\{b\}$, we give a
definition $C_{\neg b}$ of $\neg \{b\}$, since the definition of
$C_{\neg b}$ is close to the intuitions provided in the proof of
Lemma~\ref{lem:lowernom}.  Obviously, $\neg C_{\neg b}$ will be the
desired definition of $\{b\}$. Let $n$ be the length of the input
word $w$ and let $k=|\Gamma\cup (Q\times \Gamma)|$ be the number of
possible labelings of a cell in some configuration of the ATM.
Moreover, set $K=k^{2^n}+2^n$. 

The concept $C_{\neg b}$ takes the shape 
\[C_{\neg b} = \exists r.\top \to \big(C_\text{tree}\sqcap
C_{\text{start}}\sqcap \neg C_{\text{stop}}\sqcap
\bigsqcup_{i=0}^{2^n-1}C_i\big).\]
To understand the structure $\exists r.\top\to C'$ of $C_{\neg b}$,
recall that the proof of Lemma~\ref{lem:lowernom} relies on the
assumption that an element $d\neq b^\Imc$ has an $r$-successor. The
concepts 
$C_{\text{tree}},C_{\text{stop}},C_{\text{start}}, C_i$ provide an
``approximation'' of an accepting computation tree of the ATM $M$ on
its input $w$ in the following sense. (Note that the definition of $\neg\{b\}$ cannot
describe the full accepting computation since it is not entailed).

The concept $C_{\text{tree}}$ enforces an $s$-tree of depth $K$ that
acts as the skeleton for encoding (an initial fragment of) a
computation tree. It is labeled with concepts
$Z,B_{\forall},B_{\exists}^1,B_{\exists}^2$ in the expected way.
Formally, $C_{\text{tree}}$ is 
\begin{align*}
  \text{Path}_{s,B_\forall}^{2^n} 
  & \sqcap
  \bigsqcap_{\substack{i=\ell\cdot 2^n-1\\i<K}} \forall s^i.\bigl(
B_{\forall}\to \bigl(\exists s. (Z\sqcap
\text{Path}_{s,B_{\exists}^1}^{2^n}) \sqcap \exists s. (\neg Z\sqcap
\text{Path}_{s,B_{\exists}^2}^{2^n})\bigr) \\
& \sqcap
(B_{\exists}^1\sqcup
B_{\exists}^2)\to
\exists s.\text{Path}_{s,B_{\forall}}^{2^n}\bigr)
\end{align*}
where $\text{Path}_{s,X}^{m}$ is a concept that enforces an
$s$-path of length $m$ with each element labeled with $X$.
We refrain from giving the precise definitions of the remaining
concepts, and rather provide the intuitions.
$C_{\text{start}}$ is a concept that enforces the initial
configuration to be true in the computation tree, and 
$C_{\text{stop}}$ is a concept that is true if some element within
$K$ $s$-steps is labeled with a concept name $A_{q,a}$ for which
$\Delta(q,a)=\emptyset$.  Moreover, each $C_i$ is a concept with
$\Omc\models I_s\sqcap (A=i)\sqsubseteq C_i$; recall that we denote
with $(A=i)$ that the $A$-counter has value $i$. The disjunction over
all possible $C_i$ in $C_{\neg b}$ is needed since the $A$-counter can
take any value between $0$ and $2^n-1$ at a given element in $d\neq b^\Imc.$ More precisely, each
$C_i$ is a conjunction 
\[C_i = \bigsqcap_{j=0}^{2^n-1}\forall
r^j.C_\text{sync}^{i\oplus_{2^n}j},\] 
where $\oplus_{m}$ denotes
addition modulo $m$, and for each $m$ with $0\leq m< 2^n$,
$C_\text{sync}^m$
is a concept that coordinates the content of the $m$-th cell in every
configuration in the computation tree with the same cell in the
successor configuration(s). This can be easily realized using
value restrictions $\forall s$. 

Observe that $\Omc\models\neg\{b\}\sqsubseteq C_{\neg b}$ regardless
of whether the ATM accepts $w$ or not.
In particular, in every model of \Omc, each element $d$ satisfying
$\neg \{b\}\sqcap \exists r.\top$ satisfies the concepts $C_{\text{tree}}, C_{\text{start}}$, and $\neg C_{\text{stop}}$. Moreover, $d$ satisfies
$I_s$ and $(A=i)$ for some $i$, and thus $d$ also satisfies $C_i$.

For the converse, $\Omc\models C_{\neg b}\sqsubseteq \neg \{ b\}$,
suppose that $C_{\neg b}$ is realizable in a model \Imc of \Omc
in an element $d$ with $(d,d)\in r^\Imc$. We thus also have $d\in
(C_{\text{tree}}\sqcap C_{\text{start}}\sqcap \neg
C_{\text{stop}})^\Imc$, and $d\in C_i^\Imc$, for some $i$.
Due to the $r$-self loop,
$d\in (C_\text{sync}^m)^\Imc$, for all $m$ with $0\leq m<2^n$. But this means that
at $d$ starts the initial segment of a computation tree of
$M$ which is not
labeled with a halting configuration, and all of whose cells are
coordinated with the corresponding cell of the successor
configuration(s). By the
choice of $K$, on every path there is a configuration that occurs
twice. We can thus extend the initial fragment of the computation tree
to an infinite computation tree for the word $w$, in contradiction
to the fact that $M$ does not accept $w$.

We conclude with observing that the size of the definition $C_{\neg
b}$ of $\neg\{b\}$ is double exponential in the length $n$ of the
input word, due to the depth $K$ of the enforced tree. This is in
stark contrast with the (constant!) size of the
FO$(\Sigma)$-definition. We conjecture that one can enforce explicit
definitions of triple exponential size. For example, when using two
roles $s_1,s_2$ instead of $s$ for encoding the computation tree,
already the concept $C_{\text{tree}}$ will be of triple exponential
size. We leave a detailed analysis for future work.

\subsection{DLs with Role Inclusions}\label{sec:lower-role-inclusions}

By Theorem~\ref{thm:critdef}, it suffices to prove the following.
\begin{lemma}\label{lem:lowerrolei}
  Let $\mathcal{L}\in \Dln$ admit role inclusions.  It is
  \TwoExpTime-hard to decide for an $\Lmc$-ontology $\Omc$, concept
  $C$, and signature $\Sigma\subseteq \text{sig}(\Omc)$ whether
  $C$ and $\neg C$ are jointly consistent under \Omc modulo
  $\Lmc(\Sigma)$-bisimulations.  
\end{lemma}
As in the proof of Lemma~\ref{lem:lowernom}, we reduce the word
problem for exponentially space bounded ATMs, so let $M$ be a
$2^n$-space bounded ATM and $w=a_0\ldots a_{n-1}$ an input of length
$n$. In fact, the only
difference to the proof of Lemma~\ref{lem:lowernom} is the way in
which we enforce that exponentially many elements are
$\Lmc(\Sigma)$-bisimilar.  We first provide the reduction for
$\Lmc=\mathcal{ALCH}$ and 
\[\Sigma  =
\{r_{1},r_{2},s,Z,B_{\forall},B^{1}_{\exists},B^{2}_{\exists}\}\cup
\{A_\sigma\mid \sigma\in \Gamma\cup (Q\times \Gamma)\}.\]
The symbols $s,Z,B_{\forall},B^{1}_{\exists},B^{2}_{\exists}$ and
$A_\sigma$, $\sigma\in \Gamma\cup (Q\times \Gamma)$, play exactly the
same role as above. The main difference is that we replace the nominal
$b$ by an $r$-chain of length $n$. The ontology $\Omc$ contains the
RIs $r\sqsubseteq r_{1},r\sqsubseteq r_{2}$ and the CI $\neg \exists
r^{n}.\top \sqcap \exists r_{1}^{n}.\top \sqsubseteq R$. As usual
$\exists r^n$ abbreviates a sequence of $n$ times $\exists r$. 

To see how we use these inclusions, suppose there exist models $\Imc$ and $\Jmc$ of $\Omc$ 
and $d\in \Delta^{\Imc}$, $e\in \Delta^{\Jmc}$ such
that 
\begin{itemize}

  \item $d\in (\exists r^{n}.\top)^{\Imc}$;

  \item $e\in (\neg \exists r^{n}.\top)^{\Jmc}$;

  \item $\Imc,d \sim_{\mathcal{ALCH},\Sigma} \Jmc,e$.

\end{itemize}
then it follows that $e\in R^{\Jmc}$: due to $\Imc,d
\sim_{\mathcal{ALC},\Sigma} \Jmc,e$ and $d\in (\exists
r_1^n.\top)^\Imc$, we also have $e\in (\exists r_1^n.\top)^\Imc$. Let
now $d'$ be an element reachable from $d$ via an $r$-path of length
$n$ (which exists due to $d\in (\exists r^n.\top)^\Imc$). Since $r\sqsubseteq r_{i}$ for $i=1,2$, there are also arbitrary
$r_1/r_2$-paths of length $n$ from $d$ to $d'$. Since $\Imc,d
\sim_{\mathcal{ALC},\Sigma} \Jmc,e$, there are also arbitrary
$r_1/r_2$-paths of length $n$ starting in $e$ and whose end points are all
$\ALCH(\Sigma)$-bisimilar to $d'$ and thus also mutually
$\ALCH(\Sigma)$-bisimilar. The concept name $R$ will enforce that 
\begin{itemize}

  \item[$(\ast)$] the end point of any $r_1/r_2$-path of length $n$
    starting in $e$ carries a counter value that describes the
    path in a canonical way.

\end{itemize}
We can thus use these $2^n$ different, but bisimilar end points to
start the infinite trees which mimick the computation tree of $M$ as
in the proof of Lemma~\ref{lem:lowernom}. Along these we maintain the same two counters
as there:

\begin{itemize}

  \item one counter starting at $0$ and counting modulo $2^n$ to
    divide the tree into configurations of length $2^n$; 

  \item another counter starting at \emph{the value of the counter on
    the leaf} and also counting modulo $2^{n}$.

\end{itemize}
Formally, the ontology $\Omc$ is constructed as follows. In order to
realize~$(\ast)$ above, we use concept names $A_i$, $0\leq i<n$
realizing the counter and the following concept inclusions: 
\begin{align*}
  R & \sqsubseteq R_{0} \\
  R_{i} & \sqsubseteq \forall r_{1}.(A_{i}\sqcap R_{i+1}) \sqcap
  \forall r_{2}.(\neg A_{i}\sqcap R_{i+1}) & i<n \\
  R_i \sqcap A_j & \sqsubseteq \forall r_1.A_j \sqcap \forall r_2.A_j
  & 0\leq j < i<n
  \\
  R_i \sqcap \neg A_j & \sqsubseteq \forall r_1.\neg A_j \sqcap
  \forall r_2.\neg A_j & 0\leq j < i<n \\
  R_{n} & \sqsubseteq L_{R}
\end{align*}
Using the concept name $L_{R}$, we start the $s$-trees with two counters, realized using
concept names $U_i$ and $V_i$, $0\leq i<n$, and
initialized to $0$ and the value of the $A$-counter, respectively:
\begin{align*}
  L_{R} & \sqsubseteq (U=0) \\
  L_{R} & \sqsubseteq A_j\leftrightarrow V_j &0\leq j<n\\
%
%
  \top & \sqsubseteq \exists s.\top
\end{align*}
The structure of the computation tree, the initial configuration, and
the coordination between consecutive configurations is done using the
same concept inclusions as in the proof of Lemma~\ref{lem:lowernom},
starting from inclusion~\eqref{eq:startrolei} and replacing $I_s$ with
$L_R$. We can then prove the following very similarly to
Lemma~\ref{lem:sim1}.

\begin{lemma}\label{lem:sim2}
  The following conditions are equivalent:
  \begin{enumerate} 

    \item $M$ accepts $w$;

    \item there exist models $\Imc$ and $\Jmc$ of $\Omc$ such that
      $\Imc,d\sim_{\mathcal{ALCH},\Sigma} \Jmc,e$, for some $d\in
      (\exists r^{n}.\top)^{\Imc}$ and $e\not\in (\exists
      r^{n}.\top)^{\Jmc}$.  

  \end{enumerate} 		

\end{lemma}

\begin{proof} ``1 $\Rightarrow$ 2''. If $M$ accepts $w$, there is a
  computation tree of $M$ on $w$. We construct a single interpretation
  $\Imc$ with $\Imc,d\sim_{\ALCH,\Sigma} \Imc,e$ for some
  $d,e$ with $d\in (\exists r^n.\top)^\Imc$ and $e\notin (\exists
  r^n.\top)^\Imc$ as follows. Let $\widehat\Jmc$ be the infinite tree-shaped
  interpretation that represents the computation tree of $M$ on $w$ as
  described above, that is, configurations are represented by
  sequences of $2^n$ elements linked by role $s$ and labeled by
  $B_\forall,B_\exists^1,B_\exists^2$ depending on whether the
  configuration is universal or existential, and in the latter case
  the superscript indicates which choice has been made for the
  existential state. Finally, the first element of the first successor
  configuration of a universal configuration is labeled with
  $Z$. Observe that $\widehat\Jmc$
  interprets only the symbols in $\Sigma$ as non-empty. Now, we
  obtain interpretations  $\Imc_k$, $k<2^n$ from $\widehat\Jmc$ by interpreting
  non-$\Sigma$-symbols as follows: 
  \begin{itemize}

    \item the root of $\Imc_k$ satisfies $L_R$; 

    \item the $U$-counter starts at $0$ at the root and counts modulo
      $2^n$ along each $s$-path;

    \item the $V$-counter starts at $k$ at the root and counts modulo
      $2^n$ along each $s$-path;

    \item the auxiliary concept names of the shape $A_\sigma^i$ and
      $A_\sigma'$ are interpreted in a minimal way so as to satisfy the
      concept inclusions that enforce the coordination between
      consecutive configurations (cf. the concept inclusions in proof
      of Lemma~\ref{lem:lowernom}). 

  \end{itemize}
  Now obtain $\Imc$ from $\widehat\Jmc$ and the $\Imc_k$ as follows:
  First, create a path of length $n$ from some element $d$ so that
  consecutive elements are connected with $r,r_1,r_2$, and
  identify the end point of the path with the root of $\widehat\Jmc$. Then
  create a binary tree of depth $n$, rooted in $e$, in which left children are always
  $r_1$-successors and right children are always $r_2$-successors.
  Label the nodes of the tree with $R_i$ and $A_j$ as described above and 
  identify the leaf having $A$-counter value $k$ with the root of
  $\Imc_k$, for all $k<2^n$. $\Imc$ is as
  required since, by construction, $d\in (\exists r^n.\top)^\Imc$,
  $e\notin (\exists r^n.\top)^\Imc$, and 
  the reflexive, transitive, and symmetric closure of 
  \begin{itemize}

    \item all pairs $(d',e')$ such that $d'$ has distance
      $\ell\leq n$ from $d$ and $e'$ has distance $\ell$ from
      $e$, and 

    \item all pairs $(b,b')$, with $b$ in $\widehat\Jmc$ and
      $b'$ a copy of $b$ in some tree $\Imc_{k}$

  \end{itemize}
  is an $\ALCH(\Sigma)$-bisimulation $S$ on $\Imc$ with $(d,e)\in S$.

  ``2 $\Rightarrow$ 1''.  Assume that $\Imc,d\sim_{\ALCH,\Sigma} \Jmc,e$
  for models $\Imc,\Jmc$ of \Omc and 
  some $d,e$ with $d\in (\exists r^n.\top)^\Imc$ and $e\notin
  (\exists r^n.\top)^\Jmc$. As argued above, there are $r_1/r_2$-paths
  of length $n$ whose end points carry all possible counters $<2^n$ and are
  all $\ALCH(\Sigma)$-bisimilar. 
  In addition, all these end points root $s$-trees which are
  $\ALCH(\Sigma)$-bisimilar. Since $\Jmc$ is a model of $\Omc$, all
  these $s$-trees are additionally labeled with some auxiliary concept
  names not in $\Sigma$, depending on the value of the $A$-counter of
  the corresponding leaf. Using the concept inclusions in \Omc and the
  arguments given in their description, it can be shown that all
  $s$-trees contain a computation tree of $M$ on input $w$ (which is
  solely represented with concept names in $\Sigma$).  
\end{proof}

The same proof works as well for $\ALCH^u$ as the relation $S$
constructed in the direction ``$1\Rightarrow 2$'' above is actually an
$\ALCH^u(\Sigma)$-bisimulation. For
$\Lmc\in\{\mathcal{ALCHI},\mathcal{ALCHI}^u\}$, we have to slightly
adapt the model construction in ``$1\Rightarrow 2$'', following the
idea provided in Example~\ref{ex:4} (except that we do not need to
take the union of $\Imc_1,\Imc_2$ here, since we construct a single
interpretation $\Imc=\Imc_1=\Imc_2$).  Let $d_0,\ldots,d_n$ be the
elements on the $r$-path that starts in $d$, that is, $d_0=d$ and
$d_\ell$ has distance $\ell$ from $d$.  Recall that $(d_\ell,e')\in S$
for every element $e'$ in level $\ell$ in the binary tree rooted at
$e$. Observe that $S$ is not an $\Lmc(\Sigma)$-bisimulation since, for
$\ell>0$, $d_\ell$ has both an $r_1$ and an $r_2$-predecessor (both
are $d_{\ell-1}$), but elements in the binary tree lack either an
$r_1$- or an $r_2$-predecessor. To repair this, we add for every
element $e'$ in level $\ell>0$ in the binary tree the following
connections: 
\[ (d_{\ell-1},e')\in r_1^\Imc\quad\text{ and }\quad(d_{\ell-1},e')\in
r_2^\Imc.\]
It can be verified that the modified interpretation is still a model
of \Omc, and that $S$ is an $\Lmc(\Sigma)$-bisimulation as required. 

We conclude the section by remarking that one can analyze the
structure of the explicit $\ALCH(\Sigma)$-definitions that arise in
the proof of Lemma~\ref{lem:lowerrolei} along the lines of
Section~\ref{sec:shapeofinterpolant}. In contrast to that section,
the size of the FO-definition \[\varphi(x_1) = \exists
  x_2\ldots\exists x_n. \bigwedge_{i=1}^{n-1}r_1(x_i,x_{i+1})\wedge
r_2(x_{i},x_{i+1})\]
of $\exists r^n.\top$ under \Omc is not constant, but depends on $n$.

\subsection{CI-interpolant Existence}\label{sec:lowerci}

We show the \TwoExpTime lower bound for CI-interpolant existence
stated in Theorem~\ref{thm:ciinterpolants}.  We employ the ontology
$\Omc$, individual $b$, and signature $\Sigma$ constructed in the
proof of Lemma~\ref{lem:sim1} and remind the reader that the claim of
Lemma~\ref{lem:sim1} holds also for $\mathcal{ALCO}^{u}$ and
$\mathcal{ALCOI}^{u}$. Let $\Omc_1$ be defined as $\Omc$ without
$\{b\} \sqsubseteq \exists r.\{b\}$and with $\neg \{b\} \sqcap \exists
r.\top \sqsubseteq I_{s}$ replaced by $\exists r.\top \sqsubseteq
I_{s}$. Also define $\Omc_{2}$ as $\Omc$ without $\{b\} \sqsubseteq
\exists r.\{b\}$ and with all concept and role names not in $\Sigma$
replaced by fresh symbols. Transform $\Omc_{2}$ into an equivalent
ontology of the form $\{\top\sqsubseteq D\}$.  Observe that $\Omc_{1}$
does not use $b$. In fact, the shared symbols of $\Omc_{1}$ and the CI
$\forall u.D \sqcap \{b\} \sqsubseteq \neg \exists r.\{b\}$ are
exactly the symbols in $\Sigma$. 
The \TwoExpTime lower bound now follows from the following lemma.

\begin{lemma}\label{lem:twoexp11} Let $\Lmc\in \{\mathcal{ALCO}^{u},\mathcal{ALCOI}^{u}\}$. Then the following conditions are equivalent:
	\begin{enumerate} 
		\item Point~2 of Lemma~\ref{lem:sim1} holds; that is, there exist models $\Imc$ and $\Jmc$ of $\Omc$ such that $\Imc,b^{\Imc}\sim_{\Lmc,\Sigma} \Jmc,d$, for some
		$d\not=b^{\Jmc}$;  
		\item there does not exist an $\Lmc$-CI interpolant for $\Omc_{1}$ and $\forall u.D \sqcap \{b\} \sqsubseteq \neg \exists r.\{b\}$.
	\end{enumerate}
\end{lemma}
\begin{proof}
	Assume Point~(1) holds and take $\Imc,\Jmc$, and $d$ witnesssing this. We may assume that
	$\Imc=\Jmc$ is the interpretation constructed in the proof of ``1 $\Rightarrow$ 2'' of Lemma~\ref{lem:sim1}. Assume for a proof by contradiction that $\Omc'$
	is an $\Lmc$-CI interpolant for $\Omc_{1}$ and $\forall u.D \sqcap \{b\} \sqsubseteq 
	\neg \exists r.\{b\}$. Let $\Imc'$ denote the restriction of $\Imc$ to elements that cannot be reached from $b$ along a path following $r^{\Imc}$ or $s^{\Imc}$ and reinterpret $b$ as an element of $\Delta^{\Imc'}$. Then $\Imc'$ is a model of $\Omc_{1}$ by the definition of $\Imc$ and since $\Omc_{1}$ does not contain any CIs with the individual $b$ . Moreover, we have $\Imc,b^{\Imc}\sim_{\Lmc,\Sigma} \Imc',d$ since $b\not\in\Sigma$. Then, as $\Lmc$ admits the universal role, 
	$\Imc$ is a model of $\Omc'$. We now reinterpret in $\Imc$ the fresh concept and role names in $\Omc_{2}$ in the same way as the original ones in $\Imc$ and obtain a model $\Imc''$ with $\Delta^{\Imc}=D^{\Imc''}$ since $\Imc$ is a model of $\Omc$. But then $\Imc''\not\models \forall u.D \sqcap \{b\} \sqsubseteq \neg \exists r.\{b\}$ and so
	(as $\Imc''$ is still a model of $\Omc'$ since $\text{sig}(\Omc')\subseteq \Sigma$) $\Omc'\not\models \forall u.D \sqcap \{b\} \sqsubseteq \neg \exists r.\{b\}$, a contradiction.
	
	Conversely, assume there does not exist an $\Lmc$-CI interpolant for $\Omc_{1}$ and $\forall u.D \sqcap \{b\} \sqsubseteq \neg \exists r.\{b\}$. As in the proof of Theorem~\ref{thm:critint} and using the fact that $\Lmc$ admits the universal role, we obtain a model $\Jmc$ of $\Omc_{1}$, $d\in \Delta^{\Jmc}$, and an interpretation
	$\Imc$ with $\Imc\not\models \forall u.D \sqcap \{b\} \sqsubseteq \neg \exists r.\{b\}$ and  $\Imc,b^{\Imc}\sim_{\Lmc,\Sigma} \Jmc,d$. We may assume that $\Imc$ and $\Jmc$ are disjoint. Observe that $\Jmc$ satisfies all CIs in $\Omc$ with the exception of $\{b\} \sqsubseteq \exists r.\{b\}$. By reinterpreting in $\Imc$ the original concept and role names in $\Omc$ in the same way as the fresh concept and role names in $\Omc_{2}$, we obtain a model $\Imc'$ of $\Omc$. Take the union $\Imc'\cup \Jmc$ of $\Imc'$ and $\Jmc$ with $b^{\Imc'\cup \Jmc}$ defined as $b^{\Imc}$. Then $\Imc' \cup \Jmc$ is a model of $\Omc$ such that $\Imc'\cup \Jmc,b^{\Imc'\cup \Jmc}\sim_{\Lmc,\Sigma} \Imc'\cup \Jmc,d$, for some
	$d\not=b^{\Jmc}$, as required for Point~(1).
\end{proof}


\section{Upper Bound Proofs without Ontology}\label{sec:complexity}

The upper bound for Points~1 and~2  of Theorem~\ref{thm:main2} is a consequence of
the respective upper bounds in Theorem~\ref{thm:main1}. For
showing the upper bounds of Points~3 and~4 in Theorem~\ref{thm:main2}, we prove
that joint consistency is in \NExpTime and then apply Theorem~\ref{thm:critint}. 
Indeed, the \NExpTime upper bound follows directly
from the following exponential size witness model property.

\begin{lemma}\label{lem:witness}
  Let $\Lmc\in \Dln$ admit neither the universal role nor both inverse
  roles and nominals simultaneously. Let $\Omc$ be a set of RIs,
  $C_{1},C_{2}$ $\Lmc$-concepts, and $\Sigma$ a signature. If $C_{1}$
  and $C_{2}$ are jointly consistent under $\Omc$ modulo
  $\Lmc(\Sigma)$-bisimulations, then there exist pointed interpretations $\Imc_1,d_1$ and $\Imc_2,d_2$ with $\Imc_{1},\Imc_{2}$ models of 
  $\Omc$ and of at most exponential size in $||O||+||C_1||+||C_2||$ such that $d_1\in
  C_{1}^{\Imc_1}$, $d_2\in C_{2}^{\Imc_2}$, and $\Imc_1,d_1\sim_{\Lmc,
  \Sigma} \Imc_2,d_2$.
%
\end{lemma}  
Before we prove Lemma~\ref{lem:witness}, we introduce some notation.
The \emph{depth} of a concept $C$ is the number of nestings of
existential restrictions in $C$.  For instance, a concept name has
depth $0$ and $\exists r.\exists r.B$ has depth $2$.
Given the ontology $\Omc$, concepts $C_{1}, C_{2}$, and the signature $\Sigma$, we use the notation introduced in Section~\ref{sec:upperbound}. For instance, the set of concepts $\Xi$, $\Xi$-types $t$, and mosaics $(T_{1},T_{2})$ are defined as in Section~\ref{sec:upperbound}.
While in Section~\ref{sec:upperbound} we used the relation $\rightsquigarrow_{r}$ between mosaics to guide the construction of interpretations, here we use a relation between mosaics that is directly induced by interpretations. Assume interpretations $\Imc_{1}$ and $\Imc_{2}$ are given. Consider mosaics $p=(T_{1}(d),T_{2}(d))$ and $q=(T_{1}(d'),T_{2}(d'))$ 
such that there exists a role name $r\in \Sigma$ with $(d,d')\in r^{\Imc_{i}}$,
for some $i\in \{1,2\}$. Then define, for every role name $s$
    and $i \in \{ 1, 2 \}$,
    relations $R_{p,q}^{s,i}\subseteq T_{i}(d)\times T_{i}(d')$ by setting
    $(t,t')\in R_{p,q}^{s,i}$ if there exist $e$ and $e'$ realizing $t$ and $t'$, respectively, with $(T_{1}(e),T_{2}(e))= p$ and $(T_{1}(e'),T_{2}(e'))=q$, such that $(e,e')\in s^{\Imc_{i}}$. 

\medskip

Now assume that $C_{1}$ and $C_{2}$ are jointly consistent under
$\Omc$ modulo $\Lmc(\Sigma)$-bisimulations. By definition, there exist
pointed models $\Imc_{1},d_{1}$ and $\Imc_{2},d_{2}$ of $\Omc$ such
that $d_{1}\in C_{1}^{\Imc_{1}}$, $d_{2}\in C_{2}^{\Imc_{2}}$, and
$\Imc_{1},d_{1}\sim_{\Lmc, \Sigma} \Imc_{2},d_{2}$. Let $k$ be the
maximum depth of $C_{1}, C_{2}$.

%
%
%
%
%

We start with the case involving nominals and without inverse roles.
We construct exponential size $\Jmc_{1},\Jmc_{2}$ with the same
properties as $\Imc_{1},\Imc_{2}$ above. Intuitively, $\Jmc_i$ is
obtained via a suitable unraveling operation up to the depth
$k$ of the concepts $C_1,C_2$; during the unraveling, we take
care of the nominals and, moreover, restrict the outdegree of the
produced interpretation by keeping only necessary successors. 
Formally, let $\mathcal{B}$ be 
some minimal set of mosaics defined by $\Imc_{1},\Imc_{2}$ such that 
\begin{itemize}


  \item $(T_{1}(d_{1}),T_{2}(d_1))\in \Bmc$;

  \item \Bmc contains every mosaic generated by some nominal, or
    formally, 
    $(T_1(d),T_2(d))\in \Bmc$ for every $d\in \Delta^{\Imc_i}$ such that
    $d=a^{\Imc_i}$ for some nominal $a\in \text{sig}(C_i)$;

  \item for every type $t$ realized in $\Imc_{i}$ there exists
     $(T_{1},T_{2})\in \mathcal{B}$ with $t\in T_{i}$.

\end{itemize}  
%
Intuitively, \Bmc serves to describe the behavior of the root of 
the unraveling (first item), of the nominals (second item), and of
potential witnesses for existential restrictions for
non-$\Sigma$-roles (third item). Observe that the size of $\mathcal{B}$ is
at most exponential in the size of $\mathcal{O},C_{1},C_{2}$. 
To restrict the
outdegree, select, for any mosaic $p=(T_{1},T_{2})$ defined by
$\Imc_{1},\Imc_{2}$ and any $\exists s.C\in t\in T_{i}$ such that
there exists $r\in \Sigma$ with $\Omc\models s \sqsubseteq r$, a
mosaic $q=(T_{1}',T_{2}')$ such that $(t,t')\in R_{p,q}^{s,i}$ and $C
\in t'$, and denote the resulting set by $\mathcal{S}(p)$.
%
	Form the set $\mathcal{T}$ of sequences 
	\[
	\sigma = p_{0}\cdots p_{j}=(T_{1}^{0},T_{2}^{0})\cdots (T_{1}^{j},T_{2}^{j})
      \]
	with
	$j\leq k$, $p_{0}\in \mathcal{B}$ and $p_{i+1}\in \mathcal{S}(p_{i})$ for $i < j$.
	Let $\text{tail}(\sigma)= p_{j}$ and
	$\text{tail}_{i}(\sigma)= T_{i}^{j}$.
	We next define the domain of $\Jmc_{1}$ and $\Jmc_{2}$ as
	
	\begin{eqnarray*}
		\Delta^{\Jmc_{i}} & = & \{ (t,p) \mid t\in \text{tail}_{i}(p),p\in \mathcal{B}\} \cup\\
		& &      \{ (t,\sigma) \mid \sigma\in \mathcal{T}, t\in \text{tail}_{i}(\sigma), |\sigma|>1, \text{$t$ contains no nominal}\}
	\end{eqnarray*}
	and define the interpretation of individual, concept and role names in $\Jmc_{1},\Jmc_{2}$ in the expected way: 
\begin{itemize}
	\item for any individual name $a$
	and $(T_{1},T_{2})\in \mathcal{B}$ with $\{ a \} \in t \in T_{i}$, we set $a^{\Jmc_{i}}= (t,(T_{1},T_{2}))$;
	
	\item for any concept name $A$, $(t,\sigma) \in A^{\Jmc_{i}}$ iff $A\in t$;
	
	\item for any role name $r$ we let for $\sigma p\in \mathcal{T}$, 
	\begin{itemize}
		\item $((t,\sigma),(t',\sigma p))\in r^{\Jmc_{i}}$ if $(t,t') \in R_{\text{tail}(\sigma),p}^{r,i}$ and $t'$ contains no nominal;
		\item $((t,\sigma),(t',p)) \in r^{\Jmc_{i}}$ if $(t,t') \in R_{\text{tail}(\sigma),p}^{r,i}$ and $t'$ contains a nominal.
	\end{itemize}
Next assume that $\text{tail}(\sigma)=(T_{1},T_{2})$ and $\sigma$ has
length $k$. If $\text{tail}(\sigma')=(T_{1},T_{2})$ for some
$|\sigma'|<k$, then choose as $r$-successors of any element of the
form $(t,\sigma)$ exactly the $r$-successors of $(t,\sigma')$ defined
above. If no such $\sigma'$ exists, then all elements of the form $(t,\text{tail}(\sigma))$ have distance exactly $k$ from the roots (since no nominal occurs in any type in any mosaic in $\sigma$) and no successors are added.

It remains to take care of existential restrictions $\exists r.C$ for the role names $r$
that do not entail any role name in $\Sigma$.
If $\sigma\in \mathcal{T}$, $\exists r.C\in t\in T_{i}$ with $\text{tail}_{i}(\sigma) = T_{i}$
and $\Omc\not \models r \sqsubseteq s$ for any $s\in\Sigma$, we add $((t,\sigma),(t',p))$ to $r^{\Jmc_{i}}$ (and all $s^{\Jmc_{i}}$ with $\Omc\models r\sqsubseteq s$) for some $p=(T_{1}',T_{2}')\in \mathcal{B}$ and $t'\in T_{i}'$ with $C\in t'$ such that there are $e,e'$ realizing $t,t'$ in $\Imc_{i}$ and $(e,e')\in r^{\Imc_{i}}$. 	
\end{itemize}
The following example illustrates the construction of $\Jmc_{1},\Jmc_{2}$ using the interpretations $\Imc_{1},\Imc_{2}$ introduced in Example~\ref{ex:3}.

\begin{example}
\label{ex:5}
	Let $t_{0}= \textnormal{tp}_{\Xi}(\Imc_{1},a^{\Imc_{1}})$,
	$t_{1}=\textnormal{tp}_{\Xi}(\Imc_{2},b^{\Imc_{2}})$, and
	$t_{2}=\textnormal{tp}_{\Xi}(\Imc_{2},d)$. We ignore the types realized by $b^{\Imc_{1}}$ in $\Imc_{1}$ and by $a^{\Imc_{2}}$ in $\Imc_{2}$ as they are not relevant for understanding  the construction. Then only the mosaic $p=(T_{1},T_{2})$ with $T_{1}=\{t_{0}\}$ and $T_{2}=\{t_{1},t_{2}\}$ remains and $\Jmc_{1}$ and $\Jmc_{2}$ are depicted in Figure~\ref{fig:ex5}. 
	\qedex
	
\begin{figure}[th]
\centering
\begin{tikzpicture}
\tikzset{
dot/.style = {draw, fill=black, circle, inner sep=0pt, outer sep=1pt, minimum size=2pt}
}

\draw (-2,0.5) node[label=$\Jmc_{1}$] (J1) {};

\draw (0,0) node[dot, label=south:$(t_{0}\mathit{,}\,p)$] (a) {};


\draw[->, >=stealth]  (a) edge [loop above] node[] {$r$} ();

\draw (8,0.5) node[label=$\Jmc_{2}$] (J2) {};

\draw (4,-0.5) node[dot, label=south:$(t_{1}\mathit{,}\,p)$] (b2) {};

\draw (4,0.5) node[dot, label=north:$(t_{2}\mathit{,}\,p)$] (d1) {};

\draw (6,-0.5) node[dot, label=south:$(t_{2}\mathit{,}\,p)$] (a2) {};

\draw (6,0.5) node[dot, label=east:$(t_{2}\mathit{,}\,p)$] (d2) {};


\draw[->, >=stealth] (b2) -- (d1) node[midway, left] {$r$};

\draw[->, >=stealth] (d1) -- (d2) node[pos=0.5,above] {$r$};

\draw[->, >=stealth] (a2) -- (d2) node[midway, right] {$r$};

\draw[->, >=stealth]  (d2) edge [loop above] node[] {$r$} ();

%
%

\end{tikzpicture}

\caption{Interpretations $\Jmc_{1}$ and $\Jmc_{2}$ illustrating Example~\ref{ex:5}.}
\label{fig:ex5}

\end{figure}
	
\end{example}

We show that $\Jmc_{1},\Jmc_{2}$ are as required. First, for $i \in \{ 1, 2 \}$, $\Jmc_{i} \models \Omc$ follows from the definition of $\Jmc_{i}$ and the fact that $\Imc_{i} \models \Omc$.
Indeed, given $r \sqsubseteq s \in \Omc$, let $((t, \sigma), (t', \sigma')) \in r^{\Jmc_{i}}$. This means that $(t, t') \in R^{r,i}_{\textnormal{tail}(\sigma), \textnormal{tail}(\sigma')}$,
that is, 
 there exist $e,e'$ realizing $t$ and $t'$, respectively, with $(T_{1}(e),T_{2}(e))= \textnormal{tail}(\sigma)$ and $(T_{1}(e'),T_{2}(e'))=\textnormal{tail}(\sigma')$, such that $(e,e')\in r^{\Imc_{i}}$. Since $\Imc_{i} \models \Omc$, we obtain that $(e,e') \in s^{\Imc_{i}}$ as well, and thus $(t, t') \in R^{s,i}_{\textnormal{tail}(\sigma), \textnormal{tail}(\sigma')}$, meaning that $((t, \sigma), (t', \sigma')) \in s^{\Jmc_{i}}$. Hence, $\Jmc_{i} \models r \sqsubseteq s$.

%

We next prove that, for every $(t, \sigma) \in \Delta^{\Jmc_{i}}$ and every concept $C\in \Xi$
of depth $\leq k - | \sigma |$,
\begin{equation*}
\label{eq:connom}
(t,\sigma) \in C^{\Jmc_{i}} \text{ iff } C \in t. {\color{red} }
\end{equation*}
The proof is by induction on the structure of $C$.
%
%
We consider the case $C = \exists r.D$, where $D$ has depth $< k - | \sigma |$.
We can assume that $|\sigma| < k$, since for $|\sigma| = k$ the claim holds trivially.

$(\Rightarrow)$ Let $(t, \sigma) \in \exists r.D^{\Jmc_{i}}$. Then $\exists r.D \in t$ follows by construction of $r^{\Jmc_{i}}$ as we only have $((t,\sigma),(t',\sigma')\in r^{\Jmc_{i}}$ if
there are $e,e'$ realizing $t,t'$ in $\Imc_{i}$ such that $(e,e')\in r^{\Imc_{i}}$.


$(\Leftarrow)$ Let $\text{tail}(\sigma)=p=(T_{1},T_{2})$ and suppose that $\exists r.D \in t \in T_{i}$.
We distinguish two cases.

\begin{itemize}
	\item There exists $s \in \Sigma$ such that $\Omc \models r \sqsubseteq s$.
Then there exists $q=(T_{1}',T_{2}') \in \mathcal{S}(p)$ and $t' \in T'_{i}$ such that $(t, t') \in R^{r,i}_{p,q}$ and $D \in t'$. We distinguish two cases.
\begin{itemize}
	\item $t'$ does not contain nominals. Then we have that $((t, \sigma), (t', \sigma q)) \in r^{\Jmc_{i}}$. By inductive hypothesis, $(t', \sigma q) \in D^{\Jmc_{i}}$, and thus $(t, \sigma) \in \exists r.D^{\Jmc_{i}}$.
	\item $t'$ contains a nominal. Then we have that $((t, \sigma), (t', q)) \in r^{\Jmc_{i}}$. By inductive hypothesis, $(t', q) \in D^{\Jmc_{i}}$, hence $(t, \sigma) \in \exists r.D^{\Jmc_{i}}$.
\end{itemize}
	\item For every $s \in \Sigma$, $\Omc \not \models r \sqsubseteq s$. By definition of $\Jmc_{i}$, we have $((t,\sigma), (t',q)) \in r^{\Jmc_{i}}$, for some $q = (T'_{1},T'_{2}) \in \Bmc$ and $t' \in T'_{i}$ such that $D \in t'$. By inductive hypothesis, $(t', q) \in D^{\Jmc_{i}}$. Thus, $(t,\sigma) \in \exists r.D^{\Jmc_{i}}$.
\end{itemize}
Next observe that the relation
\[
S= \{ ((t,\sigma),(t',\sigma'))\in \Delta^{\Jmc_{1}}\times \Delta^{\Jmc_{2}} \mid \text{tail}(\sigma) = \text{tail}(\sigma')\}
\]
is an $\mathcal{ALCHO}(\Sigma)$-bisimulation.
Indeed, for $((t, \sigma), (t', \sigma')) \in S$, we have the following.
%
\begin{description}
	\item{[AtomC]}
	Let $(t, \sigma) \in A^{\Jmc_{1}}$ and $A\in\Sigma$. By definition of $\Jmc_{1}$, we have that $(t, \sigma) \in A^{\Jmc_{1}}$ iff $A \in t \in \textnormal{tail}_{1}(\sigma)$, and thus $A \in t' \in \textnormal{tail}_{2}(\sigma)=\textnormal{tail}_{2}(\sigma')$, by definition of mosaics.
	But then $(t', \sigma') \in A^{\Jmc_{2}}$. The converse direction is analogous.
	\item{[AtomI]} Let $(t, \sigma) = a^{\Jmc_{1}}$ and $a\in \Sigma$. By definition of $\Jmc_{1}$, $(t, \sigma) = a^{\Jmc_{1}}$ iff $\{a\} \in t \in \textnormal{tail}_{1}(\sigma)$, and thus $\{a\} \in t' \in \textnormal{tail}_{2}(\sigma)=\textnormal{tail}_{2}(\sigma')$, by definition of mosaics. But then $(t', \sigma') = a^{\Jmc_{2}}$.
	\item{[Forth]} Suppose that $((t, \sigma),(\hat{t}, \hat{\sigma})) \in r^{\Jmc_{1}}$ with $r\in\Sigma$.
	
	First, consider the case with $|\sigma|,|\sigma'| <k$. We have two possibilities.
	\begin{itemize}
		\item $\hat{t}$ does not contain nominals. The following proof is illustrated in Figure~\ref{fig:forth}. There is a mosaic $p$ with $\hat{\sigma} = \sigma p$, and from $((t, \sigma), (\hat{t}, \sigma p)) \in r^{\Jmc_{1}}$ we obtain $(t, \hat{t}) \in R^{r,1}_{\textnormal{tail}(\sigma),p}$. This means that there exist $d,\hat{d}$ realizing $t$ and $\hat{t}$, respectively, with $(T_{1}(d),T_{2}(d))= \textnormal{tail}(\sigma)$ and $(T_{1}(\hat{d}),T_{2}(\hat{d}))=p$, such that $(d,\hat{d})\in r^{\Imc_{1}}$.
		As $t'\in T_{2}(d)$, there exists $e\in \Delta^{\Imc_{2}}$ with $\Imc_{1}, d \sim_{\mathcal{ALCHO}, \Sigma} \Imc_{2}, e$
		and $e$ realizes $t'$. By the definition of bisimulations, there exists $\hat{e}$ with $(e,\hat{e})\in r^{\Imc_{2}}$ and $\Imc_{1}, \hat{d} \sim_{\mathcal{ALCHO},\Sigma} \Imc_{2}, \hat{e}$. Assume that $\hat{e}$ realizes $\hat{t}'$. Then $\hat{t}' \in T_{2}(\hat{d})$
		and $(\hat{t},\hat{t}')\in R^{r, 2}_{\textnormal{tail}(\sigma'), p}$.
		Now we consider again two possibilities.
		\begin{itemize}
			\item $\hat{t}'$ does not contain nominals. Then from $(t', \hat{t}') \in R^{r, 2}_{\textnormal{tail}(\sigma'), p}$ we obtain $((t',\sigma'),(\hat{t}',\sigma'p)) \in r^{\Jmc_{2}}$. Since $\textnormal{tail}(\sigma p) = \textnormal{tail}(\sigma' p)$, we also obtain $((\hat{t}, \sigma p), (\hat{t}', \sigma' p)) \in S$.
			\item $\hat{t}'$ contains nominals. Then from $(t', \hat{t}') \in R^{r, 2}_{\textnormal{tail}(\sigma'), p}$ we obtain $((t',\sigma'),(\hat{t}',p)) \in r^{\Jmc_{2}}$. Since $\textnormal{tail}(\sigma p) = p$, we get that $((\hat{t}, \sigma p), (\hat{t}', p)) \in S$ as well.
		\end{itemize}
In both cases, we obtain some $(\hat{t}', \hat{\sigma}')$ with $((t', \sigma'), (\hat{t}', \hat{\sigma}')) \in r^{\Jmc_{2}}$ and $((\hat{t}, \hat{\sigma}), (\hat{t}', \hat{\sigma}')) \in S$, as required.
		\item $\hat{t}$ contains nominals. In this case, $\hat{\sigma} = p$ for some mosaic $p$, and from 
		$((t, \sigma), (\hat{t}, p)) \in r^{\Jmc_{1}}$ we obtain $(t, \hat{t}) \in R^{r,1}_{\textnormal{tail}(\sigma),p}$. Now we can reason as above.
	\end{itemize}	
Now consider the case with $| \sigma | = k$ and $|\sigma'|<k$.
As $\text{tail}(\sigma) = \text{tail}(\sigma')$, there exists $\sigma''$ such that $|\sigma''|<k$ and $\text{tail}(\sigma'') = \text{tail}(\sigma)$ and the $r$-successors of any node of the form $(t,\sigma)$ are exactly the $r$-successors of $(t,\sigma'')$, and thus to show [Forth] one can proceed as above.
The same argument applies if $| \sigma | < k$ and $|\sigma'|=k$ and if $| \sigma |= |\sigma'|=k$
and there exists $\sigma''$ with $\text{tail}(\sigma'')=\text{tail}(\sigma')=\text{tail}(\sigma)$
and $|\sigma''|<k$. Finally, if $| \sigma |= |\sigma'|=k$ but there does not exist any 
$\sigma''$ with $\text{tail}(\sigma'')=\text{tail}(\sigma')=\text{tail}(\sigma)$
and $|\sigma''|<k$, then there are no $r$-successors to consider.

	\item{[Back]} Dual to [Forth].
\end{description}

\begin{figure}[th]
\centering
\begin{tikzpicture}
\tikzset{
dot/.style = {draw, fill=black, circle, inner sep=0pt, outer sep=1pt, minimum size=2pt}
}

\draw (-1,0.5) node[label=$\mathcal{I}_{1}$] (I1) {};

\draw (0,0) node[dot, label=south:$d$] (d) {};

\draw (0,1) node[dot, label=north:$\hat{d}$] (dh) {};


\draw[->, >=stealth] (d) -- (dh) node[midway, left] {$r$};

\draw (4,0.5) node[label=$\mathcal{I}_{2}$] (I2) {};

\draw (3,0) node[dot, label=south:$e$] (e) {};

\draw (3,1) node[dot, label=north:$\hat{e}$] (eh) {};


\draw[->, >=stealth] (e) -- (eh) node[midway, right] {$r$};


\path[black, dashed, bend left] (d) edge (e);
\draw (1.5,0.35) node[label=$\sim_{\mathcal{ALCHO},\Sigma}$] () {};

\path[black, dashed, bend left] (dh) edge (eh);
\draw (1.5,1.35) node[label=$\sim_{\mathcal{ALCHO},\Sigma}$] () {};

\path[black] (5,-0.5) edge (5,1.5);

\draw (6,0.5) node[label=$\mathcal{J}_{1}$] (J1) {};

\draw (11,0.5) node[label=$\mathcal{J}_{2}$] (J2) {};

\draw (7,0) node[dot, label={[shift={(-0.15,-0.75)}]:$(t\mathit{,}\,\sigma)$}] (t) {};

\draw (7,1) node[dot,  label={[shift={(-0.15,0.1)}]:$(\hat{t}\mathit{,}\,\hat{\sigma})$}] (th) {};

\draw (10,0) node[dot, label={[shift={(0.15,-0.75)}]:$(t'\mathit{,}\,\sigma')$}] (t') {};

\draw (10,1) node[dot, label={[shift={(0.15,0.1)}]:$(\hat{t}'\mathit{,}\,\hat{\sigma}')$}] (t'h) {};


\draw[->, >=stealth] (t) -- (th) node[midway, left] {$r$};

\draw[->, >=stealth] (t') -- (t'h) node[midway, right] {$r$};


\path[black, dashed, bend left] (t) edge (t');
\draw (8.5,0.35) node[label=$S$] () {};

\path[black, dashed, bend left] (th) edge (t'h);
\draw (8.5,1.35) node[label=$S$] () {};
%

\end{tikzpicture}

\caption{
Proof step to show that $S$ satisfies [Forth], with $\hat{\sigma}=\sigma p$ and $\hat{\sigma}'=\sigma'p$ (if $\hat{t}'$ does not contain nominals) and $\hat{\sigma}'=p$ (if $\hat{t}'$ contains nominals).}\label{fig:forth}

\end{figure}

Observe that the models $\Jmc_{i}$, $i = 1,2$, are at most exponential in the size of $\Omc, C_{1}, C_{2}$. Moreover, we have $(T_{1}(d_{1}),(T_{2}(d_{1}))\in \mathcal{B}$ and so $(\textnormal{tp}_{\Xi}(\Imc_{1},d_{1}),T_{1}(d_{1})) \in C_{1}^{{\Jmc_{1}}}$, 
$(\textnormal{tp}_{\Xi}(\Imc_{2},d_{2}),T_{2}(d_{1}))\in C_{2}^{{\Jmc_{2}}}$, and
\[
((\textnormal{tp}_{\Xi}(\Imc_{1},d_{1}),T_{1}(d_{1})),(\textnormal{tp}_{\Xi}(\Imc_{2},d_{2}),T_{2}(d_{1}))\in S,
\]
as required.

We next consider the case with inverse roles, but without nominals.    
In this case, we let $\mathcal{B}$ be some minimal set of mosaics defined by $\Imc_{1},\Imc_{2}$ containing $(T_{1}(d_{1}),T_{2}(d_{1}))$ and such that for every type $t$ realized in $\Imc_{i}$ there exists $(T_{1},T_{2})\in \mathcal{B}$ with $t\in T_{i}$.
We extend the relations $R_{p,q}^{s,i}$ defined previously to inverse roles $s$ in the obious way and 
select for any mosaic $p=(T_{1},T_{2})$ and any $\exists s.C\in t\in T_{i}$ such that there exists a $\Sigma$-role $r$ with $\Omc\models s \sqsubseteq r$ a mosaic $q=(T_{1}',T_{2}')$ such that $(t,t')\in R_{p,q}^{s,i}$ and $C\in t'$ and denote the resulting set by $\mathcal{S}(p)$.

Form again the set $\mathcal{T}$ of sequences 
\[
\sigma = p_{0}\cdots p_{j}=(T_{1}^{0},T_{2}^{0})\cdots (T_{1}^{j},T_{2}^{j})
\]
with $j\leq k$, $p_{0}\in \mathcal{B}$ and $p_{i+1}\in \mathcal{S}(p_{i})$ for $i < j$.
Let
$\text{tail}(\sigma)= p_{j}$ and
$\text{tail}_{i}(\sigma)= T_{i}^{j}$.
We next define the domain of $\Jmc_{1}$ and $\Jmc_{2}$ as
\[
	\Delta^{\Jmc_{i}}  = \{ (t,\sigma) \mid \sigma\in \mathcal{T}, t\in \text{tail}_{i}(\sigma)\}
\]
We define interpretations $\Jmc_{1},\Jmc_{2}$ in the expected way. 
\begin{itemize}
	\item For any concept name $A$, $(t,\sigma) \in A^{\Jmc_{i}}$ iff $A\in t$;
	\item Let $r$ be a role name. Then we let for $\sigma p\in \mathcal{T}$, 
	\begin{itemize}
		\item $((t,\sigma),(t',\sigma p))\in r^{\Jmc_{i}}$ if $(t,t') \in R_{\text{tail}(\sigma),p}^{r,i}$;
		\item $((t',\sigma p),(t,\sigma)) \in r^{\Jmc_{i}}$
		if $(t,t') \in R_{\text{tail}(\sigma),p}^{r^{-},i}$.
	\end{itemize}
	\item We still have to take care of existential restrictions $\exists r.C$ with $r$ a role 	that does not entail any $\Sigma$-role.
	If $\sigma\in \mathcal{T}$, $\exists r.C\in t\in T_{i}$ with $\text{tail}_{i}(\sigma) = T_{i}$
	and $\Omc\not \models r \sqsubseteq s$ for any $\Sigma$-role $s$, we add $((t,\sigma),(t',p))$ to $r^{\Jmc_{i}}$ (and all $s^{\Jmc_{i}}$ with $\Omc\models r\sqsubseteq s$) for some $p=(T_{1}',T_{2}')\in \mathcal{B}$ and $t'\in T_{i}'$ with $C\in t'$ such that there are $e,e'$ realizing $t,t'$ in $\Imc_{i}$ and $(e,e')\in r^{\Imc_{i}}$. 		
\end{itemize}

The fact that $\Jmc_{i} \models \Omc$, for $i \in \{ 1, 2 \}$, is proved similarly to the case with nominals.
%
One can also prove again by induction on the structure of $C$ that for every 
$(t, \sigma) \in \Delta^{\Jmc_{i}}$ and every $C\in \Xi$ of depth
$\leq k - |\sigma|$,
\begin{equation*}
\label{eq:coninv}
(t,\sigma) \in C^{\Jmc_{i}} \text{ iff } C \in t.
\end{equation*}
Next we observe that
the relation
\[
S= \{ ((t,\sigma),(t',\sigma))\in \Delta^{\Jmc_{1}}\times \Delta^{\Jmc_{2}} \mid \sigma \in \mathcal{T}\}
\]
is an $\mathcal{ALCHI}(\Sigma)$-bisimulation.
Indeed, it can be seen, similar to the case with nominals, that $S$ satisfies [AtomC].
We now give a proof of [Forth]. We provide the proof for role names; 
the proof for inverse roles is similar. 
\begin{description}
	\item{[Forth]} Let $((t, \sigma), (t', \sigma)) \in S$ and $((t, \sigma), (\hat{t}, \hat{\sigma})) \in r^{\Jmc_{1}}$. We distinguish two cases.
	Assume first that there exists a mosaic $p$ with $\hat{\sigma} = \sigma p$. Then 
	$(t,\hat{t}) \in R^{r,1}_{\text{tail}(\sigma),p}$.
	Thus, there exist $d, \hat{d}$ realizing $t, \hat{t}$, respectively, such that $(T_{1}(d), T_{2}(d)) = \text{tail}(\sigma)$, $(T_{1}(\hat{d}), T_{2}(\hat{d})) = p$, and $(d,\hat{d}) \in r^{\Imc_{1}}$.
	Since $((t, \sigma), (t', \sigma)) \in S$, there exists $e$ realizing $t'$
such that $\Imc_{1}, d \sim_{\mathcal{ALCHI},\Sigma} \Imc_{2}, e$.
As $d$ and $e$ are bisimilar,
	 we also have some $\hat{e} \in \Delta^{\Imc_{2}}$ such that $(e,\hat{e}) \in r^{\Imc_{2}}$ and $\Imc_{1}, \hat{d} \sim_{\mathcal{ALHI},\Sigma} \Imc_{2}, \hat{e}$, with $\hat{e}$ realizing some $\hat{t}'$. Hence, $(t', \hat{t}') \in R^{r,2}_{\text{tail}(\sigma),p}$, and it follows that $((t',\sigma), (\hat{t}', \sigma p)) \in r^{\Jmc_{2}}$. Moreover, $((\hat{t}, \sigma p), (\hat{t}', \sigma p)) \in S$.
	
	Assume now that $\sigma=\hat{\sigma}p$ for some mosaic $p$.  Then 
	$(\hat{t},t) \in R^{r^{-},1}_{\text{tail}(\hat{\sigma}),p}$.
	Thus, there exist $\hat{d},d$ realizing $\hat{t},t$, respectively, such that $(T_{1}(\hat{d}), T_{2}(\hat{d})) = \text{tail}(\hat{\sigma})$, $(T_{1}(d), T_{2}(d)) = p$, and $(\hat{d},d) \in (r^{-})^{\Imc_{1}}$.
	Since $((t, \sigma), (t', \sigma)) \in S$, there exists $e$ realizing $t'$
	such that $\Imc_{1}, d \sim_{\mathcal{ALCHI},\Sigma} \Imc_{2}, e$.
	As $d$ and $e$ are bisimilar,
	we also have some $\hat{e} \in \Delta^{\Imc_{2}}$ such that $(\hat{e}, e) \in (r^{-})^{\Imc_{2}}$ and $\Imc_{1}, \hat{d} \sim_{\mathcal{ALHI},\Sigma} \Imc_{2}, \hat{e}$, with $\hat{e}$ realizing some $\hat{t}'$. Hence, $(\hat{t}',t') \in R^{r^{-},2}_{\text{tail}(\hat{\sigma}),p}$, and it follows that $((t',\sigma), (\hat{t}', \hat{\sigma})) \in r^{\Jmc_{2}}$. Moreover, $((\hat{t}, \hat{\sigma}), (\hat{t}', \hat{\sigma})) \in S$.
\end{description}
Observe that again the models $\Jmc_{i}$, $i = 1,2$, are of at most exponential size in the size of $\Omc, C_{1}, C_{2}$. We also have $(T_{1}(d_{1}),T_{2}(d_{1}))\in \mathcal{B}$ and so $(\textnormal{tp}_{\Xi}(\Imc_{1},d_{1}),T_{1}(d_{1})) \in C_{1}^{{\Jmc_{1}}}$, 
$(\textnormal{tp}_{\Xi}(\Imc_{2},d_{2}),T_{2}(d_{1}))\in C_{2}^{{\Jmc_{2}}}$, and
	\[
	((\textnormal{tp}_{\Xi}(\Imc_{1},d_{1}),T_{1}(d_{1})),(\textnormal{tp}_{\Xi}(\Imc_{2},d_{2}),T_{2}(d_{1}))\in S,
      \]
	as required.

\section{Lower Bound Proofs without Ontology}\label{sec:lower-without-ontology}

In this section, we first show (the hardness part of) Points~1 and~2 of
Theorem~\ref{thm:main2} by a reduction of the case with ontologies,
and then show the nondeterministic exponential time lower bounds for
Points~3 and~4 of that Theorem. Points~1 and~2 of Theorem~\ref{thm:main2} are a direct consequence of the following lemma.
\begin{lemma}\label{lem:red33}
  Let $\Lmc\in \Dln$ admit the universal role or both inverse roles and nominals. Then the following holds: 
  \begin{enumerate}

    \item if $\Lmc$ admits RIs, then projective $\Lmc$-definition existence can be reduced in
      polynomial time to RI-ontology projective $\Lmc$-definition
      existence;

    \item if $\Lmc$ does not admit RIs, then projective $\Lmc$-definition existence can be reduced in
      polynomial time to ontology-free projective $\Lmc$-definition
      existence.

  \end{enumerate}   	
\end{lemma}
\begin{proof}
	Assume $\Omc$, $C$, $C_{0}$, and $\Sigma$ are given. 
	We may assume that $\Omc$ takes the form $\{\top \sqsubseteq D\}\cup \Omc'$ with $\Omc'$ a set of RIs.
	
	Assume first that $\Lmc$ admits the universal role. Then one can easily show that there exists 	
	an explicit $\Lmc(\Sigma)$-definition of $C_{0}$ under $\Omc$ and $C$ iff there exists an explicit $\Lmc(\Sigma)$-definition of $C_{0}$ under $\Omc'$ and $C \sqcap \forall u.D$.
	
	Now assume that $\Lmc$ admits inverse roles and nominals. We use the spy-point technique to encode the universal role~\cite{DBLP:conf/csl/ArecesBM99}. Introduce a fresh individual $a$ and a fresh role name $r_{0}$ and define
	$U$ as the conjunction of the concepts 
	\[
	\{a\}, \quad \exists r_{0}.\{a\}, \quad \exists r_{0}.(\{b\} \sqcap \exists r_{0}.\{a\}), \quad \forall r_{0}^{-}.\forall s.\exists r_{0}.\{a\},
      \]
	for all $s\in \{r,r^{-}\}$ with $r\in \text{sig}(\Omc, C,
	C_{0})$ and $b\in \text{sig}(\Omc,C,C_{0})$. Observe that if
	$d\in (U \sqcap \forall r_{0}^{-}.F)^{\Imc}$ for some
	interpretation $\Imc$ and concept $F$, then $e\in F^{\Imc}$
	holds for all elements $e$ in $\Delta^{\Imc}$ that can be reached in $\Imc$ from $d$ or any $b^{\Imc}$ with $b\in \text{sig}(\Omc,C,C_{0})$ along roles in $\text{sig}(\Omc,C,C_{0})$. It follows that for any $\mathcal{L}(\Sigma)$-concept $E$, we have
	\[
	\mathcal{O}\models C \sqsubseteq (C_{0} \leftrightarrow E) \quad \mbox{ iff } \quad \mathcal{O}'\models (C \sqcap U \sqcap \forall r_{0}^{-}.D) \sqsubseteq (C_{0} \leftrightarrow E).
      \]
	Hence there exists 	
	an explicit $\Lmc(\Sigma)$-definition of $C_{0}$ under $\Omc$ and $C$ iff there exists an explicit $\Lmc(\Sigma)$-definition of $C_{0}$ under $\Omc'$ and $C \sqcap U \sqcap \forall r_{0}^{-}.D$. 	
\end{proof}
We show the lower bound for Theorem~\ref{thm:main2}, Points~3 and~4, by proving \NExpTime-hardness for the version of joint consistency formulated in Theorem~\ref{thm:critdef}.
We reduce the exponential torus tiling problem. A \emph{tiling system} is a triple
$P=(T,H,V)$, where $T = \{0,\dots,k\}$ is a finite set of \emph{tile
types} and $H,V \subseteq T \times T$ are the \emph{horizontal}
and \emph{vertical} matching conditions, respectively. 
An \emph{initial condition} for
$P$ takes the form $c = (c_0,\dots,c_{n-1}) \in T^n$. A mapping
$\tau: \{0,\dots,2^{n}-1\} \times \{0,\dots,2^{n}-1\} \to T$ is a
\emph{solution for $P$ and $c$} if 
$\tau(i,0) = c_{i}$ for all $i < n$, and 
for all $i,j < 2^{n}$, the
following conditions hold (where $\oplus_k$ denotes addition modulo $k$):
\begin{itemize}
  \item if $\tau(i,j) = t_1$ and $\tau(i \oplus_{2^{n}} 1,j) =
    t_2$, then $(t_1,t_2) \in H$;

  \item if $\tau(i,j) = t_1$ and $\tau(i,j \oplus_{2^{n}} 1) =
    t_2$, then $(t_1,t_2) \in V$.

    %
\end{itemize}
It is well-known that the problem of deciding whether there is a
solution for given $P$ and $c$ is \NExpTime-hard~\cite[Section 5.2.2]{DL-Textbook}.
For the following constructions, assume a tiling system $P$ and an
initial condition $c$ of length $n$.

\medskip
For the reduction for $\mathcal{ALCO}$, we give concepts
$C,C_0$ and a signature $\Sigma$ such that here exist
$\Imc_{1},d_{1}\sim_{\mathcal{ALCO}, \Sigma} \Imc_{2},d_{2}$ with
$d_{1} \in (C \sqcap C_{0})^{\Imc_{1}}$ and $d_{2}\in (C \sqcap \neg
C_{0})^{\Imc_{2}}$ iff $P$ has a solution given $c$. We start with setting
\[
C_{0} = \exists r^{2n}.\{a\} \sqcap \forall r^{2n}.\{a\}
\]
with $a\not\in\Sigma$ and $r\in \Sigma$. In addition to $r$, $\Sigma$ contains concept names $B_0,\dots,B_{2n-1}$ that serve as bits in the
binary representation of \emph{grid positions} $(i,j)$ with $0\leq i,j\leq 2^{n}-1$, where
bits $B_0,\ldots,B_{n-1}$ represent the horizontal position $i$ and
$B_n,\ldots,B_{2n-1}$ the vertical position $j$, and concept names
$T_{0},\ldots,T_{k}$ representing tile types. We also use the
following concept names that are not in $\Sigma$: another four sets of
concepts names $A_0,\dots,A_{2n-1}$
and $V_0,\dots,V_{2n-1}$ with $V\in \{X,Y,Z\}$ that also serve as 
bits in the binary representation of grid position $(i,j)$ with $0\leq i,j\leq 2^{n}-1$, and
concept names $R_{0},\ldots,R_{2n}$, $M$, $M_{1}$, and $M_{2}$. We now define the concept $C$ 
as a conjunction of several concepts. The first conjunct is  
\[
\neg C_{0} \sqcap \exists r^{2n}.\top \rightarrow R_{0}.
\]
Intuitively, $R_{0}$ generates a binary $r$-tree of depth $2n$ with $R_{i}$ true at level $i$ for $0\leq i\leq 2^{2n}$ and each leaf represents a grid position $(i,j)$ using the concept names $A_i$. To achieve this let $C$ contain the following conjuncts for generating the binary tree:
\begin{align*}
  \bigsqcap_{0\leq i <2n}\forall r^i.&\big(R_{i} \rightarrow (\exists r.(A_{i} \sqcap R_{i+1})
  \sqcap \exists r.(\neg A_{i} \sqcap R_{i+1}))\big)\\
  \bigsqcap_{1\leq i< 2n}\bigsqcap_{0\leq j<i}\forall r^{i}.&\big((A_{j} \rightarrow \forall
  r.A_{j}) \sqcap (\neg {A}_{j} \rightarrow \forall r.\neg
  {A}_{j})\big)
\end{align*}
As usual, $\forall r^i$ abbreviates a sequence of $i$ times
$\forall r$.
%
 
We next express using additional conjuncts of $C$ that any leaf $d$
representing $(i,j)$ using $A_i$ has the following properties~(A) -- (C):
\begin{enumerate}[label=(\Alph*)]

  \item $d$ has an $r$-successor representing $(i,j)$ using
    $B_i$ with a tile type $T_{(i,j)}$ true in it;
    moreover, no $r$-successor of $d$ representing $(i,j)$ satisfies a
    tile type different from $T_{(i,j)}$. This is achieved using the
    marker $M$ which holds in exactly those $r$-successors of $d$ that
    represent $(i,j)$ using $B_i$. The latter
    condition is expressed using the counter
    $X_{i}$ which represents $(i,j)$ on all
    $r$-successors of $d$. In detail, we add the following conjuncts
    to $C$: 
    \begin{align*}
      \forall r^{2n}.&\exists r.M\\
      %
      \forall r^{2n}.&\big(\bigsqcap_{i<2n}(A_{i} \rightarrow
      \forall r.X_{i}) \sqcap (\neg{A}_{i} \rightarrow \forall
      r.\neg{X}_{i})\big)\\
      \forall r^{2n+1}.&\big(M\leftrightarrow
      \bigsqcap_{i<2n}(X_{i}\leftrightarrow B_{i}) \sqcap
      (\neg{X}_{i} \leftrightarrow \neg{B}_{i})\big)\\
      \forall r^{2n}.&\big(\forall r.(M \rightarrow \bigsqcup_{i\leq k}
      T_{i}) \sqcap \bigsqcap_{i\leq k} \exists r.(M \sqcap T_{i})
      \rightarrow \forall r.(M \rightarrow T_{i})\big) \\
      \forall
      r^{2n+1}.&\bigsqcap_{i\not=j}\neg (T_{i} \sqcap T_{j})
    \end{align*}

  \item $d$ has an $r$-successor representing
    $(i\oplus_{2^{n}}1,j)$ using $B_i$ with a
    tile type $T_{(i,j)}^{\text{right}}$ true in it such that
    $(T_{(i,j)},T_{(i,j)}^{\text{right}})\in H$; moreover, no
    $r$-successor of $d$ representing $(i\oplus_{2^{n}}1,j)$
    satisfies a tile type different from
    $T_{(i,j)}^{\text{right}}$. This is achieved in a similar way as (A)
    using the marker $M_{1}$ which holds in exactly those
    $r$-successors of $d$ that represent $(i\oplus_{2^{n}}1,j)$
    using $B_i$. The latter condition is
    expressed using the counter $Y_{i}$ which
    represents $(i\oplus_{2^{n}}1,j)$ on all $r$-successors of
    $d$. The implementation of these conditions is similar to (A)
    and omitted.

  \item $d$ has an $r$-successor representing $(i,
    j\oplus_{2^{n}}1)$ using $B_{i}$ with a tile type
    $T_{(i,j)}^{\text{up}}$ true in it such that
    $(T_{(i,j)},T_{(i, j)}^{\text{up}})\in V$; moreover, no
    $r$-successor of $d$ representing $(i, j\oplus_{2^{n}}1)$
    satisfies a tile type different from $T_{(i,j)}^{\text{up}}$.
    This is achieved in a similar way as (A) using the marker $M_{2}$
    which holds in exactly those $r$-successors of $d$ that
    represent $(i,j\oplus_{2^{n}}1)$ using $B_i$. The latter
    condition is expressed using the counter
    $Z_{i}$ which represents
    $(i,j\oplus_{2^{n}}1)$ on all $r$-successors of $d$. The
    implementation is again similar to (A) and omitted.

\end{enumerate}
  Finally, we ensure that the initial condition holds, that is $T_{(i,0)}=c_{i}$ for $i<n$.  To this end we add the conjuncts \[\forall r^{2n}.(A=(i,0)
    \rightarrow (\forall r.(M \rightarrow c_{i})))\] for $i<n$, where
    $A=(i,0)$ stands for the representation of $(i,0)$ using
    $A_{i}$; for instance, $A=(0,0)$ stands for
    $\bigsqcap_{0\leq i <2n}\neg{A}_{i}$.

This finishes the definition of $C,C_0$ and we verify next that they are as required. 

\smallskip \noindent \textit{Claim}. There exist
$\Imc_{1},d_{1}\sim_{\mathcal{ALCO}, \Sigma} \Imc_{2},d_{2}$ with
$d_{1} \in (C \sqcap C_{0})^{\Imc_{1}}$ and $d_{2}\in (C \sqcap \neg
C_{0})^{\Imc_{2}}$ iff $P$ has a solution given $c$.

\smallskip \noindent \textit{Proof of the Claim.} Observe that if
$\Imc_{1},d_{1}\sim_{\mathcal{ALCO}, \Sigma} \Imc_{2},d_{2}$ with
$d_{1} \in (C \sqcap C_{0})^{\Imc_{1}}$ and $d_{2}\in (C \sqcap \neg
C_{0})^{\Imc_{2}}$, then there are elements $e_{(i,j)}$, $0\leq i,j\leq
2^{n}-1$ such that 
\[ \Imc_{1},a^{\Imc_{1}}\sim_{\mathcal{ALCO},
\Sigma} \Imc_{2},e_{(i,j)}\]
and $e_{(i,j)}$ has (at least) three
$r$-successors satisfying Conditions~(A) to (C) and the initial condition.  By
$\Sigma$-bisimilarity and since $r\in \Sigma$, all $e_{(i,j)}$ have
$r$-successors satisfying the same concept names in $\Sigma$. Hence,
since the concept names $B_{i}$ and  $T_{i}$ are in $\Sigma$, for
every grid position $(i,j)$ every $e_{(i',j')}$ has an $r$-successor
representing $(i,j)$ using $B_{i}$ and all
$r$-successors representing $(i,j)$ using $B_{i}$
satisfy the same tile type $T_{(i,j)}$. Moreover,
$T_{(i\oplus_{2^{n}}1,j)}= T_{(i,j)}^{\text{right}}$ and
$T_{(i,j\oplus_{2^{n}}1)}= T_{(i,j)}^{\text{up}}$.  It follows that
the mapping $\tau$ defined by setting $\tau(i,j)=T_{(i,j)}$ is a
solution of $P$ given $c$.

Conversely, assume that $P$ and $c$ have a solution $\tau$. The
definition of an interpretation $\Imc$ with elements $d_{1}$ and
$d_{2}$ such that $\Imc,d_{1}\sim_{\mathcal{ALCO}, \Sigma} \Imc,d_{2}$
with $d_{1} \in (C \sqcap C_{0})^{\Imc}$ and $d_{2}\in (C \sqcap \neg
C_{0})^{\Imc}$ is rather straightforward.  An abstract version is
depicted in Figure~\ref{fig:lower}. We omit the counters, and note
that $a^\Imc$ and \emph{all} elements at level $R_{2n}$ have, for all
$0\leq i,j<2^n$, an $r$-successor representing (using concept names
$B_i$) grid position $(i,j)$ which satisfies the concept name $T_{\tau(i,j)}$. We show only the
three special successors from Conditions~(A)--(C). This finishes the
proof of the Claim and thus the reduction for \ALCO.
 
\begin{figure}[th]
\centering
\begin{tikzpicture}
\tikzset{
dot/.style = {draw, fill=black, circle, inner sep=0pt, outer sep=1pt, minimum size=2pt}
}


\draw (-0.4,0) node[label=center:$d_{1}$] (I1) {};


\draw (0,0) node[dot, label=north:$C \sqcap C_{0}$] (F2n) {};

\draw (0,-1) node[dot] (1) {};

\draw (0,-1.45) node {$\vdots$};

\draw (0,-2) node[dot] (i) {};

\draw (0,-3) node[dot, label=south:$a^\Imc$] (a) {};


\draw[->, >=stealth] (F2n) -- (1) node[midway, left] {$r$};
\draw[->, >=stealth] (i) -- (a) node[midway, left] {$r$};


\draw (5.5,0) node[label=center:$d_{2}$] (I2) {};

\draw (5,0) node[dot, label=north:$C \sqcap \lnot C_{0}$, label={[xshift=0cm, yshift=-0.8cm]:$R_{0}$}] (0) {};

\draw (4,-1) node[dot] (1l) {};
\draw (6,-1) node[dot] (1r) {};

\draw (5,-3) node[dot, label=east:$(i\mathit{,}j)$, label=west:$R_{2n}$] (ic) {};


\draw (4,-4) node[dot, label={[align=center]west:$(i\mathit{,}j)$}, label={south:$M$}] (m) {};
\draw (5,-4) node[dot, label={[align=center]south:$M_{1}$\\$(i\oplus_{2^{n}}1\mathit{,}j)$}] (m1) {};
\draw (6,-4) node[dot, label={[align=center]east:$(i\mathit{,}j\oplus_{2^{n}}1)$}, label={south:$M_{2}$}] (m2) {};


\draw[->, >=stealth] (0) -- (1l) node[midway, left] {$r$};
\draw[->, >=stealth] (0) -- (1r) node[midway, right] {$r$};

\draw[->, >=stealth] (ic) -- (m) node[midway, left] {$r$};
\draw[->, >=stealth] (ic) -- (m1) node[midway, right] {$r$};
\draw[->, >=stealth] (ic) -- (m2) node[midway, right] {$r$};

\draw[thick, loosely dotted] (3.5,-1.5) -- (3,-2);
\draw[thick, loosely dotted] (6.5,-1.5) -- (7,-2);
\draw[thick, loosely dotted] (5,-1.5) -- (5,-2.25);

\draw[thick, loosely dotted] (3,-3) -- (3.75,-3);
\draw[thick, loosely dotted] (6.25,-3) -- (7,-3);




\draw[] (-0.5,-2.65) -- (8, -2.65);
\draw[] (-0.5,-3.65) -- (8, -3.65);
\draw[] (-0.5,-2.65) -- (-0.5, -3.65);
\draw[] (8,-2.65) -- (8, -3.65);


\draw (1.75,0) node[label=center:$\sim_{\mathcal{ALCO}, \Sigma}$] (sim0) {};
\draw (1.75,-1) node[label=center:$\sim_{\mathcal{ALCO}, \Sigma}$] (sim1) {};

\draw (1.75,-1.45) node {$\vdots$};

\draw (1.75,-2.15) node[label=center:$\sim_{\mathcal{ALCO}, \Sigma}$] (simi) {};
\draw (1.75,-3.1) node[label=center:$\sim_{\mathcal{ALCO}, \Sigma}$] (simlast) {};

%
%
%

\end{tikzpicture}
\caption{
Interpretation $\Imc$ with elements $d_1\in (C \sqcap C_{0})^\Imc$ and
$d_2\in(C \sqcap \lnot C_{0})^\Imc$ such that
$\Imc,d_1\sim_{\ALCO,\Sigma}\Imc,d_2$.}
\label{fig:lower}
\end{figure}

%
%
%
%

\medskip
We come to the lower bound for $\mathcal{ALCH}$ and $\mathcal{ALCHI}$.
Let 
\[ \Omc=\{r \sqsubseteq r_{1},r\sqsubseteq r_{2},r_{1}\sqsubseteq
v,r_{2}\sqsubseteq v\}, \]
and $\Sigma$ contains $r_{1},r_{2}$ but not $r$ nor $v$. In addition
to $r_{1}$ and $r_{2}$, $\Sigma$ contains exactly the same concept
names as in the $\mathcal{ALCO}$ proof and we also use the same
concept names not in $\Sigma$. We aim to construct concepts $C,C_0$
such that there exist models $\Imc_1,\Imc_2$ of \Omc and $d_{1} \in (C
\sqcap C_0)^{\Imc_{1}}$ and $d_{2}\in (C \sqcap \neg
C_0)^{\Imc_{2}}$ with
$\Imc_{1},d_{1}\sim_{\mathcal{ALCH}, \Sigma} \Imc_{2},d_{2}$  iff $P$
has a solution given $c$.

We set $C_{0}=\exists r^{2n}.\top$. The concept $C$ is again a
conjunction of several concepts; we start in a similar way as for \ALCO with 
\[
\neg C_0 \sqcap \exists v^{2n}.\top \rightarrow R_{0}
\]
The concept name $R_0$ will enforce that 
\begin{itemize}

  \item[$(\ast\ast)$] the end point of any $r_1/r_2$-path of length $2n$
    starting in an element satisfying $R_0$ carries a pair of counter values $(i,j)$
    represented by concept names $A_i$ which describe the path in a
    canonical way.\footnote{Notice the similarity with
      Property~$(\ast)$ from the proof of Lemma~\ref{lem:lowerrolei}.}

\end{itemize}
To achieve this, we include the following conjuncts in $C$:
\begin{align*}
  \bigsqcap_{0\leq i<2n}\forall v^i.& \big(R_{i} \to \forall
  r_{1}.(A_{i}\sqcap R_{i+1}) \sqcap \forall r_{2}.(\neg A_{i}\sqcap
  R_{i+1})\big) \\
  \bigsqcap_{1\leq i< 2n}\bigsqcap_{0\leq j<i}\forall
  v^{i}.&\big((A_{j} \rightarrow \forall v.A_{j}) \sqcap (\neg {A}_{j}
  \rightarrow \forall v.\neg {A}_{j})\big)
\end{align*}
Note that we can use the role name $v$ to address all elements
reachable along $r_{1}/r_{2}$-paths of length $i$ via $\forall v^{i}$.
We continue the definition of $C$ in exactly the same way as for
$\mathcal{ALCO}$ except that we use $\forall v^{2n}$ to reach the end
points of the paths mentioned in~$(\ast\ast)$ and $r_{1}$-successors
of the leaves to encode a solution of the tiling problem. One can then
easily prove the following.

\smallskip
\noindent    
\textit{Claim.} There exist $\Imc_{1},d_{1}\sim_{\mathcal{ALCH}, \Sigma} \Imc_{2},d_{2}$ with $\Imc_{1},\Imc_{2}$ models of $\Omc$, $d_{1} \in (C \sqcap \exists r^{2n}.\top)^{\Imc_{1}}$, and
$d_{2}\in (C \sqcap \neg \exists r^{2n}.\top)^{\Imc_{2}}$ iff $P$ has
a solution given $c$.

\smallskip \noindent \textit{Proof of the Claim.} Observe that if
$\Imc_{1},d_{1}\sim_{\mathcal{ALCH}, \Sigma} \Imc_{2},d_{2}$ with
$\Imc_{1},\Imc_{2}$ models of $\Omc$, $d_{1} \in (C \sqcap \exists
r^{2n}.\top)^{\Imc_{1}}$, and $d_{2}\in (C \sqcap \neg \exists
r^{2n}.\top)^{\Imc_{2}}$, then there exists an element $e$ reachable
from $d_{1}$ along an $r$-path of length $2n$ in $\Imc_{1}$. Since
$d_2\in R_0^{\Imc_2}$ and $e$ is reachable via arbitrary $r_1/r_2$-paths
of length $2n$ from $d_1$, Property~$(\ast\ast)$ implies that there
are elements $e_{(i,j)}$, $0\leq i,j\leq 2^{n}-1$, reachable from
$d_{2}$ along a $v$-path of length $2n$ in $\Imc_{2}$ such that
$\Imc_{1},e \sim_{\mathcal{ALCH}, \Sigma} \Imc_{2},e_{(i,j)}$ and
$e_{i,j}$ represents the pair $(i,j)$ using the concept names $A_i$. The
remaining proof is now essentially the same as for $\ALCO$.

The converse direction is rather straightforward and similar to the
proof for \ALCO. The difference is that the binary tree over role
$r$ in the right
side of interpretation
$\Imc$ depicted in Figure~\ref{fig:lower} is now a binary
tree over roles $r_1$ (left successor) and $r_2$ (right successor). This finishes the
proof of the Claim.

\medskip   To prove the claim above for $\mathcal{ALCHI}$, we adapt the
model construction in a similar way as in the case with ontologies (Section~\ref{sec:lower-role-inclusions}).
More precisely, 
for each element $e$ at level $\ell>0$ in the binary tree below $d_2$, add
$(d,e)\in r_1^\Imc$ and $(d,e)\in r_2^\Imc$, where $d$ is the element
in distance $\ell-1$ from $d_1$. One can then verify that $\Imc$ is as
required, that is, $\Imc,d_1\sim_{\mathcal{ALCHI}} \Imc,d_2$, $d_1\in
(C\sqcap C_0)^\Imc$, and $d_2\in (C\sqcap \neg C_0)^\Imc$.


\section{Computation Problem}\label{sec:computation}

In the previous sections, we have presented algorithms for
\emph{deciding} the existence of interpolants and explicit
definitions, but these algorithms (and their correctness proofs) do
not give immediately rise to a way of \emph{computing} interpolants
and explicit definitions in case they exist. Intuitively, this is due
to the fact that compactness is used in the proof of the
model-theoretic characterization of interpolant and explicit
definition existence in terms of joint consistency modulo
bisimulations which was provided in Theorems~\ref{thm:critint}
and~\ref{thm:critdef}, respectively. In this section, we address the
computation problem for logics in $\Dln$ that do not admit nominals,
by showing that we can actually compute interpolants in case they
exist. We use \emph{DAG representation} for the interpolants; recall
that in DAG representation common sub-formulas are stored only
once, and that thus DAG representation is more succinct than formula
representation. Our approach is inspired by 
a recent note on a type elimination based computation
of interpolants in modal logic~\cite{baldernote} which was originally
provided for the guarded fragment~\cite{DBLP:journals/tocl/BenediktCB16}.
 
\begin{theorem} \label{thm:compute} 
  Let $\Lmc\in \Dln$ not admit nominals, \Omc be an $\Lmc$-ontology, $C_1,C_2$ be \Lmc-concepts, and $\Sigma$
  be a signature. Then, if there is an $\Lmc(\Sigma)$-interpolant for
  $C_1\sqsubseteq C_2$ under \Omc, we can compute the DAG
  representation of an $\Lmc(\Sigma)$-interpolant in time
  $2^{2^{p(n)}}$ where $p$ is a polynomial and
  $n=||\Omc||+||C_1||+||C_2||$.
  %
%
%
%
  %
\end{theorem}
Note that this implies that the DAG representation is also of double
exponential size, and that a formula representation of the interpolant
can be computed in triple exponential time.  Moreover, this also
allows us to compute explicit definitions since, given $\Omc$, $C$,
and $\Sigma$, any $\Lmc(\Sigma)$-interpolant for
$C_{\Sigma}\sqsubseteq C$ under $\Omc\cup \Omc_{\Sigma}$ is an
explicit $\mathcal{L}(\Sigma)$-definition of $C$ under $\Omc$, where
$\Omc_{\Sigma}$ and $C_{\Sigma}$ are obtained from $\Omc$ and $C$ by
replacing all symbols not in $\Sigma$ by fresh symbols. We conjecture
that the triple exponential upper bound on formula size is actually
tight, 
 given the discussion on the explicit definitions that arise in
the hardness proofs in Sections~\ref{sec:shapeofinterpolant}
and~\ref{sec:lower-role-inclusions}. 

Let $\Lmc$, \Omc, $C_1,C_2$, and $\Sigma$ be as in
Theorem~\ref{thm:compute}. By Theorem~\ref{thm:critint}, the existence
of an $\Lmc(\Sigma)$-interpolant for $C_1\sqsubseteq C_2$ under \Omc
is equivalent to joint consistency of $C_1$ and $\neg C_2$ under \Omc
modulo $\Lmc(\Sigma)$-bisimulations. Recall that we have provided
before Lemma~\ref{lem:mosaic-elim-wo-nominals} in
Section~\ref{sec:upperbound} a mosaic elimination procedure for
deciding the latter. In fact, the computation of the
$\Lmc(\Sigma)$-interpolant relies on a finer analysis of that
procedure. We need one more notion to formalize this analyis.

Let $T$ be a set of $\Xi$-types. 
Let $\Imc$ be an
interpretation and, for each $t\in T$, let $d_t$ be a domain element of
$\Imc$. 
We say that $\Imc$ and the elements $d_t$, $t\in T$ \emph{jointly realize
$T$ modulo $\Lmc(\Sigma)$-bisimulations} if, for all $t,t'\in T$,
we have that
$\mn{tp}_{\Xi}(\Imc,d_t)= t$ and $\Imc,d_t\sim_{\Lmc,\Sigma}
\Imc,d_{t'}$.  We 
call
$T$ \emph{jointly realizable under \Omc modulo
$\Lmc(\Sigma)$-bisimulations} if there is a model $\Imc$ of $\Omc$ and 
elements $d_t$ for each $t\in T$ 
that jointly realize $T$ modulo
$\Lmc(\Sigma)$-bisimulations. In contrast to the notion of
joint consistency, we require here a single model $\Imc$ of \Omc.
In what follows, let $\mn{Real}$ denote the set of all sets of types $T$
which are jointly realizable under \Omc modulo
$\Lmc(\Sigma)$-bisimulations. We can effectively determine
$\mn{Real}$ since joint realizability of a set $T$ can be decided in double
exponential time, similar to joint consistency---we refrain from
giving details. 

In the (proof of the) following lemma we show how to compute a concept
differentiating between $T_1$ and $T_2$ when $(T_1,T_2)$ is eliminated
for $T_1,T_2\in \mn{Real}$. In Lemma~\ref{lem:compute-interpolant}
below, we show how to assemble these differentiating concepts to
an interpolant (in case it exists).
\begin{lemma} \label{lem:construction}
  Let $T_1,T_2\in \mn{Real}$. 
  If $(T_1,T_2)$ is eliminated in the mosaic elimination procedure,
  then we can compute an $\Lmc(\Sigma)$-concept $I_{T_1,T_2}$ such
  that
  \begin{enumerate}
    \item for all models $\Imc$ of $\Omc$ and elements $d_t$, for each $t\in
      T_1$, 
      that jointly realize $T_1$ modulo $\Lmc(\Sigma)$-bisimulations,
      $d_t\in I_{T_1,T_2}^\Imc$ for some (equivalently: all) $t\in
      T_1$;

    \item for all models $\Imc$ of 
      $\Omc$ and elements $d_t$, for each $t\in T_2$, that jointly realize
      $T_2$ modulo $\Lmc(\Sigma)$-bisimulations, 
      $d_t\notin I_{T_1,T_2}^\Imc$ for some (equivalently: all) $t\in T_2$.

  \end{enumerate}
  Moreover, a DAG representation of $I_{T_1,T_2}$ can be computed in
  time $2^{2^{p(n)}}$ for some polynomial $p$ and
  $n=||\Omc||+||C_1||+||C_2||$.
\end{lemma}

\begin{proof}
  We compute the $I_{T_{1},T_{2}}$ inductively in the order in which
  the $(T_1,T_2)$ got eliminated in the elimination procedure. We
  distinguish cases why $(T_1,T_2)$ got eliminated. 

  Suppose first that $(T_1,T_2)$ was eliminated because of
  (failing) $\Sigma$-concept name coherence. Since $T_1,T_2$ are both
  jointly realizable, there are the following two cases.
  \begin{enumerate}[label=(\alph*)]

    \item There is a concept name $A\in \Sigma$ such that $A\in t$ for
      all $t\in T_1$, but $A\notin t$, for all $t\in T_2$. Then
      $I_{T_1,T_2} = A.$

    \item There is a concept name $A\in \Sigma$ such that $A\in t$ for
      all $t\in T_2$, but $A\notin t$, for all $t\in T_1$. Then
      $I_{T_1,T_2} = \neg A$.

  \end{enumerate}
  Clearly, in both cases, $I_{T_1,T_2}$ satisfies Points~(1)
  and~(2) of Lemma~\ref{lem:construction}.

  Now, suppose that $(T_1,T_2)$ was eliminated due to
  (failing) existential saturation from $\Smc_i$ during the
  elimination procedure. Since $T_1,T_2$ are both jointly realizable under \Omc, there are the following two cases. 
  \begin{enumerate}[label=(\alph*)]

    \item There exist $t\in T_1$, $\exists r.C\in t$, and a
      $\Sigma$-role $s$ with $\Omc\models r\sqsubseteq s$, such
      that there is no $(T_1',T_2')\in \Smc_i$ such that \textit{(i)} 
      $(T_1,T_2)\rightsquigarrow_s (T_1',T_2')$ and \textit{(ii)} there is
      $t'\in T_1'$ with $C\in t'$ and
      $t\rightsquigarrow_{r,\Omc}t'$. Then, take
      \[
	I_{T_1,T_2} = \exists s.(\bigsqcup_{\substack {T_{1}'\in
	\mn{Real}, \\ T_1\rightsquigarrow_{s}
	  T_{1}',t\rightsquigarrow_{r,\Omc}t',C\in t'\in
	  T_{1}'}}\bigsqcap_{\substack{T'_2\in \mn{Real}, \\
	T_2\rightsquigarrow_{s} T_{2}'}}I_{T_1',T_2'})\]

      \item There exist $t\in T_2$, $\exists r.C\in t$, and a 
	$\Sigma$-role $s$ with $\Omc\models r\sqsubseteq s$, such
	that there is no $(T_1',T_2')\in \Smc$ such that \textit{(i)}
	$(T_1,T_2)\rightsquigarrow_s (T_1',T_2')$ and \textit{(ii)} there
	is $t'\in T_2'$ with $C\in t'$ and
	$t\rightsquigarrow_{r,\Omc}t'$. Then, take
	\[
	  I_{T_1,T_2} = \forall s.(\bigsqcup_{\substack
	    {T_{1}'\in\mn{Real},\\ T_1\rightsquigarrow_{s}
	    T_{1}'}} \bigsqcap_{\substack{T'_{2}\in\mn{Real}, \\
	    T_2\rightsquigarrow_{s} T_{2}',t\rightsquigarrow_{r,\Omc}t',C\in t'\in
	  T_{2}'}}I_{T_1',T_2'})\]

    \end{enumerate}
    We show Points~(1) and~(2) of the lemma for Case~(a); Case~(b) is
    dual. So suppose Case~(a) applies and fix $t\in T_1,\exists
    r.C\in t$, and a $\Sigma$-role $s$ witnessing that.

    \smallskip To show Point~(1) of the lemma, let $\Imc$ be a model of
    $\Omc$ and fix $d_{t_1}$ for each $t_1\in T_1$ such that $\Imc$ and
    the $d_{t_1}$ jointly realize $T_1$ modulo
    $\Lmc(\Sigma)$-bisimulations. It suffices to show that
    $d_t\in I_{T_1,T_2}^\Imc$ for the type $t$ that was fixed in the
    application of Case~(a). Since $d_t$ realizes $t$ and $\exists
    r.C\in t$, there is some $e\in
    C^\Imc$ with $(d_t,e)\in r^\Imc$. Since $\Omc\models
    r\sqsubseteq s$, also $(d_t,e)\in s^\Imc$. Since the $d_{t_1},t_1\in
    T_1$ are mutually $\Lmc(\Sigma)$-bisimilar and $s$ is a
    $\Sigma$-role, we find elements $e_{t_1},t_1\in T_1$ such that:
    \begin{itemize}

      \item $e_{t_1},t_1\in T_1$ 
      are mutually $\Lmc(\Sigma)$-bisimilar,

      \item $(d_{t_1},e_{t_1})\in s^\Imc$, for all $t_1\in T_1$,

      \item $e_{t} = e$. 

    \end{itemize}
    Let
    \[T_1' = \{ \mn{tp}_\Xi(\Imc,e_{t_1})\mid t_1\in T_1\},\] and let
    further $T_2'\in\mn{Real}$ be arbitrary with $T_2\rightsquigarrow_s
    T_2'$. By definition of
    $T_1'$, we have $T_1'\in \mn{Real}$ and $T_1\rightsquigarrow_s T_1'$. Thus, $(T_1',T_2')$ has been
    eliminated before $(T_1,T_2)$: otherwise, Case~(a) would not apply to
    the fixed $t,\exists r.C,s$.
    By induction, we can conclude that $e=e_t\in I_{T_1',T_2'}^\Imc$, and hence $d\in
    I_{T_1,T_2}^\Imc$.

    \smallskip To show Point~(2) of the lemma, let $\Imc$ be a model of
    $\Omc$ and fix $d_{t_2}$ for each $t_2\in T_2$ such that $\Imc$ and
    the $d_{t_2}$ jointly realize $T_2$ modulo
    $\Lmc(\Sigma)$-bisimulations. Suppose, to the contrary of what has
    to be shown, that $d_{\widehat
    t}\in I_{T_1,T_2}^\Imc$ for some $\widehat t\in T_2$. Then, there
    is an $e$ with $(d_{\widehat t},e)\in s^{\Imc}$ and a 
    $T_1'\in\mn{Real}$ with $T_1\rightsquigarrow_s T_1'$ and a type $t'_1\in T_1$ with
    $t\rightsquigarrow_{r,\Omc}t_1'$ and $C\in t_1'$ such
    that 
    \begin{itemize}

      \item[$(\ast)$] $e\in I_{T_1',T}^\Imc$ for all $T\in\mn{Real}$ with
	$T_2\rightsquigarrow_s T$. 

    \end{itemize}
    Since $\Imc$ and the elements $d_{t_2},t_2\in T_2$ jointly realize $T_2$ modulo
    $\Lmc(\Sigma)$-bisimulations and $s$ is a $\Sigma$-role, there are
    elements $e_{t_2},t_2\in T_2$ such that:  
    \begin{itemize}

      \item $e_{t_2},t_2\in T_2$ are mutually $\Lmc(\Sigma)$-bisimilar,

      \item $(d_{t_2},e_{t_2})\in s^\Imc$, for all $t_2\in T_2$,

      \item $e_{\widehat t} = e$.

    \end{itemize}
    Let
    \[T_2' = \{ \mn{tp}_\Xi(\Imc,e_{t_2})\mid t_2\in T_2\}.\] 
    By definition of
    $T_2'$, we have $T_2'\in\mn{Real}$ and $T_2\rightsquigarrow_s T_2'$. Thus, $(T_1',T_2')$ has been eliminated before $(T_1,T_2)$:
    otherwise, Case~(a) would not apply to the fixed $t,\exists r.C,s$.
    By induction, we obtain $e=e_{\widehat t}\notin
    I_{T_1',T_2'}^\Imc$, in contradiction to~$(\ast)$.

    \bigskip

For the analysis of the DAG representation, observe that we can use a
single node for every $I_{T_1,T_2}$. Moreover, $I_{T_1,T_2}$ looks as
follows:
\begin{itemize}

  \item If $(T_1,T_2)$ was eliminated due to failing $\Sigma$-concept
    name coherence, $I_{T_1,T_2}$ is a single concept name $A$ or its
    negation $\neg A$.

  \item Otherwise, it is a node labeled with $\exists s$ (resp., $\forall
    s$), which has a single successor labeled with $\bigsqcup$. This
    successor has then at most double exponentially many successor
    nodes, each labeled with $\bigsqcap$ and each having at most
    double exponentially many successor nodes $I_{T_1,T_2}$.

\end{itemize}
Overall, we obtain double exponentially many nodes in the DAG and
the DAG can be constructed in double exponential time (both in
$p(||\Omc||+||C_1||+||C_2||)$).
\end{proof}

\begin{lemma} \label{lem:compute-interpolant}
  Suppose the result $\Smc^*$ of the mosaic elimination procedure does not
  contain a pair $(T_1,T_2)\in\mn{Real}\times\mn{Real}$ such that
  $C_1\in t_1$ and $\neg C_2\in t_2$ for some types $t_1\in T_1$ and
  $t_2\in T_2$. Then, 
  \[C = \bigsqcup_{\substack{T_1\in\mn{Real}: \\\text{there is $t_1\in
  T_1$ with $C_1\in t_1$}}}\bigsqcap_{\substack{T_2\in\mn{Real}: \\
\text{there is $t_2\in T_2$ with $\neg C_2\in t_2$}}} I_{T_1,T_2}\]
  is an $\Lmc(\Sigma)$-interpolant for $C_1\sqsubseteq C_2$ under
  $\Omc$. Moreover, a DAG representation of $C$ can be computed in
  time $2^{2^{p(n)}}$, for some polynomial $p$ and
  $n=||\Omc||+||C_1||+||C_2||$.
\end{lemma}

\begin{proof}
  We have to show that $\Omc\models C_1\sqsubseteq C$ and 
  $\Omc\models C\sqsubseteq C_2$.

  For $\Omc\models C_1\sqsubseteq C$, let $\Imc$ be a model of $\Omc$
  and suppose $d\in C_1^\Imc$. Let $T_1 = \{\mn{tp}_{\Xi}(\Imc,d)\}$
  consist of the single type of $d$. Clearly, $T_1\in\mn{Real}$. Let
  $T_2\in\mn{Real}$ be arbitrary such that $\neg C_2\in t$, for some
  $t\in T_2$. By assumption of Lemma~\ref{lem:compute-interpolant},
  $(T_1,T_2)$ got eliminated in the elimination procedure. Point~(1) of
  Lemma~\ref{lem:construction} implies $d\in I_{T_1,T_2}^\Imc$. Hence,
  $d\in C^\Imc$.

  For $\Omc\models C\sqsubseteq C_2$, let \Imc be a model of $\Omc$
  and let $d\in (\neg C_2)^\Imc$. Now, let $T_1\in\mn{Real}$ be
  arbitrary such that $C_1\in t$ for some $t\in T_1$, and set $T_2 =
  \{ \mn{tp}_{\Xi}(\Imc,d) \}$.  Clearly, $T_2\in \mn{Real}$. By
  assumption of Lemma~\ref{lem:compute-interpolant}, $(T_1,T_2)$ got
  eliminated in the elimination procedure. Point~(2) of
  Lemma~\ref{lem:construction} implies $d\notin I_{T_1,T_2}^\Imc$.
  Hence, $d\notin C^\Imc$.

  \smallskip
  For the analysis of the DAG representation of $C$, it suffices to
  recall that the DAG representations of the $I_{T_1,T_2}$ provided in
  Lemma~\ref{lem:construction} can be computed in time
  $2^{2^{p(n)}}$, and to observe that $C$ adds only one
  $\bigsqcup$ node and at most double exponentially many
  $\bigsqcap$-nodes. 
%
\end{proof}
To conclude the section, we give some intuition as to why the proof of
Theorem~\ref{thm:compute} cannot be easily adapted to logics from
$\Dln$ that admit nominals. Recall that in any two interpretations
$\Imc_1,\Imc_2$, every nominal $a$ is realized (modulo bisimulation)
in exactly one mosaic. We addressed this by starting the elimination
procedure for all possible choices of mosaics realizing the nominals.
More specifically, in the proof of Lemma~\ref{lem:dexp}, we showed
%
there is an interpolant for
$C_1\sqsubseteq C_2$ under \Omc iff, for all maximal sets \Umc of mosaics
that are good for nominals, the mosaic elimination
procedure started with \Umc leads to an $\Smc^*$ which does not satisfy Condition~2 of
Lemma~\ref{lem:mosaic-elim}, which is akin to the precondition of
Lemma~\ref{lem:compute-interpolant} above. It is, however, unclear how to combine these
different runs of the elimination procedure in proving analogues of Lemmas~\ref{lem:construction}
and~\ref{lem:compute-interpolant}. An alternative approach might be to
derive the interpolants from a suitably constrained proof of
$\Omc\models C\sqsubseteq D$ in an appropriate proof system, see
e.g.~\cite{DBLP:journals/jphil/Rautenberg83}.

\section{Some Consequences for Modal and Hybrid Logics}
\label{sec:modal}
In this section we formulate a few consequences of our results in terms of modal and hybrid logics. We focus on interpolant existence and do not discuss the transfer of results on explicit definition existence as they can be obtained in a similar way. We consider the local consequence relation and formulate results for standard hybrid modal languages without the backward modality but with any combination of nominals,
the $@$-operator, and the universal modal modality. We also briefly discuss the reformulation of description logics with role inclusions into modal logic with inclusion conditions on the accessibility relations. For detailed introductions to (hybrid) modal logics we refer the reader to~\cite{DBLP:journals/jsyml/ArecesBM01,areces200714}.

Let ML$_{@}^{u}$ denote the modal hybrid language constructed using
the rule 
\[ \varphi,\psi \quad := \quad p \mid \top \mid i \mid \neg \varphi
\mid \varphi \wedge \psi \mid \Box \varphi \mid @_{i}\varphi \mid
\Box_{u}\varphi, \] 
where $p$ ranges over a countably infinite set of \emph{propositional
variables}, $i$ ranges over a countably infinite set of
\emph{nominals}, $\Box$ ranges over an infinite set of modal operators
$\Box_{0},\dots$, and $\Box_{u}$ denotes the \emph{universal
modality}.  The fragment of ML$_{@}^{u}$ without the universal
modality is denoted ML$_{@}$, the fragment of ML$_{@}$ without the
operators $@_{i}$ is denoted ML$_{n}$, and the fragment of ML$_{n}$
without nominals is the standard language of polymodal logic and
denoted ML. By ML$_{n}^{u}$ we denote the fragment of ML$_{@}^{u}$
without the operators $@_{i}$ and by ML$^{u}$ the extension of ML with
the universal modality.

The signature $\text{sig}(\varphi)$ of a formula $\varphi$ is the set of propositional variables, nominals, and modal operators (without the universal role) occurring in it.

The language $ML_{@}^{u}$ and its fragments are interpreted in
\emph{Kripke models} $\mathfrak{M}=(W,(R_{i})_{i<\omega},V)$ with $W$
a nonempty set of \emph{worlds}, $R_{i}\subseteq W\times W$
accessibility relations, and $V$ a valuation such that $V(p) \subseteq
W$ for every propositional variable $p$, and $V(i)\subseteq W$ a
singleton for every nominal $i$. Then the truth relation
$\mathfrak{M},w\models \varphi$ between \emph{pointed models}
$\mathfrak{M},w$ with $w\in W$ and formulas $\varphi$ is defined
inductively as follows:
%
\begin{alignat*}{3}
  \mathfrak{M}, w &\models \top, &&\\
  \mathfrak{M}, w &\models p & \qquad \text{iff} \qquad& w \in V(p),\\
  \mathfrak{M}, w &\models i & \qquad \text{iff} \qquad & V(i) = \{ w \},\\
  \mathfrak{M}, w &\models \lnot \psi & \qquad \text{iff} \qquad & \mathfrak{M}, w
  \not \models \psi,\\
  \mathfrak{M}, w &\models \psi \land \chi & \qquad \text{iff} \qquad &\mathfrak{M},
  w  \models \psi \text{ and } \mathfrak{M}, w  \models \chi, \\
  \mathfrak{M}, w &\models \Box_{n} \psi & \qquad \text{iff} \qquad & \mathfrak{M}, v
  \models \psi, \text{ for every }v \in W\text{ such that } \\
  & & & (w,v) \in
  R_{n}, \\
  \mathfrak{M}, w &\models @_{i} \psi & \qquad \text{iff} \qquad & \mathfrak{M}, v
  \models \psi, \text{ for the unique element $v\in V(i)$},\\
  \mathfrak{M}, w &\models \Box_{u} \psi & \qquad \text{iff} \qquad &\mathfrak{M}, v
  \models \psi, \text{ for every } v \in W.
\end{alignat*}
%
%
We set $\mathfrak{M}\models \varphi$ if $\mathfrak{M},w\models \varphi$
for all $w\in W$.  Observe that the $@$-operator can be defined using
the universal modality as $@_{i}\varphi=\Box_{u}(i\rightarrow
\varphi)$ and so ML$_{n}^{u}$ and ML$_{@}^{u}$ have the same
expressive power.
 
There are two natural notions of consequence studied in modal and
hybrid logics, local and global entailment, which also give rise to different
notions of interpolants. We focus here on local entailment and
briefly discuss global entailment at the end of this section.
We say that $\varphi$ \emph{locally entails} $\psi$, in symbols
$\varphi \models_{loc} \psi$, if for all pointed models
$\mathfrak{M},w$, if $\mathfrak{M},w\models \varphi$ then
$\mathfrak{M},w\models \psi$.
We note that deciding $\models_{loc}$ is \PSpace-complete for any of the languages introduced above without the universal modality and \ExpTime-complete for any of the languages introduced above with the universal modality~\cite{areces200714}. 

We formulate the interpolant existence problems for hybrid modal
logics in the expected way. Call a formula $\chi$ an
\emph{interpolant for $\varphi,\psi$} if $\text{sig}(\chi) \subseteq
\text{sig}(\varphi)\cap\text{sig}(\psi)$, $\varphi\models_{loc} \chi$
and $\chi\models_{loc} \psi$.
\begin{definition} Let $\Lmc$ be any of the languages introduced
  above. Then the \emph{interpolant existence problem for $\Lmc$} is
  the problem to decide for any $\varphi,\psi\in \Lmc$ whether there
  exists an interpolant for $\varphi,\psi$ in $\Lmc$.
\end{definition} 
Observe that since ML and ML$^{u}$ enjoy the Craig
interpolation property (if $\varphi\models_{loc}\psi$ then an
interpolant for $\varphi,\psi$ exists~\cite{gabbayInterpolation}), the interpolant
existence problem reduces to checking $\varphi\models_{loc} \psi$ and
is \PSpace-complete for ML and \ExpTime-complete for ML$^{u}$. The
following tight complexity bounds for their extensions with nominals
and the $@$-operator are the main result of this section. 
\begin{theorem}

  \begin{enumerate}
    \item Let $\Lmc\in \{\text{ML}_{n},\text{ML}_{@}\}$. Then the interpolant existence problem for $\Lmc$ is \coNExpTime-complete.

    \item The interpolant existence problem for $\text{ML}_{n}^{u}$ is
      \TwoExpTime-complete.  \end{enumerate}		These results
  also hold if one considers the language with a single modal operator
  only.  \end{theorem}
\begin{proof}
  (1) Let $\cdot^{m}$ be the obvious bijection between
  $\mathcal{ALCO}$-concepts and  ML$_{n}$-formulas and denote by
  $\cdot^{d}$ its inverse. Then $\models C \sqsubseteq D$ iff
  $C^{m}\models_{loc}D^{m}$ for any $\mathcal{ALCO}$-concepts $C,D$.
  Hence the following conditions are equivalent, for all formulas
  $\varphi,\psi\in \text{ML}_{n}$:
  \begin{itemize}

    \item there exists an interpolant for $\varphi,\psi$ in ML$_n$;

    \item there exists an $\mathcal{ALCO}(\Sigma)$-interpolant for
      $\varphi^{d},\psi^{d}$, where
      $\Sigma=\text{sig}(\varphi^{d})\cap \text{sig}(\psi^{d})$.

  \end{itemize}	   
  The \coNExpTime-completeness for interpolant existence for ML$_{n}$
  now follows from Point~3 of Theorem~\ref{thm:main2}. We now come to
  ML$_{@}$. We did not consider the operator $@$ for DLs as it does
  not play a large role in description logic research.\footnote{An
  exception is the investigation of updates for description logic
  knowledge bases where the expressive power of the $@$-operator
  plays a significant role~\cite{DBLP:journals/ai/LiuLMW11}.} Note,
  however, that $\mathcal{ALCO}$ can be extended to the DL
  $\mathcal{ALCO}_{@}$ with $@$ in a straightforward way by setting
  $@_{a}C:=\forall u.(\{a\}\rightarrow C)$. The expressive power of
  $\mathcal{ALCO}_{@}$-concepts is characterized by
  \emph{$\mathcal{ALCO}_{@}(\Sigma)$-bisimulations}, where an
  $\mathcal{ALCO}(\Sigma)$-bisimulation $S$ between interpretations
  $\Imc$ and $\Jmc$ is an $\mathcal{ALCO}_{@}(\Sigma)$-bisimulation
  if $(a^{\Imc},a^{\Jmc})\in S$ for any $a\in \Sigma$.  Then one can
  prove Lemma~\ref{lem:equivalence} also for $\mathcal{ALCO}_{@}$.
  Next one can prove the characterization
  (Theorem~\ref{thm:critint}) for $\mathcal{ALCO}_{@}$ in exactly
  the same way as for $\mathcal{ALCO}$, and finally one can extend
  the \NExpTime-upper bound proof for joint consistency modulo
  $\mathcal{ALCO}(\Sigma)$-bisimulations to joint consistency modulo
  $\mathcal{ALCO}_{@}(\Sigma)$-bisimulations
  (Lemma~\ref{lem:witness}) by observing that for all nominal
  generated mosaics $(T_{1}(d),T_{2}(d))$ we now have that
  $T_{i}(d)\not=\emptyset$ for $i=1,2$. Hence
  $(a^{\Imc},a^{\Jmc})\in S$ for any $a\in \Sigma$, for the
  bisimulation $S$ constructed in the proof of
  Lemma~\ref{lem:witness}. 

  The lower bound proof for Theorem~\ref{thm:main2}, Point~3, provided
  in Section~\ref{sec:lower-without-ontology} still goes through as it
  does not use any nominal in the shared signature and so using $@$
  does not make any difference. Note, moreover, that
  it uses only a single role name $r$ which corresponds to using a single
  modal operator. 

  (2) can be proved in the same way as (1) by observing that there is
  a bijection $\cdot^{m}$ between $\mathcal{ALCO}^{u}$-concepts and
  ML$_{n}^{u}$-formulas, that $\models C \sqsubseteq D$ iff
  $C^{m}\models_{loc}D^{m}$ for any $\mathcal{ALCO}^{u}$-concepts
  $C,D$, and then applying Point~1 of Theorem~\ref{thm:main2}.
  Note that the lower bound holds
  for a single role, see Lemma~\ref{lem:lower-single-role}, which
  again translates to a single modal operator (and the universal modality).
%
%
\end{proof}

Description logics with RIs correspond to modal logics determined by Kripke models satisfying inclusions $R_{i}\subseteq R_{j}$ between accessibility relations $R_{i}$ and $R_{j}$. For any finite set $I$ of pairs $(i,j)$ let $\mathcal{M}_{I}$ denote the class of Kripke models satisfying $R_{i}\subseteq R_{j}$ for all $(i,j)\in I$. Define the consequence relation $\models_{loc}^{I}$ in the usual way by setting $\varphi\models_{loc}^{I} \psi$ if for all pointed models $\mathfrak{M},w$ with $\mathfrak{M}\in \mathcal{M}_{I}$, if $\mathfrak{M},w\models \varphi$ then $\mathfrak{M},w\models \psi$. We then obtain the following complexity result directly from Points~4 and ~2 of Theorem~\ref{thm:main2}, respectively.
\begin{theorem}
	For all finite $I$, the interpolant existence problem for $\models_{loc}^{I}$ 
	in ML is in \coNExpTime. There exists a finite $I$ such that the interpolant existence problem for $\models_{loc}^{I}$	in ML is \coNExpTime-hard.
	
	For all finite $I$, the interpolant existence problem for $\models_{loc}^{I}$ 
	in ML$^{u}$ is in \TwoExpTime. There exists a finite $I$ such that the interpolant existence problem for $\models_{loc}^{I}$	in ML$^{u}$ is \TwoExpTime-hard.
\end{theorem}
We close this section with a brief discussion of interpolant existence for the global consequence relation.
We say that $\varphi$ \emph{globally entails} $\psi$, in symbols
$\varphi \models_{glo} \psi$, if for all models
$\mathfrak{M}$ from $\mathfrak{M}\models \varphi$ it follows that
$\mathfrak{M}\models \psi$. Call a formula $\chi$ a 
\emph{global interpolant for $\varphi,\psi$} if $\text{sig}(\chi) \subseteq
\text{sig}(\varphi)\cap\text{sig}(\psi)$, $\varphi\models_{glo} \chi$
and $\chi\models_{glo} \psi$. The \emph{global interpolant existence problem for $\Lmc$} is the problem to decide for any $\varphi,\psi\in \Lmc$ whether there
exists a global interpolant for $\varphi,\psi$ in $\Lmc$.
It is straightforward to show that global interpolant existence corresponds to CI-interpolant existence in DLs in the same way as interpolant existence for the local consequence relation corresponds to ontology-free interpolant existence in DLs.
We therefore obtain \TwoExpTime-completeness of global interpolant existence for the language ML$_n^{u}$ from Theorem~\ref{thm:ciinterpolants}. We conjecture that the same result holds for global interpolant existence for ML$_{n}$ and ML$_{@}$ but leave the proofs for future work.  

\section{Conclusion}\label{sec:conclusion}
We have investigated the problem of deciding the existence of
interpolants and explicit definitions for description and modal logics
with nominals and role inclusions, and we also presented an algorithm
computing them for logics with role inclusions. There are many
challenging problems left for future work, for instance, an algorithm computing interpolants for logics with nominals and the design and implementation of practical algorithms that could be applied in supervised concept learning and referring expression generation. From a theoretical viewpoint it would be of interest to gain a better understanding of when the existence of interpolants is computationally harder than entailment, for logics that do not enjoy the CIP. Logics to consider include more expressive DLs with nominals such as those also admitting qualified number restrictions and/or transitive roles and extensions of the two-variable fragment of FO with counting and/or further constraints on relations~\cite{DBLP:journals/siglog/KieronskiPT18}. Another class of interest are decidable fragments of first-order modal logics and products of modal logics which both often do not enjoy the CIP~\cite{DBLP:journals/jsyml/Fine79,DBLP:journals/ndjfl/MarxA98}.
Here it would be of interest to consider logics such as the one-variable or monodic fragments of K and S4 for which the complexity of interpolant existence was left open~\cite{DBLP:journals/corr/abs-2303-04598}. Finally, is it possible to prove general transfer results (for example, for families of normal modal logics) stating that decidable entailment implies decidability of interpolant existence? 

\section*{Acknowledgements}
Frank Wolter was supported by EPSRC grant EP/S032207/1.
Ana Ozaki was supported by NFR grant 316022. We thank Agi Kurucz, Michael Zakharyaschev, and two referees for helpful comments.


\appendix

\section{Proofs for Section~\ref{sec:background}}
For the proofs of Theorems~\ref{thm:cipandbdp} and~\ref{thm:main3}, we require a few definitions and observations.
Given an interpretation $\Imc = (\Delta^{\Imc}, \cdot^{\Imc})$ and non-empty set $V\subseteq \Delta^{\Imc}$, we define the \emph{restriction of $\Imc$ to $V$} 
as $\Imc_{|V} = (\Delta^{\Imc_{|V}}, \cdot^{\Imc_{|V}})$, where $\Delta^{\Imc_{|V}} = \Delta^{\Imc} \cap V$,
$B^{\Imc_{|V}} = B^{\Imc} \cap V$, for every concept name $B$,
$r^{\Imc_{|V}} = r^{\Imc} \cap (V \times V)$, for every role name $r$, and
$a^{\Imc_{|V}} = a^{\Imc}$ if $a^{\Imc}\in V$, for every individual name $a$. Note that $\Imc_{|V}$ does not interpret individual names that are not interpreted in $V$.
Hence, in our constructions of restrictions, we always make sure that all relevant indviduals are interpreted in $V$. The \emph{relativization} $C_{|A}$ of a concept $C$ to a concept name $A$ describes, in any interpretation $\Imc$, the extension of $C$ in the restriction of $\Imc$ to $A^{\Imc}$. In detail, define $C_{|A}$ inductively by setting $\top_{|A}= A$, $B_{|A}=B \sqcap A$, $(\neg C)_{|A}=A\sqcap \neg C_{|A}$, $(C\sqcap D)_{|A} = C_{|A} \sqcap D_{|A}$,
$\{a\}_{|A}= \{a\} \sqcap A$, and $(\exists r.C)_{|A}= A \sqcap \exists r.C_{|A}$. Then the following observations can be shown by induction over the construction of $C$.
For any interpretation $\Imc$ with $A^{\Imc}\not=\emptyset$ and $A\not\in \text{sig}(C)$ such that all $a\in \text{sig}(C)$ are interpreted in $A^{\Imc}$, $d\in C_{|A}^{\Imc}$ iff $d\in C^{\Imc_{|A^{\Imc}}}$,
for all $d\in \Delta^{\Imc}$. Given an ontology $\Omc$, we set
$\Omc_{|A} = \{C_{|A}\sqsubseteq D_{|A}\mid C \sqsubseteq D\in \Omc\}$. Then $\Imc\models \Omc_{|A}$ iff $\Imc_{|A^{\Imc}}\models \Omc$ whenever all $a\in \text{sig}(\Omc)$ are interpreted in $A^{\Imc}$ and $A^{\Imc}\not=\emptyset$.

We next introduce a generalization of cartesian products called bisimulation products. Consider pointed models $\Imc_{1},d_{1}$
and $\Imc_{2},d_{2}$ with $\Imc_{1},d_{1}\sim_{\mathcal{ALCO},\Sigma} \Imc_{2},d_{2}$. Take a
bisimulation $S$ witnessing this. Then the \emph{bisimulation product $\Imc$ induced by $S$} is defined as follows: the
domain $\Delta^{\Imc}$ of $\Imc$ is the set of all pairs
$(e_{1},e_{2})\in S$. The concept and role names in $\Sigma$ are
interpreted as in cartesian products: $(e_{1},e_{2})\in B^{\Imc}$ iff $(e_{1},e_{2})\in S$ and $e_{i}\in B^{\Imc_{i}}$ for $i=1,2$, and $((e_{1},e_{2}),(e_{1}',e_{2}'))\in r^{\Imc}$ iff $(e_{1},e_{2}),(e_{1}',e_{2}')\in S$ and $(e_{i},e_{i}')\in r^{\Imc_{i}}$
for $i=1,2$. A nominal $a$ in $\Sigma$
is interpreted as $(a^{\Imc_{1}},a^{\Imc_{2}})$ if $a$ is in the
domain (equivalently, range) of $S$. Note that we have projection
functions $f_{i}: S \rightarrow \Delta^{\Imc_{i}}$ with $f_{i}(e_{1},e_{2})=e_{i}$, for $i=1,2$. We denote by $f_{i}(S)$ the image of $S$ under $f_{i}$ in $\Delta^{\Imc_{i}}$ and set $f_{i}^{-1}(V):=\{ (e_{1},e_{2})\in S \mid e_{i}\in V\}$, for any $V\subseteq \Delta^{\Imc_{i}}$.
The following lemma is straightforward.
\begin{lemma}\label{lem:bisimprod}
	Let $\Lmc\in \{\mathcal{ALCO},\mathcal{ALCIO}\}$. If $S$ is an $\Lmc(\Sigma)$-bisimulation, then $f_{i}$ is an $\Lmc(\Sigma)$-bisimulation 
	between $\mathcal{I}$ and $\mathcal{I}_{i}$, for $i=1,2$.
\end{lemma} 
A subset $V$ of $\Delta^{\Imc}$ is called \emph{closed in $\Imc$} if $e'\in V$ 
whenever $(e,e')\in r^{\Imc}$ and $e\in V$, where $r$ is a role name.
$V$ is called \emph{fully closed in $\Imc$} if it is closed 
in $\Imc$ and, moreover, $e\in V$ whenever $(e,e')\in r^{\Imc}$ and $e'\in V$, where $r$ is a role name. Observe that $\mathcal{ALCO}$-concepts are preserved under relativization to closed subsets in the sense that $d\in C^{\Imc}$ iff $d\in C^{\Imc_{|V}}$ for any closed subset $V$ of $\Delta^{\Imc}$, $d\in V$, and $\mathcal{ALCO}$-concept $C$.
It is easy to see that $\mathcal{ALCIO}$-concepts are not always preserved under relativization to closed subsets. They are, however, preserved under relativization to fully closed subsets. Concepts using the universal role are not always preserved under relativization to fully closed subsets.

For a pointed interpretation $\Imc,d$ denote by $\Delta^{\Imc}_{\downarrow d}$
the smallest subset of $\Delta^{\Imc}$ such that $d\in \Delta^{\Imc}_{\downarrow d}$ and $\Delta^{\Imc}_{\downarrow d}$ is closed in $\Imc$. 
The restriction of $\Imc$ to $\Delta^{\Imc}_{\downarrow d}$ is denoted $\Imc_{\downarrow d}$ and called the \emph{interpretation generated by $d$ in~$\Imc$}. It follows from our observation about closed subsets above that $e\in D^{\Imc}$ iff $e\in D^{\Imc_{\downarrow d}}$ holds for all $e\in \Delta^{\Imc}_{\downarrow d}$ and all 
$\mathcal{ALCO}$-concepts $D$.  We are in a position now to prove Theorem~\ref{thm:cipandbdp}.

\bigskip

\thmcipandbdp*
\begin{proof}
	Point~(1)
	follows from the proofs of our complexity lower bounds for explicit definition existence. More specifically, in the proofs of the lower bounds we present
	concepts that are implicitly definable, but not explicitly definable.
	
	For Point~(2), assume first that $\Lmc\in \Dln\setminus \{\mathcal{ALCO},\mathcal{ALCHO}\}$ does not admit nominals.
	Then the BDP follows essentially from
	Theorem~2.5.4 in~\cite{Ten05}, see also~\cite{TenEtAl06}. Theorem~2.5.4 is formulated in terms of modal logic. Hence observe that modal logics are syntactic variants of descriptions logics (see Section~\ref{sec:modal} for details) and that inverse roles, role inclusions, and the universal role can be introduced as first-order definable conditions on frame classes that are
	preserved under generated subframes and bisimulation products of frames.
	
	Next assume that $\Lmc\in \Dln\setminus \{\mathcal{ALCO},\mathcal{ALCHO}\}$  admits nominals and the universal role. Then the @-operator from hybrid logic can be expressed in $\Lmc$ (see again Section~\ref{sec:modal}) and so the BDP follows from Theorem~6.2.4 in \cite{Ten05}, see also~\cite{TenEtAl06}.
	
	It remains to prove that $\mathcal{ALCOI}$ and $\mathcal{ALCHOI}$ enjoy the BDP. This is done using bisimulation products.
	We also use the characterization of the existence of explicit definitions using bisimulations provided in Theorem~\ref{thm:critdef}. Consider an $\mathcal{ALCHOI}$-ontology $\Omc$,
	let $A$ be a concept name, and let $C$ be an $\mathcal{ALCHOI}$-concept. Let $\Sigma=\text{sig}(\Omc,C)\setminus \{A\}$. Assume
	$A$ is not explicitly $\mathcal{ALCHOI}(\Sigma)$-definable under $\Omc$ and $C$. By
	Theorem~\ref{thm:critdef}, we find  pointed models $\Imc_{1},d_{1}$
	and $\Imc_{2},d_{2}$ such that $\Imc_{i}$ is a model of $\Omc$ and $d_{i}\in C^{\Imc_{i}}$
	for $i=1,2$, $d_{1}\in A^{\Imc_{1}}$, $d_{2}\not\in A^{\Imc_{2}}$, and
	$\Imc_{1},d_{1}\sim_{\mathcal{ALCOI},\Sigma} \Imc_{2},d_{2}$. Take a
	bisimulation $S$ witnessing this.  Let $\Imc$ be defined as the bisimulation product induced by $S$. By Lemma~\ref{lem:bisimprod},
	the projection functions $f_{i}: S \rightarrow \Delta^{\Imc_{i}}$ are
	$\mathcal{ALCOI}(\Sigma)$-bisimulations between $\Imc$ and
	$\Imc_{i}$. Moreover, as we have inverse roles, the image of $S$
	under $f_{i}$ is fully closed in $\Imc_{i}$. We
	now define interpretations $\Jmc_{1}$ and $\Jmc_{2}$ as the
	interpretation $\Imc$ except that $A^{\Jmc_{i}}=
	f_{i}^{-1}(A^{\Imc_{i}})$, for $i=1,2$. Then the $f_{i}$ are
	$\mathcal{ALCOI}(\Sigma\cup \{A\})$-bisimulations between
	$\Jmc_{i}$ and $\Imc_{i}$, for $i=1,2$.  Note, however, that the
	$\Jmc_i$ do not necessarily interpret all nominals, as for a
	nominal $\{a\}$, the element $a^{\Imc_1}$ might be in a different
	connected component than $d_1$ (equivalently: $a^{\Imc_2}$ is in a
	different connected component than $d_2$). To address this, let $\Imc'$ be
	the restriction of $\Imc_{1}$ to $\Delta^{\Imc_{1}}\setminus
	f_{i}(S)$.  Then, obtain $\Jmc_{i}'$ by taking the disjoint
	union of $\Jmc_{i}$ and $\Imc'$, $i=1,2$. It follows from the construction, the 
	preservation properties of $\mathcal{ALCOI}$-bisimulations, and the fact that 
	the universal role is not used in $\Omc$ nor $C$ that the following conditions hold:
	\begin{itemize}
		\item[(a)] $\Jmc_{1}'$ and $\Jmc_{2}'$ are models of $\Omc$, $(d_{1},d_{2})\in C^{\Jmc_{1}'}$, and $(d_{1},d_{2})\in C^{\Jmc_{2}'}$;
		\item[(b)] the $\Sigma$-reducts of $\Jmc_{1}'$ and $\Jmc_{2}'$ coincide;
		\item[(c)] $(d_{1},d_{2})\in A^{\Jmc_{1}'}$ but $(d_{1},d_{2})\not\in A^{\Jmc_{2}'}$. 
	\end{itemize}
	But then $A$ is not implicitly $\Sigma$-definable under $\Omc$ and $C$, as required.
	
	Finally, to show Point~(2), we have to argue that $\mathcal{ALCO}$ and $\mathcal{ALCHO}$ do not enjoy the BDP. A counterexample to the BDP is given in~\cite{TenEtAl06}. It
	illustrates nicely the way in which the addition of inverse roles or the universal role to $\mathcal{ALCHO}$ restores the BDP. As this is relevant in the proof of Point~(3) and also in the analysis of non-projective definitions later in Section~\ref{sec:nonprojdefex}, we give the example here.
	Let 
	\[
	\Omc = \{A\sqsubseteq \{a\}, \{b\} \sqcap B
	\sqsubseteq \exists r.(\{a\}\sqcap A), \{b\} \sqcap \neg B
	\sqsubseteq \exists r.(\{a\} \sqcap \neg A)\}
	\]
	and set $\Sigma=\{a,B,b,r\}=\text{sig}(\Omc)\setminus\{A\}$. Then $\Omc \models A\equiv \{a\} \sqcap \exists r^{-}.(B\sqcap \{b\})$ and so $A$ is explicitly $\mathcal{ALCOI}(\Sigma)$-definable under $\Omc$. Also $\Omc \models A\equiv \{a\} \sqcap \exists u.(B\sqcap \{b\})$, and so $A$ is also explicitly $\mathcal{ALCO}^{u}(\Sigma)$-definable under $\Omc$.
	Note, however, that $A$ is not explicitly $\mathcal{ALCO}(\Sigma)$-definable under $\Omc$.
	Indeed, the models $\Imc, \Jmc$ of $\Omc$ depicted in Figure~\ref{fig:defincompl} (where $a^{\Imc} = a$ in $\Imc$ and $a^{\Jmc} = a$ in $\Jmc$, and similarly for $b$) show that $\Imc, a^{\Imc} \sim_{\mathcal{ALCO},\Sigma} \Jmc, a^{\Jmc}$, with $a^{\Imc} \in A^{\Imc}$ and $a^{\Jmc} \not \in A^{\Jmc}$. By Lemma~\ref{lem:equivalence}, we have $\Imc, a^{\Imc} \equiv_{\mathcal{ALCO},\Sigma} \Jmc, a^{\Jmc}$, and hence there cannot be any $\mathcal{ALCO}(\Sigma)$-concept $C$ such that $\Omc \models A \equiv C$.
	
	\begin{figure}[th]
		\centering
		\begin{tikzpicture}
			\tikzset{
				dot/.style = {draw, fill=black, circle, inner sep=0pt, outer sep=1pt, minimum size=2pt}
			}

			\draw (-2,0.5) node[label=$\Imc$] (I) {};
			
			\draw (4,0.5) node[label=$\Jmc$] (J) {};

			\draw (-0.5,-0.5) node[dot, label=west:$b\mathit{,}\, B$] (b1) {};
			
			\draw (-0.5,0.5) node[dot, label=west:$a\mathit{,}\,A$] (d1) {};

			\draw (2.5,-0.5) node[dot, label=east:$b\mathit{,}\,\lnot B$] (b2) {};

			\draw (2.5,0.5) node[dot, label=east:$a\mathit{,}\,\lnot A$] (d2) {};

			
			\draw[->, >=stealth] (b1) -- (d1) node[midway, left] {$r$};
			
			
			\draw[->, >=stealth] (b2) -- (d2) node[midway, right] {$r$};

			\path[black, dashed, bend left] (d1) edge (d2);
			\draw (1,0.15) node[label=$\sim_{\mathcal{ALCO}, \Sigma}$] () {};
			

			%
			%
			
		\end{tikzpicture}
		
		\caption{Models $\Imc, \Jmc$ of $\Omc$ used to show that $A$ is not explicitly $\mathcal{ALCO}(\Sigma)$-definable under $\Omc$.}
		\label{fig:defincompl}
		
	\end{figure}
	
	It follows that non explicit $\mathcal{ALCO}(\Sigma)$-definability of $A$ is caused by the fact that one cannot reach $b$ from $a$ along a path following the role name $r$.
	One can reach $b$, however, using the universal role or inverse roles.

	For Point~(3), it remains to show that $\mathcal{ALCHO}$ enjoys the BDP for ontologies containing RIs only. The BDP for $\mathcal{ALCO}$ with empty ontology follows immediately. 
	The proof is similar to the proof of Point~(2) above for $\mathcal{ALCHIO}$ with fully closed subsets replaced by point generated interpretations.
	
	Assume $\Omc$ contains RIs only,
	$C$ is an $\mathcal{ALCO}$-concept, and $A$ is a concept name. Let $\Sigma= \text{sig}(\Omc,C) \setminus \{A\}$. 
	Assume $A$ is not
	explicitly {$\mathcal{ALCO}(\Sigma)$-}definable
	under $\Omc$ and $C$. By
	Theorem~\ref{thm:critdef}, we find pointed models $\Imc_{1},d_{1}$
	and $\Imc_{2},d_{2}$ such that $\Imc_{i}$ is a model of $\Omc$ and $d_{i}\in C^{\Imc_{i}}$
	for $i=1,2$, $d_{1}\in A^{\Imc_{1}}$, $d_{2}\not\in A^{\Imc_{2}}$, and
	$\Imc_{1},d_{1}\sim_{\mathcal{ALCHO},\Sigma} \Imc_{2},d_{2}$. Take a
	bisimulation $S$ witnessing this. As we do not have the universal role nor inverse roles,
	we may assume that $S$ is a
	$\mathcal{ALCO}(\Sigma)$-bisimulation between the set $\Delta^{\Imc_{1}}_{\downarrow
		d_{1}}$ generated by $d_{1}$ in $\Imc_{1}$ and the
	set $\Delta^{\Imc_{2}}_{\downarrow d_{2}}$ generated by $d_{2}$
	in $\Imc_{2}$. Let $\Imc$ be defined as the bisimulation product induced by $S$. By Lemma~\ref{lem:bisimprod},
	the projection functions $f_{i}: S \rightarrow \Delta^{\Imc_{i}}$ are
	$\mathcal{ALCO}(\Sigma)$-bisimulations between $\Imc$ and
	$\Imc_{i}$. Let $\Jmc_{1}$ and $\Jmc_{2}$ be again defined as the
	interpretation $\Imc$ except that $A^{\Jmc_{i}}=
	f_{i}^{-1}(A^{\Imc_{i}})$, for $i=1,2$. Then the $f_{i}$ are
	$\mathcal{ALCO}(\Sigma\cup \{A\})$-bisimulations between
	$\Jmc_{i}$ and $\Imc_{i}$, for $i=1,2$. As in the proof above, $\Imc$ does not necessarily
	interpret all nominals in $C$. As $\Omc$ is an ontology using RIs only, this problem can be addressed in a straightforward manner. Define
	interpretations $\Jmc_{i}'$ as the disjoint union of $\Jmc_{i}$ 
	and the singleton interpretation $\Imc'$ with domain $\{d\}$ such that $a^{\Jmc_{i}}=d$ for all $a$ not interpreted in $\Imc$ and $B^{\Imc'}=r^{\Imc'}=\emptyset$ for all concept and role names $B$ and $r$. Then $\Jmc_{1}'$ and $\Jmc_{2}'$ satisfy the conditions (a) to (c) above and show that $A$ is not implicitly $\Sigma$-definable under $\Omc$ and $C$.
\end{proof}

\section{Proof of Theorem~\ref{thm:main3}}
\label{sec:nonprojdefex}
We show Theorem~\ref{thm:main3} which states that for $\Lmc\in \{\mathcal{ALCO},\mathcal{ALCHO}\}$ non-projective \Lmc-definition existence of concept names 
is \ExpTime-complete. 
The lower bound follows from the corresponding lower bound for subsumption and any upper bound for $\mathcal{ALCHO}$ trivially implies the same upper bound for $\mathcal{ALCO}$. We therefore focus on the upper bound for $\mathcal{ALCHO}$
and show the following criterion for (the complement of) non-projective explicit definability of concept names.
\begin{restatable}{lemma}{critone}
	\label{lem:restr}
	Let $\Omc$ be an $\mathcal{ALCHO}$-ontology,
	$C$ an $\mathcal{ALCO}$-concept, and $A$ a concept name. Let $\Sigma= \text{sig}(\Omc,C) \setminus \{A\}$. Then $A$ 
	is not explicitly $\mathcal{ALCHO}(\Sigma)$-definable under $\Omc$ and $C$ iff there are pointed interpretations
	$\Imc_{1},d$ and $\Imc_{2},d$ such that 
	\begin{itemize} 
		\item
		$\Imc_{i}$ is a model of $\Omc$ and $d\in C^{\Imc}$, for $i=1,2$; 
		\item the
		$\Sigma$-reducts of ${\Imc_{1}}_{\downarrow d}$ and  ${\Imc_{2}}_{\downarrow d}$ coincide; 
		\item
		$d\in A^{\Imc_{1}}$ and $d\not\in A^{\Imc_{2}}$.
	\end{itemize} 
\end{restatable} 
\begin{proof} \ Clearly, if the conditions of Lemma~\ref{lem:restr}
	hold, then $A$ is not explicitly $\mathcal{ALCHO}(\Sigma)$-definable under
	$\Omc$ and $C$, by Theorem~\ref{thm:critdef}. Conversely, assume  $A$ is not
	explicitly $\mathcal{ALCHO}(\Sigma)$-definable under $\Omc$ an $C$. By
	Theorem~\ref{thm:critdef}, we find pointed models $\Imc_{1},d_{1}$
	and $\Imc_{2},d_{2}$ such that $\Imc_{i}$ is a model of $\Omc$ and $d_{i}\in C^{\Imc_{i}}$
	for
	$i=1,2$, $d_{1}\in A^{\Imc_{1}}$, $d_{2}\not\in A^{\Imc_{2}}$, and
	$\Imc_{1},d_{1}\sim_{\mathcal{ALCHO},\Sigma} \Imc_{2},d_{2}$. Take a
	bisimulation $S$ witnessing this. As we do not admit the universal role nor inverse roles,
	we may assume that $S$ is a
	bisimulation between the set $\Delta^{\Imc_{1}}_{\downarrow
		d_{1}}$ generated by $d_{1}$ in $\Imc_{1}$ and the
	set $\Delta^{\Imc_{2}}_{\downarrow d_{2}}$ generated by $d_{2}$
	in $\Imc_{2}$. Let $\Imc$ be the bisimulation product induced by $S$. By Lemma~\ref{lem:bisimprod},
	the projection functions $f_{i}: S \rightarrow \Delta^{\Imc_{i}}$ are
	$\mathcal{ALCO}(\Sigma)$-bisimulations between $\Imc$ and
	$\Imc_{i}$. However, as in the proof of
	Theorem~\ref{thm:cipandbdp}, $\Imc$ does not necessarily
	interpret all nominals. We address this in the following.
	Let $\Jmc_{i}$ be the restriction of $\Imc_{i}$ to
	$\Delta^{\Imc_{i}}\setminus \Delta^{\Imc_{i}}_{\downarrow d_{i}}$, for $i=1,2$.  We now define
	interpretations $\Jmc_{1}'$ and $\Jmc_{2}'$ as follows:  $\Jmc_{i}'$
	is the disjoint union of $\Imc$ and $\Jmc_{i}$ extended by
	\begin{itemize} \item adding to the interpretation of $A$ all
		elements in
		$f_{i}^{-1}(A^{\Imc_{i}})$; 
		\item adding $(e,(e_{1},e_{2}))$ with $(e_{1},e_{2})\in S$ to
		the interpretation of a role name $r$ if $e\in \Delta^{\Imc_{i}}\setminus
		\Delta^{\Imc_{i}}_{\downarrow d_{i}}$, $e_{i}\in \Delta^{\Imc_{i}}_{\downarrow d_{i}}$, 
		and $(e,e_{i})\in r^{\Imc_{i}}$.
	\end{itemize} 
	Using the condition that $\Omc$ and $C$ do not use the universal role nor inverse roles,
	one can show that the interpretations $\Imc_{i}:=\Jmc_{i}'$ and $d:=(d_{1},d_{2})$ satisfy the conditions of Lemma~\ref{lem:restr}. 
\end{proof}
We now show that the conditions of Lemma~\ref{lem:restr} can be checked in \ExpTime by providing a polynomial time reduction to checking non $\mathcal{ALCHO}$-subsumption. 
Let $\Omc$ be an $\mathcal{ALCHO}$-ontology, $C$ an $\mathcal{ALCO}$-concept, $A$ a concept
name, and $\Sigma=\text{sig}(\Omc,C)\setminus \{A\}$. 
Take 
\begin{itemize}
	\item a concept name $D$ for the domain $\Delta^{{\Imc_{1}}_{\downarrow d}}$ of the interpretation ${\Imc_{1}}_{\downarrow d}$ generated by $d$;
	\item concept names $D_{i}$ for the domain $\Delta^{\Imc_{i}}$ of $\Imc_{i}$, $i=1,2$;
	\item a copy $A'$ of $A$;
	\item and copies $a'$ of the individual names $a$ in $\Sigma$. 
\end{itemize}
We let $\Omc^{c}$ denote the set of CIs in $\Omc$ and $\Omc^{r}$ denote the set of RIs in $\Omc$. Let $\Omc^{cc}$ be the ontology obtained from $\Omc^{c}$ by replacing $A$ by $A'$ and all
nominals $a$ in $\Omc^{c}$ by $a'$. Let $\Omc^{c}_{|D_{1}}$ be the
relativization of $\Omc^{c}$ to $D_{1}$ and let $\Omc^{cc}_{|D_{2}}$ be the
relativization of $\Omc^{cc}$ to $D_{2}$. Now consider the following ontology encoding Points~1 and 2 of Lemma~\ref{lem:restr}:
\begin{eqnarray*} \Omc' & = & \Omc^{r} \cup \Omc^{c}_{|D_{1}} \cup \Omc^{cc}_{|D_{2}} \cup 
	\{D \sqsubseteq \forall r.D \mid r\in \Sigma\} \cup 
	\{D\sqsubseteq D_{1}, D \sqsubseteq D_{2}\} \cup \\ &   & \{\{a\}
	\sqsubseteq D_{1}\mid a\in \Sigma\}\cup \{ \{a'\}
	\sqsubseteq D_{2} \mid a\in \Sigma\} \cup \\ &
	& \{ D \sqcap
	\{a\} \sqsubseteq \{a'\} \mid a\in \Sigma\} \cup \{ D \sqcap
	\{a'\} \sqsubseteq \{a\} \mid a\in \Sigma\} 
\end{eqnarray*}      
Observe that we have to treat the individual names in $\Sigma$ differently from the concept names in $\Sigma$ as we have to ensure that they are interpreted in $D_{1}$ and $D_{2}$ respectively. As their interpretation might be different outside $D$, we have to introduce copies of the individual names and then state that those that are interpreted in $D$ are actually interpreted in the same way. It is now straightforward to show that the conditions of Lemma~\ref{lem:restr} hold iff $\Omc' \not\models A \sqcap D\sqsubseteq A'$.
%

\bibliographystyle{plain}


\end{document}